%% file: main.tex
\documentclass[10pt,journal,compsoc]{IEEEtran}

\ifCLASSOPTIONcompsoc
  \usepackage[nocompress]{cite}
\else
  \usepackage{cite}
\fi

\ifCLASSINFOpdf
\else
\fi

\usepackage{xcolor}
\usepackage{soul}
\usepackage{amsthm}

\def\BibTeX{{\rm B\kern-.05em{\sc i\kern-.025em b}\kern-.08emT\kern-.1667em\lower.7ex\hbox{E}\kern-.125emX}}

\usepackage{nicefrac}
\usepackage{siunitx}
\usepackage{array,framed}
\usepackage{booktabs}
\usepackage{
  color,
  float,
  epsfig,
  wrapfig,
  graphics,
  graphicx,
  subcaption
}
\usepackage{epstopdf}
\usepackage{textcomp,amssymb}
\usepackage{setspace}
\usepackage{latexsym,fancyhdr,url}
\usepackage{enumerate}
\usepackage{mathrsfs}
\usepackage{array}
\usepackage{lipsum}

\usepackage{algorithmicx}
\usepackage{algorithm}

\usepackage[noend]{algpseudocode}
\usepackage{graphics}
\usepackage{xparse}
\usepackage{xspace}
\usepackage{multirow}
\usepackage{csvsimple}
\usepackage{makecell}

\usepackage{changepage}

\usepackage{
  tikz,
  pgfplots,
  pgfplotstable
}
\usepackage{hyperref}

\usetikzlibrary{
  shapes.geometric,
  arrows,
  external,
  pgfplots.groupplots,
  matrix
}

\pgfplotsset{compat=1.9}
\usepackage{mathtools}

\DeclareMathAlphabet{\mathcal}{OMS}{cmsy}{m}{n}

\newtheorem{example}{\textbf{Example}}

\newtheorem{theorem}{\textbf{Theorem}}

\newtheorem{definition}{\textbf{Definition}}

\DeclareGraphicsExtensions{%
    .png,.PNG,%
    .pdf,.PDF,%
    .jpg,.mps,.jpeg,.jbig2,.jb2,.JPG,.JPEG,.JBIG2,.JB2}

\newcommand{\revision}{\color{black}}

\newcommand{\pc}{$p$-cohesion\xspace}
\newcommand{\pcs}{$p$-cohesions\xspace}
\newcommand{\pcalg}{\texttt{$p$-Cohesion}\xspace}
\newcommand{\pcalgOld}{\texttt{$p$-Cohesion$^*$}\xspace}
\newcommand{\elvalg}{\texttt{ELV}\xspace}
\newcommand{\kc}{$k$-clique\xspace}
\newcommand{\kcs}{$k$-cliques\xspace}

\newcommand{\ie}{\textit{i.e.}\xspace}
\newcommand{\etal}{\textit{et al.}\xspace}
\newcommand{\resp}{\textit{resp.}\xspace}
\newcommand{\ddp}{{DDP}\xspace}
\newcommand{\elv}{{ELV}\xspace}
\newcommand{\mre}{\texttt{MRE}\xspace}

\usepackage{enumitem}
\setlist{itemsep=0pt,parsep=0pt}

\begin{document}
\title{Decentralized Privacy Preservation for Critical Connections in Graphs}

\author{Conggai~Li,
        Wei~Ni, \IEEEmembership{Fellow,~IEEE,}
        Ming~Ding, \IEEEmembership{Senior Member, IEEE},
        Youyang~Qu, \IEEEmembership{Member, IEEE}, \\
        Jianjun~Chen,
        David~Smith, \IEEEmembership{Member, IEEE,}
        Wenjie~Zhang, \IEEEmembership{Senior Member, IEEE},
        and 
        Thierry~Rakotoarivelo, \IEEEmembership{Senior Member, IEEE}
\IEEEcompsocitemizethanks{
\IEEEcompsocthanksitem C. Li, W. Ni, M. Ding, Y. Qu, D. Smith, T. Rakotoarivelo are with the Data 61, Commonwealth Scientific and Industrial Research Organisation (CSIRO), Sydney, NSW 2016. E-mail: \{conggai.li, wei.ni, ming.ding, youyang.qu, david.smith, thierry.rakotoarivelo\}@data61.csiro.au.\protect\\
\IEEEcompsocthanksitem J. Chen is with the Australian Artificial Intelligence Institute, University of Technology Sydney, Sydney, NSW 2007. Email: jianjun.chen@student.uts.edu.au.\protect\\
\IEEEcompsocthanksitem W. Zhang is with the School of Computer Science and Engineering, University of New South Wales, Sydney, NSW 2052. E-mail: zhangw@cse.unsw.edu.au.
}
\thanks{Manuscript received; revised.}}

\IEEEtitleabstractindextext{%
\input{abstract}
}

\maketitle

\IEEEdisplaynontitleabstractindextext
\IEEEpeerreviewmaketitle

\input{intro}    
\input{preliminary}
\input{problemStatement}
\input{minimalPC}
\input{DDP}

\input{results}

\input{background}
\input{conclusion}


\ifCLASSOPTIONcaptionsoff
  \newpage
\fi

{\footnotesize
\bibliographystyle{IEEEtran}
\bibliography{bib}
}

\input{bio}


\end{document}

%% file: abstract.tex
\begin{abstract}
Many real-world interconnections among entities can be characterized as graphs.
Collecting local graph information with balanced privacy and data utility has garnered notable interest recently.
This paper delves into the problem of identifying and protecting critical information of entity connections for individual participants in a graph based on cohesive subgraph searches.
This problem has not been addressed in the literature.
To address the problem, 
we propose to extract the critical connections of a queried vertex using a fortress-like cohesive subgraph model known as \pc.
A user's connections within a fortress are obfuscated when being released, to protect critical information about the user.
Novel merit and penalty score functions are designed to measure each participant's critical connections in the minimal \pc,
facilitating effective identification of the connections.
We further propose to preserve the privacy of a vertex enquired by only protecting its critical connections when responding to queries raised by data collectors. 
We prove that, under the decentralized differential privacy (DDP) mechanism, one's response satisfies $(\varepsilon, \delta)$-DDP when its critical connections are protected while the rest remains unperturbed.
The effectiveness of our proposed method is demonstrated through extensive experiments on real-life graph datasets.
\end{abstract}

%% file: intro.tex
\IEEEraisesectionheading{\section{Introduction}\label{sec:intr}}

\IEEEPARstart{G}{raphs} have been widely used to model relationships among entities, 
\textit{e.g.}, social networks and protein interactions.
Connections among entities in graphs might reveal individual participants' private information~\cite{DBLP:journals/pvldb/CormodeSBK09,DBLP:conf/www/BrandtL14}.
A serious privacy concern arises when graphs are probed and queried; \textit{e.g.}, queries, such as subgraph counting~\cite{DBLP:conf/ipps/ChakaravarthyKM16}, may be raised. Replies to these queries can adversely affect the privacy of individuals on the graphs,
{\revision
especially in the face of powerful graph analytic tools, \textit{e.g.}, graph neural networks and graph convolutional networks, that can be misused to recover the structures of graphs and excavate private information based on queries, \textit{e.g.}, subgraph counting queries.}
Several mechanisms have been proposed to enhance the privacy of entities involved in a graph and their associated vertices and edges~\cite{pgd2021}.
For example, specific noises can be added to perturb the replies for such graph analytics~\cite{DBLP:conf/ccs/SunXKYQWY19}, with limited success because of several unresolved issues.

One of the key challenges in privacy-enhanced graph analytics is correctly identifying critical connections of a vertex that are vulnerable to privacy leakage, and effectively obfuscating such sensitive information.
In many cases, 
a cohesive subgraph regarding a vertex captures noteworthy transactional interactions among the vertex and its peers that are particularly important to it.
A variety of cohesive subgraph models has been developed in the literature to find core connections from a whole graph,
such as $k$-core~\cite{seidman1983network,DBLP:conf/icde/YangWQZCL19},
$k$-truss~\cite{DBLP:conf/dasfaa/ZhangYZQLZ18,DBLP:conf/sigmod/HuangCQTY14}, clique~\cite{DBLP:journals/pvldb/FangCLH16,DBLP:journals/vldb/YuanQLCZ16},
and $p$-cohesion~\cite{morris2000contagion,watts2002simple,pastor2015epidemic,DBLP:journals/kais/LiZZQZL21}.
Among these, \pc provides a fortress-like cohesive subgraph, which has been shown effective in applications such as modelling the epidemic's spread into the finite user group~\cite{zanette2001critical,pastor2015epidemic,sun2021transmission}.
Given a graph $G$ and a critical number $p \in (0,1)$, the \pc refers to a connected subgraph, in which every vertex has at least a fraction $p$ of its neighbors within the subgraph~\cite{morris2000contagion}.
These connections, identified by the critical number $p$, are called critical connections and can disclose sensitive information about a user, such as infection status~\cite{pastor2015epidemic}, and therefore, should not be released without proper obfuscation.
Users outside the fortress contain less sensitive information and can be disclosed without compromising the users within the fortress. 
The minimal \pc of a vertex, an elementary unit of a \pc, is not only cohesive but dense as well, 
containing the most critical information/relationship regarding the vertex in the graph.


However, there are no existing solutions specifically designed to protect the privacy of critical connections.
In~\cite{DBLP:conf/ccs/SunXKYQWY19},
{\underline{E}xtended \underline{L}ocal \underline{V}iew} (\elv) was proposed as an elementary unit for graph data protection, 
which only focuses on the proximity of a vertex enquired and does not capture the critical connections of the vertex. 
This technique cannot be directly applied to protect the privacy of the critical connection. Rigorous qualification and validation are required. 

In this paper, we propose to protect the privacy of a vertex by giving priority to protecting its critical connections. 
The critical connections of a vertex collect the edges forming a dense subgraph centered around the vertex, and can capture essential interactions between the vertex and others, potentially containing sensitive information about the vertex.
We propose to identify the critical connections of a vertex as a dense minimal \pc that includes the vertex enquired.
To achieve this, 
we propose new merit and penalty score functions to identify the minimal \pc, \ie, critical connections, for each vertex. 
The functions find \pcs with higher densities and stronger cohesiveness.
To protect the privacy of the identified critical connections, 
we formulate a differentially private query release problem ($k$-clique counting) under the \underline{D}ecentralized \underline{D}ifferential \underline{P}rivacy (DDP) mechanism, 
where a two-phase framework~\cite{DBLP:conf/ccs/SunXKYQWY19} is applied. 
We qualify the use of the $(\varepsilon, \delta)$-DDP for preserving the privacy of targeted vertices upon query releases on the minimal \pcs.

\begin{figure}[t]
\begin{center}
    \includegraphics[width=0.7\columnwidth]{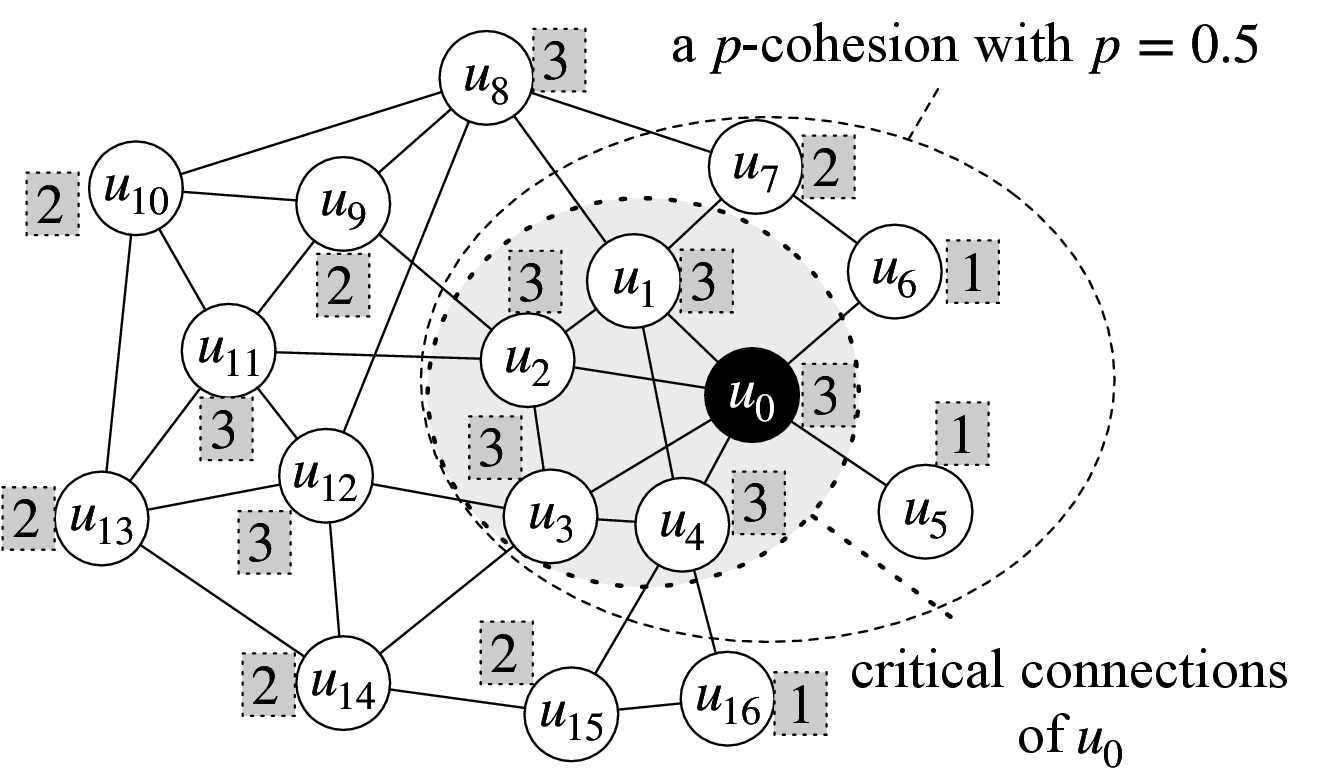}
\end{center}
\vspace{-2mm}
\caption{\small{Motivating Example}}
\label{sec:intro:motivation}
\end{figure}

\begin{example}
\label{sec:intr:examp1}
Suppose that we have a small social network $G$ with $17$ users, which is illustrated in Fig.~\ref{sec:intro:motivation}.
For each user $u_i \in G$, the number of neighbors required in a \pc is shown with a grey-filled dotted line square, when $p = 0.5$.
User $u_0$'s \pc $C_p$ is enclosed in a dashed oval, encompassing all relationships. By contrast, $u_0$'s critical relationships (a minimal p-cohesion $MC_p$) are marked within a grey-filled dashed circle—these connections are 
heavily interconnected.
Applying the DDP mechanism to protect $u_0$'s all connections would introduce excessive noise
and deteriorate the data utility.
By contrast, 
protecting $u_0$'s connections with a focus on those critical, \ie, $MC_p$,
under the same mechanism can balance data utility and privacy.
\end{example}

\vspace{2mm}
\noindent \textbf{Contributions.}
The main contributions are as follows:

\begin{itemize}
    \item \textit{Identify critical connections}. 
    To identify the critical connections, 
    we propose new score functions to effectively extract minimal \pcs (\ie, critical connections) from a graph with different scales.
    Those score functions are proposed based on graph density, 
    which is conducive to finding minimal \pcs with higher densities and strong cohesiveness.
    
    \item \textit{Privacy preservation for critical connections in response to queries}. 
    To protect the identified critical connections of a vertex in response to a query (\ie, $k$-clique counting), 
    we adopt the DDP mechanism and prove that $(\varepsilon, \delta)$-DDP can be satisfied by only protecting the information of those identified critical connections.

    \item \textit{Experimental Verification}. 
    We conduct comprehensive experiments on real-life graph datasets.
    For $k$-clique counting, we demonstrate our proposed framework provides $(\varepsilon, \delta)$-DDP guarantee to protect the critical connections of vertices and achieves higher data utility than the ELV-based solution~\cite{DBLP:conf/ccs/SunXKYQWY19}.
\end{itemize}

The remainder of this paper is as follows. In Section~\ref{sec:pre}, we present the preliminary knowledge for this paper, followed by problem formulation and high-level solution in Section~\ref{sec:prob}. The procedure for identifying critical connections is shown in Section~\ref{sec:mpc}, which is the foundation of the differentially private query release proposed in Section~\ref{sec:ddp}. The performance of the model is evaluated in Section~\ref{sec:eval}. Related studies are reviewed in Section~\ref{sec:relwork}. A conclusion is reached in Section~\ref{sec:conclusion}. 

%% file: preliminary.tex
\section{Preliminary}
\label{sec:pre}

\begin{table}[t]
\small
  \centering
  \vspace{6mm}
  \caption{Summary of Notations}
  \vspace{1mm}
    \label{tb:notations}
    \resizebox{\columnwidth}{!}{\begin{tabular}{|c|l|}
      \hline
      \textbf{Notation}   & \textbf{Definition}                   \\ \hline \hline

      $G$           &  An unweighted and undirected graph  \\ \hline

      $V(G),E(G)$   &  The vertex set and the edge set of the graph $G$  \\ \hline

      $u$ ($v$)      &  A vertex in the graph \\ \hline

      $n$,  $m$     &  Numbers of vertices and edges in $G$  \\ \hline

      $N(v, G)$     &  The set of adjacent vertices (neighbors) of $v$ in $G$  \\ \hline

      $\mathcal{F}$ &  A subgraph model  \\ \hline

      $N_{1h}^v(G)$ &  The set of $1$-hop neighbors of $v$ in $G$, \ie, $N(v, G)$  \\ \hline

      $deg(v, G)$   &  The degree of $v$ in $G$, \ie, $deg(v, G)$ = $|N(v, G)|$  \\ \hline

      $C_p(v, G)$   &  A \pc containing $v$ in $G$  \\ \hline

      $MC_p(v, G)$  &  A minimal \pc containing $v$ in $G$  \\ \hline

      $d(G)$        &  The density of graph $G$  \\ \hline

      {\revision $\Gamma_S(v)$} &  {\revision The count of shape $S$ containing $v$ in graph $G$}  \\ \hline
      
      {\revision $\Gamma_{S_{in}}(v)$} &  {\revision The count of shape $S$ containing $v$ that entirely in $MC_p(v, G)$}  \\ \hline
      
      {\revision $\Gamma_{S_{out}}(v)$} &  {\revision \makecell[l]{The count of shape $S$ containing $v$ that are partially or entirely \\ outside $MC_p(v, G)$, \ie, $\Gamma_{S_{out}}(v) = \Gamma_S(v) - \Gamma_{S_{in}}(v)$}}  \\ \hline
\end{tabular}}
\end{table}

Given an undirected and unweighted graph $G = (V, E)$, where $V$ (\resp $E$) is the set of vertices (\resp edges), we use $n = |V|$ and $m = |E|$ to denote the number of vertices and edges, respectively.
Suppose $N(v, G)$ is the neighbor of $v$ in $G$, we use $deg(v, G)$ to denote the degree of $v$ in $G$, where $|N(v, G)| = deg(v, G)$.
Let $S$ denote a subgraph of $G$.
We use $V(S)$ to represent the vertices of $S$.
When the context is clear, we may omit the target graph in notations, \textit{e.g.}, using $N(v)$ (\textit{resp.}, $S$) instead of $N(v, G)$ (\textit{resp.}, $V(S)$).
Table~\ref{tb:notations} summarizes the notations.

\vspace{2mm}
\noindent \textbf{$p$-Cohesion}.
In~\cite{DBLP:journals/kais/LiZZQZL21}, it is shown that \pc is a fortress-like cohesive subgraph, in which, every vertex has at least a fraction $p$ of its neighbors in the subgraph, \ie, at most a fraction $(1-p)$ of its neighbors outside.
In other words, a \pc can ensure inner cohesiveness.

\begin{definition}
\label{def:pc}
\textbf{($p$-Cohesion)}.
Given a graph $G$, a real value $p \in (0, 1)$, and a connected subgraph $S$, if for any vertex $v \in S$ with $deg(v, S) \geq \lceil p \times deg(v, G) \rceil$, we say $S$ is a \pc, denoted by $C_p(G)$.
\end{definition}

Let $C_p(v, G)$ be the \pc containing $v$.
We use $C_p(v)$ instead of $C_p(v, G)$ when the context is clear.

\vspace{2mm}
\noindent \textbf{Minimal \pc}.
A subgraph $S$ is a minimal \pc if no proper subgraph $S^\prime \subset S$ that qualifies as a \pc.
For example, irrespective of the chosen $p$ value, for every vertex in graph $G$, the whole graph $G$ is a \pc but not minimal.
Since minimal \pc is the elementary unit of a \pc, we concentrate on determining the minimal \pc that includes a specific vertex $v$.
For a \pc $C_p(v)$, by deleting every redundant vertex in $C_p(v)$ except $v$, we can obtain the minimal \pc that includes $v$, denoted by $MC_p(v, G)$.
In the following, when we say ``the \pc containing $v$'', we mean the minimal one.

\begin{definition}
\label{def:cc}
\textbf{(Critical Connection)}.
Given a graph $G$ and a vertex $v$, we say the connections identified by the minimal \pc containing $v$ are $v$'s critical connections.
\end{definition}

\vspace{2mm}
\noindent \textbf{Differential Privacy}.
\underline{D}ifferential \underline{P}rivacy (DP) was first introduced by Dwork \etal~\cite{DBLP:conf/tcc/DworkMNS06},
and has been widely used in the privacy-enhanced release and analysis of data.

\begin{definition}
\label{def:dp}
\textbf{(DP)}. A randomization mechanism $\mathcal{M}$ is $(\varepsilon, \delta)$-differentially private if, for any two neighboring datasets $D$ and $D^\prime$ that differs by one record and for all $\mathcal{S} \subseteq$ Range$(\mathcal{M})$, 

\begin{small}
\begin{equation}
\label{sec:pre:eq:dp}
  \Pr[\mathcal{M}(D) \in \mathcal{S}] \leq \Pr[\mathcal{M}(D^\prime) \in \mathcal{S}] \cdot e^\varepsilon + \delta,
\end{equation}
\end{small}
\end{definition}

\noindent
where $\varepsilon$ is the privacy budget that controls the strength of privacy
protection, and $\delta \in [0, 1]$ denotes a failure probability.
When $\delta = 0$, $\mathcal{M}$ is $\varepsilon$-differentially private.

\vspace{2mm}
\noindent \textbf{Neighboring Graph}.
Given a graph $G = (V, E)$ with vertex (\resp edge) set is $V$ (\resp $E$), its neighboring graph $G^\prime = (V^\prime, E^\prime)$ can be generated by either adding/removing an isolated vertex in $V$ or by adding/removing an edge in $E$~\cite{DBLP:conf/icdm/HayLMJ09, DBLP:conf/ccs/SunXKYQWY19}.
In this paper, we focus on the edge privacy model, \ie, $V = V^\prime$ and $||E| - |E \cap E^\prime|| = 1$.

\vspace{2mm}
\noindent \textbf{Neighboring Subgraph}.
Given a graph $G = (V, E)$, its neighboring graph $G^\prime = (V, E^\prime)$ and a subgraph model $\mathcal{F}$ (\textit{e.g.}, \pc).
For a vertex $v \in V$, under model $\mathcal{F}$, its subgraph on $G$ is denoted by $S = \mathcal{F}(v, G)$.
The neighboring subgraph $S^\prime$ of $S$ on $G^\prime$ is $S^\prime = \mathcal{F}(v, G^\prime)$, which may be different in both edges and vertices.

\vspace{2mm}
\noindent \textbf{Decentralized Differential Privacy (DDP)}.
In a graph, each vertex can release its own information to non-trusted parties using DDP.

\begin{definition}
\label{def:ddp}
\textbf{(DDP)}. Given a vertex set $V = \{v_1, v_2, \dots, v_n\}$ with $n$ nodes, a subgraph model $\mathcal{F}$, a set of randomization mechanisms $\{\mathcal{M}_i, 1 \leq i \leq n\}$ satisfy $(\varepsilon, \delta)$-DDP, if

\begin{small}
\begin{equation}
    \begin{aligned}
    \label{sec:pre:eq:ddp}
      Pr[\mathcal{M}_1(\mathcal{F}(v_1, G)) \in \mathcal{S}_1, \dots, \mathcal{M}_n(\mathcal{F}(v_n, G)) \in \mathcal{S}_n] \leq \\
      Pr[\mathcal{M}_1(\mathcal{F}(v_1, G^\prime)) \in \mathcal{S}_1, \dots, \mathcal{M}_n(\mathcal{F}(v_n, G^\prime)) \in \mathcal{S}_n] \cdot e^\varepsilon\!+\!\delta,
    \end{aligned}
\end{equation}
\end{small}

\noindent
for any two neighboring graphs $G = (V, E)$ and $G^\prime = (V, E^\prime)$ and for all subsets $\{\mathcal{S}_i \subseteq$ Range$(\mathcal{M}), 1 \leq i \leq n\}$.
\end{definition}

\begin{definition}
\label{def:ls}
\textbf{(Local Sensitivity under \ddp~\cite{DBLP:conf/ccs/SunXKYQWY19})}.
Given a graph $G = (V,E)$ with $n$ nodes, its arbitrary neighboring graph is $G^\prime = (V, E^\prime)$, where $V = \{v_1, v_2, \dots, v_n\}$, a subgraph model $\mathcal{F}$, and a function $f$, the local sensitivity of $f$ is defined as:

\begin{small}
\begin{equation}
\begin{aligned}
\label{sec:ls:eq}
  LS(f) = \max_{G, G^\prime} \sum_{i = 1}^{n}|f(\mathcal{F}(v_i, G)) - f(\mathcal{F}(v_i, G^\prime))|
\end{aligned}
\end{equation}
\end{small}
\end{definition}
\noindent
where $\mathcal{F}(v_i, G)$ and $\mathcal{F}(v_i, G^\prime)(1 \leq i \leq n)$ are the neighboring subgraphs with respect to $G$ and $G^\prime$, respectively.

%% file: problemStatement.tex
\section{Problem and Approach Overview}
\label{sec:prob}

{\color {black}
We study the differentially private critical connection protection problem in a graph under the client-server model.
Specifically, under the DDP mechanism, we aim to protect a vertex's critical connections when responding to queries (\ie, subgraph counting) raised by data analysts.
The subgraph counting is to count the occurrences of a user-specified shape (or subgraph) $S$ (\textit{e.g.}, triangle) containing a target vertex $v$ in graph $G$.
}

\vspace{2mm}
\noindent \textbf{Problem Statement.}
Given an undirected unweighted graph $G\!=\!(V,\!E)$, a critical number $p \in (0,1)$, a privacy budget $\varepsilon$, an invalidation probability $\delta$, a set of randomization mechanisms $\{\mathcal{M}_i, 1\!\leq\!i\!\leq\!n\}$, and a user-specified shape $S$;
for each vertex $v_i\!\in\!V$, subgraphs counting under \ddp for critical connections protection is to 
\begin{itemize}
\item identify $v_i$'s critical connections from the minimal \pc containing $v_i$. $M\!C_p(v_i,\!G)$ denotes the induced subgraph of $v_i$'s critical connections from $G$;
\item count the occurrences of $S$ shapes containing $v_i$ in $MC_p(v_i,\!G)$. $\Gamma_{S_{in}}(v_i)$ denotes the count;
\item perturb the count $\Gamma_{S_{in}}(v_i)$ using $\mathcal{M}_i$. The perturbed count is denoted by $\Gamma_{S_{in}}^*(v_i)$;
\item count the occurrences of $S$ shapes containing $v_i$ in the complementary part of $MC_p(v_i,\!G)$ in $G$. The count is denoted by $\Gamma_{S_{out}}(v_i)$; and
\item report the total count $\Gamma_S(v_i) = \Gamma_{S_{in}}^*(v_i) + \Gamma_{S_{out}}(v_i)$ to the data analyst, which ensures ($\varepsilon, \delta$)-\ddp.
\end{itemize}

\vspace{2mm}
\subsection{Overview of Our Approach}
\label{sec:prob:dis}

We propose the following techniques to solve the problem.
\begin{itemize}
\item \textbf{Identifying critical connections using minimal \pcs}. 
As mentioned earlier, \pc is an important cohesive subgraph model~\cite{DBLP:journals/kais/LiZZQZL21}, which can be used to identify one's critical connections. The \pc forms a fortress; \ie, the vertices inside a \pc are more connected than those outside.
Considering the elegant fortress property of minimal \pc, for a vertex $v_i$, we use $v_i$'s minimal \pc to identify its critical connections that need to be protected.
We prove that our solution can protect the privacy of the target and its critical connections, and as a result,  can provide better data utilities than the existing ELV-based solution developed in~\cite{DBLP:conf/ccs/SunXKYQWY19}.
More details are shown in Section~\ref{sec:mpc}.
\item \textbf{Protecting critical connections identified by minimal \pcs}. 
We propose a tailored $(\varepsilon, \delta)$-DDP mechanism, which
reduces the amount of injected noises and improves the data utility by only perturbing the query results generated from one's minimal \pc and keeping the following response as it is.
For the subgraph counting query problem, the function $f$ counts the occurrence of such shapes (or subgraphs) $S$ containing a vertex $v_i$.
{\revision
Let $MC_p(v_i)$ be a minimal \pc of vertex $v_i$, which contains its critical connections and the corresponding vertices.
Let $\Gamma_{S_{in}}(v_i)$ denote the number of occurrences of $S$ shape containing $v_i$ within $MC_p(v_i)$.
This implies that both the vertices and the connections that form $S$ shape are entirely within $MC_p(v_i)$.
Let $\Gamma_{S_{out}}(v_i)$ denote the number of $S$ shape that are partially or entirely outside $MC_p(v_i)$,
and let $\Gamma_S(v_i)$ denote the total count of $S$ shape within the entire graph $G$, we have $\Gamma_S(v_i) = \Gamma_{S_{in}}(v_i) + \Gamma_{S_{out}}(v_i)$.
}
We only perturb $\Gamma_{S_{in}}(v_i)$ in response to the query regarding $v_i$, \ie, reporting $\Gamma_{S}^*(v_i) = \Gamma_{S_{in}}^*(v_i) + \Gamma_{S_{out}}(v_i)$ to the data analyst, where $\Gamma_{S_{in}}^*(v_i)$ gives the perturbed counts of the user-given shape $S$ in the minimal \pc, \ie, $\Gamma_{S_{in}}^*(v_i) = \Gamma_{S_{in}}(v_i) + Lap(\lambda)$.
In Section~\ref{sec:ddp}, we prove that perturbed response $\Gamma_{S}^*(v_i)$ from $v_i$ still satisfies $(\varepsilon, \delta)$-\ddp for given $\varepsilon$ and $\delta$.
\end{itemize}

%% file: minimalPC.tex
\section{Critical Connection Identification}
\label{sec:mpc}

We identify critical connections with minimal \pc.
According to Definition~\ref{def:pc}, any connected component is a \pc, which could be large.
We devise an expand-shrink framework: Given a vertex, we repeatedly expand it into a $p$-cohesion and then shrink it to its minimal form.

\subsection{\textbf{Expand Procedure}}
\label{sec:mpc:lsa:expand}

Starting from a queried vertex $q$, this procedure finds a $p$-cohesion subgraph for $q$ in a bottom-up manner.
Suppose $V_p$ is the vertex set of a partial \pc that includes $q$, we repeatedly expand the vertices in $V_p$ to form a \pc.

\begin{algorithm}[h]
\small
\caption{\textit{~~Expand} ($G$, $q$, $p$)}
\label{alg:expand}
\begin{flushleft}
    \mbox{} \quad \textbf{Input} \quad : $G:$ a graph, $q:$ a queried vertex, \\
    \mbox{} \qquad \ \ \qquad \quad $p:$ a real number in ($0,1$)\\
    \mbox{} \quad \textbf{Output} : a $p$-cohesive subgraph containing $q$
\end{flushleft}
\begin{algorithmic}[1]
    \State $V_p := \{q\}$; $\mathcal{Q} := \emptyset$; \Comment{$\mathcal{Q}$: maximal priority queue} \label{alg:expand_1}
    \State put all $v \in V_p$ to $\mathcal{Q}$ with key on $deg(v, G(V_p))$ \label{alg:expand_2}
    \State $s(u) \leftarrow$ calculate a score for each vertex $u$ within $G$ \label{alg:expand_3}
	\While{ $\mathcal{Q} \neq \emptyset$ } \label{alg:expandwhile}
	    \State $u \leftarrow \mathcal{Q}.pop()$; $\mathcal{N} := \emptyset$; \Comment{$\mathcal{N}$: maximal priority queue} \label{alg:expandwhile_1}
	    \State put all $w \in N(u, G) - V_p$ to $\mathcal{N}$ with key on $s(w)$ \label{alg:expandwhile_4}
	    \State $b := \max(\lceil deg(u, G) \times p \rceil - deg(u, G(V_p)), 0)$ \label{alg:expandwhile_4.5}
	    \State $T \leftarrow$ top $b$ elements in $\mathcal{N}$; $V_P := V_P \cup T$ \label{alg:expandwhile_5}
	    \For {$w \in T$} \label{alg:expandwhile_6_for}
                \State \textbf{if}~$deg(w, G(V_p)) < \lceil deg(w, G) \times p \rceil$~\textbf{then}~$\mathcal{Q}.push(w)$ \label{alg:expandwhile_6_for_if_1}
	    \EndFor \label{alg:expandwhile_6_endfor}
	    \State update score $s(v)$ for each $v \in V(G)$ \label{alg:expandwhile_7}
	\EndWhile \label{alg:expandwhile_end}
    \State \textbf{Return} $G(V_p)$ \label{alg:expandreturn}
\end{algorithmic}
\end{algorithm}

Algorithm~\ref{alg:expand} outlines the \textbf{Expand} procedure for finding a \pc containing query vertex $q$.
The set $V_p$ records the vertices included to form a \pc, starting initially with $V_p = \{q\}$ (Line~\ref{alg:expand_1}).
For each vertex $u \in V_p$, it's necessary to add $|\lceil deg(u, G) \times p \rceil - deg(u, G(V_p))|$ additional vertices to $V_p$ to ensure $u$'s inclusion in a \pc. Here, $deg(u, G)$ and $deg(u, G(V_p))$ represent the number of $u$'s neighbors in $V(G)$ and $V_p$, respectively.
A priority queue $\mathcal{Q}$ represents vertices within $V_p$ that do not yet meet $p$ constraint (Line~\ref{alg:expand_1}), ordered by $deg(v, G(V_p))$ in a descending order (Line~\ref{alg:expand_2})
If two vertices have the same degree, the one with the smaller identifier is prioritized in $\mathcal{Q}$.
For each vertex $u$, we calculate a score $s(u)$ (Line~\ref{alg:expand_3}).

In each iteration (Lines~\ref{alg:expandwhile}-\ref{alg:expandwhile_end}), the vertex $u \in \mathcal{Q}$ with the maximum $deg(u, G(V_p))$ is removed from the queue at Line~\ref{alg:expandwhile_1}.
At Lines~\ref{alg:expandwhile_1}-\ref{alg:expandwhile_4}, we generate a maximal priority queue $\mathcal{N}$ to store the candidates that can help $u$ to reach the $p$ constraint, \ie, $u$'s neighbors outside $V_p$ (\ie, $N(u, G) - V_p$).
For each $w \in \mathcal{N}$, its key is the score $s(w)$ in descending order (Line~\ref{alg:expandwhile_4}). For two vertices with the same score $s$, the vertex with a smaller identifier will be placed at the top of the queue.
For the chosen $u$, we add its additional $b$ neighbors to $V_p$, where $b$ is the number of $u$'s neighbors that $u$ needs to meet the $p$ constraint (Line~\ref{alg:expandwhile_4.5}).
Let $T$ be the top $b$ vertices in $\mathcal{N}$ with the largest scores (Line~\ref{alg:expandwhile_5}).
At Line~\ref{alg:expandwhile_5}, all vertices in $T$ are added to $V_p$.
For any $w \in T$ that do not satisfy $p$ constraint, we add it to $\mathcal{Q}$ for further expansion (Lines~\ref{alg:expandwhile_6_for}-\ref{alg:expandwhile_6_endfor}).
All vertices' scores are updated after one-round processing at Line~\ref{alg:expandwhile_7}.
The algorithm will return at Line~\ref{alg:expandreturn} if $\mathcal{Q}$ is empty. In other words, $G(V_p)$ is a \pc.

\underline{\textit{Time complexity}}.
For $G(V_p)$, let $\tilde{n}$ and $\tilde{m}$ represent the number of its vertices and edges, respectively.
The number of to-expand vertices pushed into $\mathcal{Q}$ is $\tilde{n}$ (Line~\ref{alg:expandwhile}).
To retrieve value $b$ for each vertex, we need to visit its neighbors, which requires $\mathcal{O}(m + n)$.
Meanwhile, updating scores $s(\cdot)$ for the neighbors of $u$ requires $O(deg(u,G) \times log(deg(u, G)))$.
Algorithm~\ref{alg:expand}'s time complexity is $\mathcal{O}(\tilde{n}(m + n))$.

\underline{\textit{Space complexity}}.
Sets $V_p$, $\mathcal{Q}$, $\mathcal{N}$, $s(\cdot)$, and $deg(\cdot)$ each require $\mathcal{O}(n)$ space.
Similarly, $G$ and $N(\cdot)$ require $\mathcal{O}(m\!+\!n)$ space each. Algorithm~\ref{alg:expand}'s space complexity is $\mathcal{O}(m\!+\!n)$.

\vspace{1mm}
\noindent \textbf{Score Function}.
In Algorithm~\ref{alg:expand}, at Lines~\ref{alg:expandwhile_4}-\ref{alg:expandwhile_5}, we add $b$ of $u$'s neighbors with the largest scores to $V_p$.
We define two new scores: \textit{merit} and \textit{penalty}, since adding one vertex $w \notin V_p$ to $V_p$ may need to include more of $w$'s neighbors to help $w$ stay in a \pc. A vertex that can better balance the trade-off between merit and penalty should be added to $V_p$.

\vspace{1mm}
\noindent \textit{Density}.
As mentioned in Section~\ref{sec:prob:dis}, for a queried vertex $q$, we search for a subgraph containing $q$ with higher density.
We first define the concept of graph density.
Given an undirected graph $G = (V, E)$, the density is defined as:

\begin{small}
\begin{equation}
\label{sec:mpc:eq:density}
d(G) = 2|E| / (|V|(|V| - 1)),
\end{equation}
\end{small}

\noindent
where $|V|$ (\resp $|E|$) is the number of vertices (\resp edges) in $G$, and $0 < d(G) \leq 1$.
$d(G) = 1$ means the graph $G$ is fully connected, \ie, the graph is the densest.

Given the partial vertex set $V_p$ of a \pc that includes queried vertex $q$, let $G(V_p)$ and $E(G(V_p))$ be its induced subgraph and the corresponding edge set, respectively.
We have $E(G(V_p)) = \{(u,v)|(u,v) \in E, u \in V_p, v\in V_p\}$.
Based on Eq.~(\ref{sec:mpc:eq:density}), the density of $G(V_p)$ is given by:

\begin{small}
\begin{equation}
\label{sec:mpc:eq:density_1}
d(G(V_p)) = 2|E(G(V_p)| / (|V_p|(|V_p| - 1)).
\end{equation}
\end{small}

When adding a vertex $w \notin V_p$ to $V_p$, the density is:

\begin{small}
\begin{equation}
\label{sec:mpc:eq:density_2}
d(G(V_p \cup \{w\})) = \frac{2(|E(G(V_p))| + deg(w, G(V_p)))}{(|V_p|+1)|V_p|}.
\end{equation}
\end{small}

By adding a candidate $w$ with larger $deg(w, G(V_p))$, the density of the resulting subgraph can be increased.

Furthermore, we may need to include additional neighbors of $w$ into $V_p$ to maintain \pc.
Let $l$ be the number of neighbors required for $w$ to maintain a \pc. We have $l = \max(\lceil deg(w, G) \times p \rceil - deg(w, G(V_p)),0)$.
Let $N_c(w) = \{o_1, o_2, \dots, o_l\}$ ($o_i \in N(w, G-G(V_p))\}$) be $w$'s neighbors (outside $V_p$) chosen to be expanded to $V_p$, and $\widetilde{V_p} = V_p \cup \{w\} \cup N_c(w)$.
The potential density of $G(\widetilde{V_p})$ is:

\begin{small}
\begin{equation}
\label{sec:mpc:eq:density_3}
\begin{aligned}
& d(G(\widetilde{V_p})) \approx \frac{2}{(|V_p|+1+l)(|V_p|+l)} \\
&\!\times\!(|E(G(V_p))|\!+\!deg(w, G(V_p))\!+\!\sum_{i=1}^{l}(1+deg(o_i, G(V_p)))).
\end{aligned}
\end{equation}
\end{small}

Generally speaking, the potential density depends primarily on $d^\prime = \frac{deg(w, G(V_p)) + \sum_{i=1}^{l} 1+deg(o_i, G(V_p))}{l(l+1)}$.
A higher $d^\prime$ value will result in a higher subgraph density: the vertex $w$ with larger $deg(w, G(V_p))$ and more neighbors in common with $q$ in $V_p$ should be chosen.
In light of this, we propose the following merit score.

\vspace{1mm}
\noindent \underline{\textit{Merit}}.
The inclusion of a vertex $w \notin V_p$ should (i) increase the density of $G(V_p)$ and (ii) contribute to the decrease of the number of $w$'s neighbors $u$ in $V_p$ with $deg(u, G(V_p)) < \lceil deg(u, G) \times p \rceil$.
We define the merit score of $w$ as follows:

{\revision
\begin{small}
\begin{equation}
\label{sec:mpc:eq:merit}
s\!^+\!(w)\!=\!\frac{deg(w\!,\!G\!(\!V_p\!))}{deg(w\!,\!G)}\!\times\!\frac{|N_{cn}(w\!,\!q\!,\!G\!(\!V_p\!))|}{deg(w\!,\!G)}\!\times\!\frac{|N_{ctri}(w\!,\!G\!(\!V_p\!))|}{deg(w\!,\!G)},
\end{equation}
\end{small}
}

\noindent
where $N_{cn}(w,\!q,\!G(V_p))$ is the common neighbors between $w$ and $q$ in $V_p$, and $N_{cn}(w, q, G(V_p))\!=\!\{u| u\!\in\!N(q,\!G(V_p)),\!u\!\in\!N(w, G(V_p)),\!w\!\notin\!V_p\}$.
Similarly, $N_{ctri}(w, G(V_p))$ is the set of $w$'s neighbor $u$ in $V_p$ that do not satisfy the $p$ constraint, \ie, $N_{ctri}(w, G(V_p))\!=\!\{u| u\!\in\!N(w,\!G(V_p)),\!deg(u,\!G(V_p))\!<\!\lceil\!deg(u, G)\!\times\!p\!\rceil\}$.
The first two parts of $s^+(w)$ account for the influence on the $G(V_p)$ density when adding $w$ to $V_p$. The third part of $s^+(w)$ is $w$'s contribution to the increase in the degrees of some vertices within $V_p$.
To eliminate the influence of dimension between indicators, we normalize each part of $s^+(w)$ by $deg(w,G)$ in Eq.~(\ref{sec:mpc:eq:merit}).

\vspace{1mm}
\noindent \underline{\textit{Penalty}}.
To expand a vertex $u\!\in\!V_p$, adding its neighbor $w\!\notin\!V_p$ will have A penalty effect: (i) $w$ may need more neighbors to meet $p$ constraint, and (ii) $d(G(\!\widetilde{V_p}\!))$ may be no larger than $d(G(\!V_p\!))$ if adding both $w$ and its outside neighbors $N_c(w)$.
To quantify this effect, we propose the penalty score:

{\revision
\begin{small}
\begin{equation}
\label{sec:mpc:eq:penalty}
s^-(w) = \frac{\overline{deg}(w, G(V_p))}{deg(w,G)} / (\sum_{i=1}^{l} \frac{deg(o_i, G(V_p)))}{deg(w,G)}),
\end{equation}
\end{small}
}

\noindent
where $\overline{deg}(w,\!G(\!V_p\!))\!=\!\max(0,\!\lceil\!deg(w, G)\!\times\!p\!\rceil\!-\!deg(w, G(V_p)))$ is the number of $w$'s neighbors that is required to meet the $p$ constraint, and $l\!=\!\overline{deg}(w, G(V_p))$. $o_i$ is $w$'s neighbor outside $V_p$ (\ie, $o_i \in N(w,\!G\!-\!G(V_p))$).
Adding $w$ may need more vertices to be added to $V_p$, which may decrease the density. To lessen this effect, we place $\overline{deg}(w, G(V_p))$ at the numerator.
Besides, as shown in Eq.~(\ref{sec:mpc:eq:density_3}), including a vertex $o_i$ with larger $deg(o_i, G(V_p))$ could result in a larger density. To enlarge $o_i$'s effect, we place $deg(o_i, G(V_p))$ at the denominator.
We also normalize each part by $deg(w,G)$.

Considering both merit and penalty, we propose the overall score for each vertex outside $V_p$ as follows:
{\revision
\begin{small}
\begin{equation}
\label{sec:mpc:eq:fw}
s(w) = s^+(w) - s^-(w).
\end{equation}
\end{small}
The score $s(\cdot)$ can help to decide which vertex should be added to $V_p$ in Line~\ref{alg:expandwhile_4} of Algorithm~\ref{alg:expand}
}

\begin{example}
\label{sec:mpc:lsa:expand:example1}
For the graph $G$ in Fig.~\ref{sec:intro:motivation}, suppose $q = u_0$, in Algorithm~\ref{alg:expand}, $u_0$ will be popped in the first iteration, \ie, $u = u_0$ (Line~\ref{alg:expandwhile_1}).
The candidates will be $\mathcal{N}\!=\!\{u_1,\!u_2,\!u_3,\!u_4,\!u_5,\!u_6\}$ (Line~\ref{alg:expandwhile_4}), and $b = 3$.
We may first choose $u_1$, $u_5$, and $u_6$ to $V_p$ in Lines~\ref{alg:expandwhile_4.5}-\ref{alg:expandwhile_5}, \ie, $V_p\!=\!\{u_0,\!u_1,\!u_5,\!u_6\}$.
$u_1$ will be pushed into $\mathcal{Q}$ in Lines~\ref{alg:expandwhile_6_for}-\ref{alg:expandwhile_6_endfor} since it does not satisfy $p$ constraint.
In the second iteration, $u_1$ will be popped for further process, \ie, $u\!=\!u_1$. The candidates for $u_1$ are $\mathcal{N}\!=\!\{u_2,\!u_4,\!u_7,\!u_8\}$,
from which, $u_2$ and $u_7$ will be expanded, \ie, $V_p\!=\!\{u_0,\!u_1,\!u_5,\!u_6,\!u_2,\!u_7\}$.
Since $u_1$'s $b\!=\!\max(3-1,\!0)\!=\!2$ and, $u_2$, $u_7$ are with the top two largest score in $\mathcal{N}$, \ie, $s(u_2)\!=\!s^+(u_2)\!-\!s^-(u_2)\!=\!\frac{2 \time 1 \time 1}{5}\!-\!\frac{\frac{3-2}{5}}{\dots}\!=\!-\frac{3}{5}$,
$s(u_7)\!=\!\frac{4}{3}$, $s(u_4)\!=\!-\frac{3}{5}$, $s(u_8)\!=\!-\frac{4}{5}$.
$u_2$ will be pushed into $\mathcal{Q}$ for further expernsion in Lines~\ref{alg:expandwhile_6_for}-\ref{alg:expandwhile_6_endfor}.
In the next two iterations, $u_3$ and $u_4$ will be expanded respectively, \ie, $V_p\!=\!\{u_0,\!u_1,\!u_5,\!u_6,\!u_2,\!u_7,\!u_3,\!u_4\}$.
Then $\mathcal{Q}$ will be empty. Algorithm~\ref{alg:expand} returns the subgraph induced by $V_p\!=\!\{u_0,\!u_1,\!u_5,\!u_6,\!u_2,\!u_7,\!u_3,\!u_4\}$.
\end{example}

\subsection{\textbf{Shrink Procedure}}
\label{sec:mpc:lsa:shrink}

The expanded \pc containing $q$ returned by Algorithm~\ref{alg:expand} may include redundant vertices. To further minimize the vertices number in \pc and identify $q$'s close connections, redundancies need to be removed.

\begin{algorithm}[ht]
\small
\caption{\textit{~~Shrink} ($C_p(q, G)$, $q$, $p$)}
\label{alg:shrink}
\begin{flushleft}
    \mbox{} \quad \textbf{Input}: $C_p(q, G):$ a \pc of graph $G$ containing $q$, \\
    \mbox{} \qquad \ \ \qquad $q:$ a queried vertex, $p:$ a real value within $(0,1)$\\
    \mbox{} \quad \textbf{Output} : a minimal $p$-cohesive subgraph containing $q$
\end{flushleft}
\begin{algorithmic}[1]
    \State $S \leftarrow C_p(G)$;$M_q := \{q\}$ \Comment{$M_q$: a must included vertex set} \label{alg:shrink_1}
    \State $Tag \leftarrow [0]*|V(S)|$; $Tag[q] := 2$; \label{alg:shrink_2}
    \While{ $0 \in Tag$ } \label{alg:shrinkWhile}
        \For{$v \in S$} \label{alg:shrinkWhile_1_for}
            \State \textbf{if}~$Tag[v] > 0$~\textbf{then}~continue \label{alg:shrinkWhile_1_for_1_If_1}
           \State $S^\prime \leftarrow S$; $Tag^\prime \leftarrow Tag$ \label{alg:shrinkWhile_1_for_2}
            \State $S \leftarrow S - \{v \cup E(v)\}$; $Tag[v] := 1$  \label{alg:shrinkWhile_1_for_3}
            \While{$\exists u \in S$ with $deg(u, S) < \lceil deg(u, G) \times p \rceil$} \label{alg:shrinkWhile_1_for_4_While}
                \State $S \leftarrow S - \{u \cup E(u, S)\}$; $Tag[u] := 1$ \label{alg:shrinkWhile_1_for_4_While_1}
            \EndWhile \label{alg:shrinkWhile_1_for_4_EndWhile}
            \If{$\exists u \in M_q~\&~u \notin S$} \label{alg:shrinkWhile_1_for_5_If}
                \State $S\!\leftarrow\!S^\prime$; $Tag\!\leftarrow\!Tag^\prime$;$M_q\!\leftarrow\!M_q \cup \{v\}$; $Tag[v] := 2$ \label{alg:shrinkWhile_1_for_5_If_1}
            \EndIf \label{alg:shrinkWhile_1_for_5_endIf}
        \EndFor \label{alg:shrinkWhile_1_Endfor}
    \EndWhile \label{alg:shrinkWhile_end}
    \State \textbf{Return} $S$ \label{alg:shrinkreturn}
\end{algorithmic}
\end{algorithm}

Algorithm~\ref{alg:shrink} shows the pseudo-code of \textbf{Shrink} procedure, which finds a minimal \pc that includes the queried vertex $q$.
Here, $C_p(q, G)$ represents a \pc containing $q$, which is obtained from the output of Algorithm~\ref{alg:expand}.
We use $M_q$ to record vertices that must be included in a minimal \pc that includes $q$. Initially $M_p = \{q\}$ (Line~\ref{alg:shrink_1}).
$S$ is a back-up of $C_p(q, G)$ (Line~\ref{alg:shrink_1}).
The set $Tag$ ensures the deletion of each selected vertex only once.
Initially, all vertices in $S$ have $Tag[i] = 0$ ($1 \leq i \leq |V(S)|$), except for the queried vertex $q$ ($Tag[q] = 2$).
$Tag[i] = 0$ means the vertex $i$ has not been proceeded, $Tag[i] = 1$ means vertex $i$ has been deleted, and $Tag[i] = 2$ (Line~\ref{alg:shrink_2}) means vertex $i$ must be included in the minimal \pc.

We proceed to delete vertex $v \in S$ with $Tag[v] = 0$ (Lines~\ref{alg:shrinkWhile}-\ref{alg:shrinkWhile_end}).
Sets $S^\prime$ and $Tag^\prime$ are used for recovering, since $v$'s removal may result in the removal of vertices in $M_q$ (Lines~\ref{alg:shrinkWhile_1_for_5_If}-\ref{alg:shrinkWhile_1_for_5_endIf}).
After removing $v$, its corresponding edges will be removed and its tag will be set as $Tag[v] = 1$ (Line~\ref{alg:shrinkWhile_1_for_3}).
The removal of $v$ may cause the \pc violation for some vertices in $S$, we remove those vertices and their corresponding edges at Lines~\ref{alg:shrinkWhile_1_for_4_While}-\ref{alg:shrinkWhile_1_for_4_EndWhile}.
If $v$'s removal will result in the deletion of any vertex in $M_q$, we will roll back $S$, $Tag$ and put $v$ into $M_q$, which means $v$ cannot be removed and must be included in the minimal \pc that includes $q$ (Lines~\ref{alg:shrinkWhile_1_for_5_If}-\ref{alg:shrinkWhile_1_for_5_endIf}).
After trying each vertex's removal, the algorithm returns the minimal \pc that includes $q$.

{\color {black}
\underline{\textit{Time complexity}}.
We use $\tilde{n}$ and $\tilde{m}$ to represent the number of vertices and edges in $S$, respectively.
Visiting every vertex in $S$ and $Tag$ takes $\mathcal{O}(\tilde{n})$ (Lines~\ref{alg:shrinkWhile}-\ref{alg:shrinkWhile_1_for}).
In then case when a vertex $v$ is removed in Line~\ref{alg:shrinkWhile_1_for_3} or Line~\ref{alg:shrinkWhile_1_for_4_While_1}, its neighbors may breach the $p$ constraint (Line~\ref{alg:shrinkWhile_1_for_4_While}).
Consequently, each vertex in $S$ is visited once for removal and each edge is traversed once for degree updating,  which needs $\mathcal{O}(\tilde{m}+\tilde{n})$.
The recovery of $S$ and $Tag$ requires $\mathcal{O}(\tilde{m}+\tilde{n})$ (Line~\ref{alg:shrinkWhile_1_for_5_If_1}).
Algorithm~\ref{alg:shrink}'s time complexity is $\mathcal{O}(\tilde{n}^2(\tilde{m}+\tilde{n}))$.

\underline{\textit{Space complexity}}.
Sets $M_q$, $Tag$, and $deg(\cdot)$ require $\mathcal{O}(n)$ space each, while $G$ and $S$ require $\mathcal{O}(m + n)$ space each. Algorithm~\ref{alg:shrink}'s space complexity is $\mathcal{O}(m + n)$.
}

%% file: DDP.tex
\section{Differentially Private Query Release Over Critical Connections}
\label{sec:ddp}

\subsection{Qualification of DDP on Critical Connections}
\label{sec:ddp:mpc}
In this section, we prove that by only protecting each participant's close connections (identified by the minimal \pc), the query response satisfies $(\varepsilon, \delta)$-DDP.
In the following, we focus on the fundamental problem of DP-based graph data release, \ie, subgraph counting.
Based on the description of \pc and minimal \pc, not all connections of a participant are included.
We only perturb the subgraph count within the minimal \pcs.
It is prudent to qualify the validity of DDP for the subgraph count obtained on minimal \pc.

Given graph $G$, let $N_{1h}^{v_i}(G)$ be the one-hop neighbors of $v_i$ in $G$.
Given a vertex $v_i\!\in\!G$, and its minimal \pc $M\!C_p(v_i,G)$, we have $N_{1h}^{v_i}(M\!C_p(v_i)) \subseteq N_{1h}^{v_i}(G)$.
Given a query function $f$ and a noise scale $\lambda$, $v_i$ reports its perturbed response $\Gamma_{\!S}^*(v_i)\!=\!\Gamma_{\!S_{in}}^*(v_i)\!+\!\Gamma_{\!S_{out}}(v_i)$ to the data collector, where the response generated from its critical connections in $M\!C_p(v_i)$ is $\Gamma_{\!S_{in}}^*(v_i)\!=\!\Gamma_{\!S_{in}}(v_i)\!+\!Lap(\lambda)$, $\Gamma_{\!S_{in}}(\!v_i\!)\!=\!f(M\!C_p(\!v_i\!))$, and the response generated from the the complementary part of $M\!C_p(v_i,\!G)$ in $G$ is $\Gamma_{S_{out}}$.
We show that applying a randomization mechanism $\mathcal{M}_i(M\!C_p(v_i, G))$ to the query response $\Gamma_{\!S}(\!v_i\!)$ satisfies \ddp. 

{\revision
\begin{theorem}
\label{sec:ddp:the1}
Given a graph $G$, a vertex $v_i$, its minimal \pc $MC_p(v_i)$, and a noise scale $\lambda$, we can assert that $\Gamma_{S}^*(v_i) = \Gamma_{S_{in}}(v_i) + Lap(\lambda) + \Gamma_{S_{out}}(v_i)$ ensures $\frac{LS(\Gamma_{S_{in}})}{\lambda}$-\ddp, where $LS(\Gamma_{S_{in}})$ is the maximum local sensitivity of all $\Gamma_{S_{in}}(v_i), 1 \leq i \leq n$.
\end{theorem}
}
\begin{proof}
{\revision Haipei \textit{et al.}~\cite{DBLP:conf/ccs/SunXKYQWY19} proved that when considering all connections of $v_i$ to generate a query response, the perturbed response $\Gamma_{S}^*(v_i)$ satisfies $\frac{LS(\Gamma_S^*)}{\lambda}$-\ddp.
}
In this case, $\Gamma_{S}^*(v_i) = \Gamma_{S_{in}}(v_i) + Lap(\lambda) + 0$ and  $LS(\Gamma_S) = LS(\Gamma_{S_{in}})$.

For a vertex $v_i \in G$, we identify critical connections to be protected using the minimal \pc (\ie, $MC_p(v_i)$), rather than all connections.
According to Definition~\ref{def:pc}, $deg(v_i, MC_p(v_i)) \geq \lceil deg(v_i, G) \times p \rceil$.
We can show that, for different $p$ values, our results satisfy $\frac{LS(\Gamma_{S_{in}})}{\lambda}$-\ddp, as follows.

\underline{When $p$ is large enough, \ie, $N_{1h}^{v_i}(M\!C_p(v_i))\!=\!N_{1h}^{v_i}(G)$,}
all connections of $v_i$ are in its minimal \pc $MC_p(v_i)$, \ie, $deg(v_i, MC_p(v_i)) = deg(v_i, G)$. The query response from the complementary part of the minimal \pc is $\Gamma_{S_{out}}(v_i)\!=\!0$ since all connections of $v_i$ are critical with larger $p$ values. Thus, $\Gamma_{S}^*(v_i) = \Gamma_{S_{in}}(v_i) + Lap(\lambda)$, satisfying $\frac{LS(\Gamma_{S_{in}})}{\lambda}$-\ddp.

\underline{When $p$ is small, \ie, $N_{1h}^{v_i}(MC_p(v_i)) \subset N_{1h}^{v_i}(G)$,}
not all adjacent vertices of $v_i$ are included in $MC_p(v_i)$. According to Definition~\ref{def:pc}, $deg(v_i, MC_p(v_i)) < deg(v_i, G)$, and $N_{1h}^{v_i}(MC_p(v_i)) \subset N_{1h}^{v_i}(G)$.
In this case, the targeted query subgraphs containing $v_i$ that situate outside $MC_p(v_i)$ may not be empty; \ie, $\Gamma_{S_{out}}(v_i)\!\geq\!0$.
{\revision
The final query response reported by vertex $v_i$ is $\Gamma_{S}^*(v_i) = \Gamma_{S_{in}}(v_i) + Lap(\lambda) + \Gamma_{S_{out}}(v_i)$.
Since we do not perturb the response $\Gamma_{S_{out}}$, the response $\Gamma_{S}^*(v_i)$ is a post-processed version of $\Gamma_{S_{in}}^*(v_i)$.
According to the post-processing composition property of DP~\cite{Dwork2006TCC}, $\Gamma_{S}^*(v_i)$ satisfies $\frac{LS(\Gamma_{S_{in}})}{\lambda}$-DDP.
}
\end{proof}

\vspace{-2mm}
{\color{black}
{\revision
In practice, the local sensitivity of the graph $G$ may not be available and needs to be estimated based on the feedback of all participants, $v_i,\forall i$.
We adopt a two-phase framework as in~\cite{DBLP:conf/ccs/SunXKYQWY19} to estimate local sensitivity, decide noise scale, and perturb responses to subgraph counting queries.
Specifically, Phase-$1$ requires each participant $v_i$ to obfuscate its local sensitivity $LS(\Gamma_{S_{in}}(v_i))$, and reports the obfuscated local sensitivity $\widetilde{LS(\Gamma_{S_{in}}(v_i))}$ to the data collector using ($\varepsilon_1, \delta_1$)-DDP.
Then, the data collector can evaluate the local sensitivity of the entire graph $G$ with
$LS(\Gamma_{S_{in}})\!=\!\max_{i=1}^n{\widetilde{LS(\Gamma_{S_{in}}(v_i))}}$, which also satisfies ($\varepsilon_1, \delta_1$)-DDP.
Here, $\varepsilon_1$ is the privacy budget, and $\delta_1$ is the violation probability, \ie, the probability of privacy protection being violated. In the case of the Laplace mechanism, $\delta_1$ is the probability of $LS(\Gamma_{S_{in}}(v_i)) > \widetilde{LS(\Gamma_{S_{in}}(v_i))}$, \ie, the estimated local sensitivity of vertex $v_i$ is smaller than the ground truth $LS(\Gamma_{S_{in}}(v_i))$.

In Phase-2, the data collector can decide the noise scale $\lambda$ based on $LS(\Gamma_{S_{in}})$, and send $\lambda$ to the participants. Then, each participant obfuscates its subgraph count by injecting Laplace noise $Lap(\lambda)$ to satisfy ($\varepsilon_2$, $\delta_2$)-DDP, 
and reports the obfuscated subgraph count in response to the subgraph counting queries of the data collector.
The Laplace mechanism~\cite{Dwork2006TCC} is adopted. By definition, $\lambda = \frac{LS(\Gamma_{S_{in}})}{\varepsilon_2}$, and 
$\delta_2$ is the probability of $\lambda > \frac{LS(\Gamma_{S_{in}})}{\varepsilon_2}$; \ie, the estimated noise scale of $G$ is smaller than the ground truth, $\lambda$, and consequently, privacy is violated~\cite[Lemma 4.1]{DBLP:conf/ccs/SunXKYQWY19}.
}

\begin{theorem}
\label{sec:ddp:the2}
Under the two-phase framework, the response to subgraph counting queries, with critical connections protected, satisfies $(\varepsilon_1 + \varepsilon_2, \delta_1+\delta_2)$-DDP.
\end{theorem}

\begin{proof}
For illustration convenience, this proof is based on triangle counting.
The key step of counting triangles is to count common neighbors between two endpoints of a connection.
Let $\varphi(v_i)$ (\resp, $\varphi_{in}(v_i)$) be $v_i$'s maximum common neighbors with others in graph $G$ (\resp, $MC_p(v_i)$).
Let $Y_i$ be the random variables drawn from $Lap(\cdot)$ for $v_i$.
In Phase-$1$, we have $\varphi_{in}^*(v_i) = \varphi_{in}(v_i) + Y_i$,
based on which the noise scale for Phase-$2$ can be identified, \ie, $\lambda = \max_{v_i \in G}\varphi_{in}^*(v_i)/\varepsilon_2$.
With at least ($1 - \delta_2$) probability, we have $\lambda \geq LS(\Gamma_{S_{in}})/\varepsilon_2$, since $LS(\Gamma_{S_{in}}) \leq \max_{v_i \in G}\varphi_{in}^*(v_i)$.

We verify that, for any neighboring graphs $G$ and $G^{\prime}$ and for any subgraph counts set $\Gamma_S^*$, we have

\begin{small}
\begin{equation}
\label{sec:ddp:eq:pr1}
{\rm Pr}[\Gamma_S^* \in \mathcal{S}_{\Gamma_S}, \lambda \in \mathcal{S}_{\lambda}|G] \leq e^{\varepsilon} \cdot {\rm Pr}[\Gamma_S^*  \in \mathcal{S}_{\Gamma_S}, \lambda \in \mathcal{S}_{\lambda}|G^\prime] + \delta.
\end{equation}
\end{small}

\noindent
where $\mathcal{S}_{\lambda}$ (\resp, $\mathcal{S}_{\Gamma_S}$) is an arbitrary set of outputs from Phase-$1$ (\resp, Phase-$2$).

We use $\mathcal{S}_{\lambda}^{\prime}$ to denote the subset of $\mathcal{S_{\lambda}}$ that satisfies

\begin{small}
\begin{equation}
\label{sec:ddp:eq:pr2}
\mathcal{S}_{\lambda}^\prime = \{\lambda|\lambda \in \mathcal{S}_{\lambda}~\&~\lambda \geq LS(\Gamma_{S_{in}})/\varepsilon_2\}.
\end{equation}
\end{small}

{\revision
\noindent
We have
\begin{small}
    \begin{equation}
        \label{sec:ddp:eq:pr2.1}
        \begin{aligned}
        & {\rm Pr}[\Gamma_S^* \in \mathcal{S}_{\Gamma_S}, \lambda \in \mathcal{S}_{\lambda}|G] \\
        & = {\rm Pr}[\Gamma_S^* \in \mathcal{S}_{\Gamma_S},\lambda \in \mathcal{S}_{\lambda}^{\prime}|G]\!+\!{\rm Pr}[\Gamma_S^* \in \mathcal{S}_{\Gamma_S},\lambda \in \mathcal{S}_{\lambda} \setminus \mathcal{S}_{\lambda}^{\prime}|G] \\
        & \leq {\rm Pr}[\Gamma_S^* \in \mathcal{S}_{\Gamma_S},\lambda \in \mathcal{S}_{\lambda}^{\prime}|G] + \delta_2,
        \end{aligned}
    \end{equation}
\end{small}

\noindent
since Phase-$1$ ensures $\lambda \geq LS(\Gamma_{S_{in}})/\varepsilon_2$ with at least $1 - \delta_2$ probability.
That is, in Phase-$1$, $\lambda$ is estimated using an $(\varepsilon_1, \delta_1)$-DDP. Thus, we have}
\begin{small}
\begin{equation}
\label{sec:ddp:eq:pr3}
{\rm Pr}[\lambda \in \mathcal{S}_{\lambda}^{\prime}|G] \leq e^{\varepsilon_1} \cdot {\rm Pr}[\lambda \in \mathcal{S}_{\lambda}^{\prime}|G^\prime] + \delta_1.
\end{equation}
\end{small}
Combining with Eq.~(\ref{sec:ddp:eq:pr3}), we have
{\revision
\begin{small}
    \begin{equation}
        \label{sec:ddp:eq:pr3.5}
        \begin{aligned}
        & {\rm Pr}[\Gamma_S^* \in \mathcal{S}_{\Gamma_S}, \lambda \in \mathcal{S}_{\lambda}^{\prime}|G] \\
        & = {\rm Pr}[\Gamma_S^* \in \mathcal{S}_{\Gamma_S}|\lambda \in \mathcal{S}_{\lambda}^{\prime}, G] \cdot {\rm Pr}[\lambda \in \mathcal{S}_{\lambda}^{\prime}| G] \\
        & \leq {\rm Pr}[\Gamma_S^* \in \mathcal{S}_{\Gamma_S}|\lambda \in \mathcal{S}_{\lambda}^{\prime}, G] \cdot (e^{\varepsilon_1} \cdot {\rm Pr}[\lambda \in \mathcal{S}_{\lambda}^{\prime}|G^\prime] + \delta_1) \\
        & \leq e^{\varepsilon_1} \cdot {\rm Pr}[\Gamma_S^* \in \mathcal{S}_{\Gamma_S}|\lambda \in \mathcal{S}_{\lambda}^{\prime}, G] \cdot {\rm Pr}[\lambda \in \mathcal{S}_{\lambda}^{\prime}|G^{\prime}] + \delta_1.
        \end{aligned}
    \end{equation}
\end{small}
}

Next, we show: for any $\lambda \geq LS(\Gamma_{S_{in}})/\varepsilon_2$ and any noisy subgraph count set $\chi$,

\begin{small}
\begin{equation}
\label{sec:ddp:eq:pr4}
{\rm Pr}[\Gamma_S^* = \chi|G] \leq e^{\varepsilon_2} \cdot {\rm Pr}[\Gamma_S^* = \chi|G^\prime],
\end{equation}
\end{small}

\noindent which leads to

\begin{small}
\begin{equation}
\label{sec:ddp:eq:pr4.1}
{\rm Pr}[\Gamma_S^* \in S_{\Gamma_S}|\lambda \in \mathcal{S}_{\lambda}^{\prime}, G] \leq e^{\varepsilon_2} \cdot {\rm Pr}[\Gamma_S^* \in S_{\Gamma_S}|\lambda \in \mathcal{S}_{\lambda}^{\prime}, G^\prime].
\end{equation}
\end{small}

For vertex $v_i$, let $\Gamma_S(v_i)$ and $\Gamma_S^{\prime}(v_i)$ be its triangle numbers in $G$ and $G^{\prime}$, respectively; where $\Gamma_S(v_i)\!=\!\Gamma_{S_{in}}(v_i)\!+\!\Gamma_{S_{out}}(v_i)$ and $\Gamma_S^{\prime}(v_i)\!=\!\Gamma_{S_{in}}^{\prime}(v_i)\!+\!\Gamma_{S_{out}}(v_i)$.
We have,

\begin{small}
\begin{equation}
\label{sec:ddp:eq:pr5}
\begin{aligned}
& \frac{{\rm Pr}(\Gamma_S^* = \chi|\lambda, G)}{{\rm Pr}(\Gamma_S^* = \chi|\lambda, G^\prime)} \\
&\!=\!\frac{\frac{1}{2\lambda}{\rm exp}(-\frac{1}{\lambda}\!\sum_{i=1}^{n}\!(|\Gamma_S^*(v_i)\!-\!\Gamma_S(v_i)|))}{\frac{1}{2\lambda}{\rm exp}(-\frac{1}{\lambda}\sum_{i=1}^{n}(|\Gamma_S^{*}(v_i)\!-\!\Gamma_S^{\prime}(v_i)|))} \\
&\!=\!\frac{\frac{1}{2\lambda}{\rm exp}(-\frac{1}{\lambda}\!\sum_{i=1}^{n}\!(|\Gamma_S^*(v_i)\!-\!(\Gamma\!_{S_{in}}(v_i)\!+\!\Gamma\!_{S_{out}}(v_i))|))}{\frac{1}{2\lambda}{\rm exp}(-\frac{1}{\lambda}\!\sum_{i=1}^{n}\!(|\Gamma_S^{*}(v_i)\!-\!(\Gamma_{S_{in}}^{\prime}(v_i)\!+\!\Gamma\!_{S_{out}}(v_i))|))} \\
&\!\leq\!{\rm exp}(\frac{1}{\lambda}\!\sum_{i=1}^{n}|\Gamma_{S_{in}}^\prime(v_i)\!+\!\Gamma\!_{S_{out}}(v_i)\!-\!(\Gamma\!_{S_{in}}(v_i)\!+\!\Gamma\!_{S_{out}}(v_i))|) \\
&\!\leq\!{\rm exp} (LS(\Gamma_{S_{in}})/\lambda)\!\leq\!e^{\varepsilon_2}.
\end{aligned}
\end{equation}
\end{small}

Eq.~(\ref{sec:ddp:eq:pr4}) is proved, and the two-phase framework under our setting ensures $(\varepsilon_1 + \varepsilon_2, \delta_1+\delta_2)$-DDP. 
\end{proof}
}

\subsection{Differentially Private Query Release}
\label{sec:ddp:app}

In this section, we investigate the differentially private query release problem for a graph: $k$-clique counting.
To protect critical connections when releasing responses, for each vertex, the minimal \pc is used to detect those relationships that need to be protected.
In the following, we show how to protect the critical connections identified by minimal \pc when releasing local $k$-clique counts.

In Theorem~\ref{sec:ddp:the2}, we have proved that the two-phase framework for subgraph counting ensures $(\varepsilon, \delta)$-\ddp whenever $\varepsilon_1+\varepsilon_2 \leq \varepsilon$ and $\delta_1 + \delta_2 \leq \delta$.
Let $\widetilde{LS(\Gamma_{S_{in}})}$ be the estimation of the local sensitivity under the edge privacy model~\cite{DBLP:conf/icdm/HayLMJ09}, which should be larger than the true local sensitivity $LS(\Gamma_{S_{in}})$.
The $\lambda = \frac{\widetilde{LS(\Gamma_{S_{in}})}}{\varepsilon_2} \geq \frac{LS(\Gamma_{S_{in}})}{\varepsilon_2}$ would be a possible noise scale.
In Phase-$2$, each vertex would report the subgraph number by injecting Laplace noise.

We proceed to estimate $LS(\Gamma_{S_{in}})$, \ie, the local sensitivity for counting $k$-cliques inside the minimal \pcs.

\vspace{2mm}
\noindent \textbf{$k$-Clique}.
$k$-Clique is a subgraph with exact $k$ vertices, and any two distinct vertices in the subgraph are adjacent.

For \kcs counting, under the edge privacy model~\cite{DBLP:conf/icdm/HayLMJ09}, estimating the local sensitivity $LS(\Gamma_{S_{in}})$ is to estimate the common neighbors that inside a minimal \pc between the two endpoints of a connection.
For a vertex $v_i$, to estimate its $LS(\Gamma_{S_{in}}(v_i))$, it should be one of the endpoints of the updating edge.
Updating edge $(v_i, v_j)$ will only affect the $k$-cliques containing both $v_i$ and $v_j$.
We use $\dot{N}_{cn}(v_j, v_i, MC_p(v_i))$ to denote the common neighbors between $v_i$ and $v_j$ in $v_i$'s minimal \pc. The common neighbor set is defined as $\dot{N}_{cn}(v_j, v_i, MC_p(v_i)) = \{u|u \in N(v_i, MC_p(v_i), u \in N(v_j, MC_p(v_i)), v_j \in MC_p(v_i)\}$.
When updating edge $(v_i, v_j)$, the affected number of \kcs containing $v_i$ and $v_j$ should be no more than $\tbinom{\dot{N}_{cn}(v_j, v_i, MC_p(v_i))}{k-2}$, since those \kcs consist of $v_i$, $v_j$ and $k-2$ of their common neighbors.
Each \kc will be reported $k$ times (by each of the vertices in the \kc), and the estimated local sensitivity should be

\begin{small}
\begin{equation}
\label{sec:ddp:eq:ls}
\widetilde{LS(\Gamma_{S_{in}})} = \max_{v_i \in G, v_j \in MC_p(v_i), v_j \neq v_i} k \cdot \tbinom{\dot{N}_{cn}^T(v_j, v_i, MC_p(v_i))}{k-2},
\end{equation}
\end{small}

\noindent
where $\dot{N}_{cn}^T(v_j, v_i, MC_p(v_i))$ is the probabilistic upper bound (\cite[Lemma $4.1$]{DBLP:conf/ccs/SunXKYQWY19}) of $\dot{N}_{cn}(v_j, v_i, MC_p(v_i))$.
For each $v_i \in G$ with some noise scale $\lambda_c$ and $\delta_c$, we have

\begin{small}
\begin{equation}
\label{sec:ddp:eq:cn}
\begin{aligned}
    \dot{N}_{cn}^T(v_j, v_i, MC_p(v_i))) & = \dot{N}_{cn}(v_j, v_i, MC_p(v_i)) \\
                                         & + Lap(\lambda_c) + \lambda_c \cdot \log(\frac{1}{2\delta_c}).
\end{aligned}
\end{equation}
\end{small}

Counting common neighbors of two vertices is time-consuming.
For any vertex $v_i \in G$, we have $deg(v_i, MC_p(v_i)) \geq \max_{v_j \in MC_p(v_i))} \dot{N}_{cn}(v_j, v_i, MC_p(v_i))$.
The degree of a vertex can be treated as a loose upper bound for its maximum common neighbors with others.
However, the local sensitivity based on the degree can be too large, leading to adding excessive noises.
Replacing $\max_{v_j \in MC_p(v_i)} \dot{N}_{cn}(v_j, v_i, MC_p(v_i))$ with $deg(v_i, MC_p(v_i))$ for some $v_i$ can balance the efficiency and effectiveness~\cite{DBLP:conf/ccs/SunXKYQWY19}.
For $v_i \in G$ with some $\lambda_d$ and $\delta_d$, we have

\begin{small}
\begin{equation}
\label{sec:ddp:eq:deg}
    deg^T(v_i\!,\!MC_p\!(\!v_i\!))\!=\!deg(v_i\!,\!MC_p\!(v_i\!))\!+\!Lap(\lambda_d)+\!\lambda_d\!\log(\frac{1}{2\delta_d}).
\end{equation}
\end{small}

For vertices with the top $h$ largest $deg^T$ values, they would report ``$\min(\max \dot{N}_{cn}^T(\cdot), deg^T(\cdot))$''~\cite{DBLP:conf/ccs/SunXKYQWY19}.
Let $V_D$ (\resp $V_N$) be the set of vertices reporting $deg^T$ (\resp $\max \dot{N}_{cn}^T$).
We have $u_{d} = \max_{v_i \in V_D} $ $deg^T(v_i,MC_p(v_i))$ and $u_{n} = \max_{v_i \in V_N, v_j \in MC_p(v_i), v_j \neq v_i} \dot{N}_{cn}^T(v_j, v_i, $ $MC_p(v_i))$ and $\widetilde{LS(\Gamma_{S_{in}})} = k \cdot \tbinom{\max(u_{d}, u_{n})}{k-2}$.
As proved in~\cite[Lemma $4.2$]{DBLP:conf/ccs/SunXKYQWY19}, with $\lambda_d = \frac{2}{0.5\varepsilon_1}$, $\lambda_c = \frac{h}{0.5\varepsilon_1}$, and $\delta_d = \delta_c = \frac{\delta}{2h+2}$, the following algorithm satisfies $(\varepsilon_1, \delta)$-\ddp and the returned noise scale  $\lambda = \frac{\widetilde{LS(\Gamma_{S_{in}})}}{\varepsilon_2}$ for Phase-$2$ satisfies $\lambda \geq \frac{LS(\Gamma_{S_{in}})}{\varepsilon_2}$.

\begin{algorithm}[htb]
\small
\caption{\textit{~~Phase-$1$ on \kc} ($G$, $p$, $\varepsilon_1$, $\varepsilon_2$, $\delta$, $h$, $k$)}
\label{alg:phase1:kclique}
\begin{flushleft}
    \mbox{} \quad \textbf{Input} \quad : $G:$ a graph, $p:$ a real value in ($0, 1$),\\
    \mbox{} \qquad \ \ \qquad \quad $\varepsilon_1:$ privacy budget for Phase-$1$, \\
    \mbox{} \qquad \ \ \qquad \quad $\varepsilon_2:$ privacy budget for Phase-$2$, \\
    \mbox{} \qquad \ \ \qquad \quad $\delta:$ invalidation probability, $h:$ an int value,\\
    \mbox{} \qquad \ \ \qquad \quad $k:$ the vertices constraint for \kc \\
    \mbox{} \quad \textbf{Output} : noise scale $\lambda$
\end{flushleft}
\begin{algorithmic}[1]
    \State $\lambda_d := \frac{2}{0.5\varepsilon_1}$; $\delta^\prime := \frac{\delta}{2h+2}$; $\lambda_c := \frac{h}{0.5\varepsilon_1}$ \Comment{Server} \label{alg:phase1:kclique_1}
    \State $U \leftarrow \emptyset$; $\mathcal{R} \leftarrow \emptyset$; \Comment{$\mathcal{R}$: a family of $R$} \label{alg:phase1:kclique_2}
    \For{$v \in G$} \label{alg:phase1:kclique_3}
        \State $R^\prime :=$ \textbf{Expand} $(G,v,p)$ \Comment{Algorithm~\ref{alg:expand}, Client} \label{alg:phase1:kclique_4}
        \State $R :=$ \textbf{Shrink} $(R^\prime, v, p)$ \Comment{Algorithm~\ref{alg:shrink}, Client} \label{alg:phase1:kclique_5}
        \State $\mathcal{R}.push(R)$  \label{alg:phase1:kclique_5.5}
        \State $U(v) := deg(v, R) + Lap(\lambda_d) + \lambda_d \cdot \log(\frac{1}{2\delta^\prime})$ \Comment{Client} \label{alg:phase1:kclique_6}
    \EndFor \label{alg:phase1:kclique_7}
    \State put all $v \in G$ into a maximal priority queue $\mathcal{Q}$ with key $U(v)$ \Comment{Server} \label{alg:phase1:kclique_8}
    \State $T \leftarrow$ the top $h$ elements in $\mathcal{Q}$ \Comment{Server} \label{alg:phase1:kclique_9}

    \For{$v \in T$} \label{alg:phase1:kclique_10}
        \State $R := \mathcal{R}(v)$; $cn := 0$ \label{alg:phase1:kclique_11}
        \State \textbf{for}~$u \in R$~\textbf{do}~$cn := \max(\dot{N}_{cn}(u, v, R), cn)$ \Comment{Client} \label{alg:phase1:kclique_12}
        \State $cn^* := cn + Lap(\lambda_c) + \lambda_c \cdot \log(\frac{1}{2\delta^\prime})$ \Comment{Client} \label{alg:phase1:kclique_15}
        \State $U(v) := \min(cn^*, U(v))$ \Comment{Server} \label{alg:phase1:kclique_16}
    \EndFor \label{alg:phase1:kclique_17}

    \State $\widetilde{LS} := k \cdot \tbinom{\max_{v \in G} U(v)}{k-2}$ \Comment{Server} \label{alg:phase1:kclique_18}
    \State \textbf{Return} $\lambda = \frac{\widetilde{LS}}{\varepsilon_2}$ \Comment{Server} \label{alg:phase1:kclique_return}
\end{algorithmic}
\end{algorithm}

The pseudo-code of Phase-$1$ on $k$-clique counting is given in Algorithm~\ref{alg:phase1:kclique}.
We use set $U$ (Line~\ref{alg:phase1:kclique_2}) to denote the upper bounds for each vertex's maximum common neighbors with all others in its minimal \pc $MC_p(v_i)$.
By running Algorithm~\ref{alg:expand} and Algorithm~\ref{alg:shrink}, for each vertex $v$, we identify the minimal \pc $R$ containing its critical connections, and then store it in a family set $\mathcal{R}$ (Lines~\ref{alg:phase1:kclique_4}-\ref{alg:phase1:kclique_5.5}).
In Line~\ref{alg:phase1:kclique_6}, we compute the loose upper bound based on its degree in $R$.
Vertices in $G$ with the first $h$ largest upper bounds, denoted by $T$, report their tight upper bounds (Lines~\ref{alg:phase1:kclique_8}-\ref{alg:phase1:kclique_9}).
For each $v \in T$, we compute its maximum common neighbors with vertices in its $R$, denoted by $cn$ (Line~\ref{alg:phase1:kclique_12}).
Its upper bound is the minimal value between $cn^*$ and its loose upper bound $U(v)$ (Lines~\ref{alg:phase1:kclique_15}-\ref{alg:phase1:kclique_16}).
Finally, the local sensitivity $\widetilde{LS}$ and noise scale $\lambda$ are generated for Phase-$2$ based on set $U$ (Lines~\ref{alg:phase1:kclique_18}-\ref{alg:phase1:kclique_return}).

%% file: results.tex
\section{Evaluation}
\label{sec:eval}

\subsection{Experimental Setting}
\label{sec:eval:set}
\vspace{1mm}
\noindent \textbf{Datasets}.
We used $9$ real-life graph datasets,
as described in Table~\ref{tb:datasets}. The \texttt{Celegans} and \texttt{WIKIVote} datasets are from~\cite{nr} and the remaining ones are from~\cite{konect}. The \texttt{Celegans} and \texttt{Yeast} are two biological graphs. \texttt{Celegans} (\resp \texttt{Yeast}) contains the metabolic reactions (\resp interactions) between two substrates (\resp proteins). \texttt{WIKIVote} contains all the Wikipedia voting data from the inception of Wikipedia. \texttt{USAirport} contains the flights between two US airports. \texttt{Bitcoin} is a user–user trust/distrust network from the Bitcoin Alpha platform. \texttt{Gnutella08} contains the connections between two Gnutella hosts from August $08$, $2022$. {\revision \texttt{HepTh} (\resp \texttt{HepPh}) is a co-authorship graph from arXiv, which contains the collaborations between authors who submit their papers to High Energy Physics - Theory (\resp Phenomenology).} \texttt{SisterCity} contains the ``Sister City'' relationships between two cities of the world.
Directed edges are transferred to undirected edges. 
\begin{table}[htb]
\small
  \centering
  \vspace{3mm}
  \caption{Statistics of Datasets}
    {
    \begin{tabular}{|p{0.29\columnwidth}|p{0.11\columnwidth}|p{0.13\columnwidth}|p{0.08\columnwidth}|p{0.1\columnwidth}|}
      \hline
      \textbf{Dataset}  & \textbf{Nodes}  & \textbf{Edges} & $d_{avg}$ & $d_{max}$ \\ \hline \hline
      \texttt{Celegans}~\cite{nr}  &  453 & 2,025 & 8.94 & 237\\ \hline
      \texttt{WIKIVote}~\cite{nr}  &  889 & 2,914 & 6.56 & 102\\ \hline
      \texttt{USAirport}~\cite{konect}  &  1,574 & 17,215 & 21.87 & 314\\ \hline
      \texttt{Yeast}~\cite{konect}  &  1,876 & 2,203 & 2.39 & 56\\ \hline
      \texttt{Bitcoin}~\cite{konect}  &  3,783 & 14,214 & 7.47 & 511\\ \hline
      \texttt{Gnutella08}~\cite{konect}  &  6,301 & 20,777 & 6.59 & 97\\ \hline
      \texttt{HepTh}~\cite{konect}  &  9,875 & 25,973 & 5.26 & 65\\ \hline
      \texttt{SisterCity}~\cite{konect}  &  14,274 & 20,573 & 2.88 & 99\\ \hline
      \texttt{\textcolor{black}{HepPh}}~\cite{konect}  & \textcolor{black}{34,546} & \textcolor{black}{420,877} & \textcolor{black}{24.37} & \textcolor{black}{846}\\ \hline
    \end{tabular}
    }
\label{tb:datasets}
\end{table}

\vspace{1mm}
\noindent \textbf{Algorithms}.
To the best of our knowledge, no existing work has attempted to protect the privacy of critical connections of a vertex in a graph. We implement our algorithms for \pc and minimal \pc computation with new score functions and evaluate the effectiveness for privacy preservation under \ddp. 

\vspace{1mm}
\noindent \textbf{Settings}.
For \pc computation, in our experiments, $p$ ranges from $0.1$ to $0.8$ with the default value of $0.1$.
For privacy protection of critical connections, we count \kcs under $(\varepsilon,\delta)$-DDP, utilizing the two-phase framework.
Each result is averaged over $100$ runs.
The privacy budget $\varepsilon$ ranges from $1$ to $12$.
The default value of privacy budget $\varepsilon_1$ (\resp $\varepsilon_2$) for Phase-$1$ (\resp Phase-$2$) is $\varepsilon_1 = 0.1\varepsilon$ (\resp $\varepsilon_2 = 0.9\varepsilon$).
The $h$ (Line~\ref{alg:phase1:kclique_9} of Algorithm~\ref{alg:phase1:kclique}) varies from $1$ to $12$.
Following the setting in~\cite{DBLP:conf/ccs/SunXKYQWY19,DBLP:journals/fttcs/DworkR14}, the default value for $\delta$ is $\frac{1}{n}$, where $n$ is the vertices number of the graph.
All programs are implemented in Python on a Linux machine.

\begin{figure*}[!ht]
\begin{center}
    \begin{subfigure}{0.34\textwidth} 
    \centering
        \includegraphics[width=\textwidth]{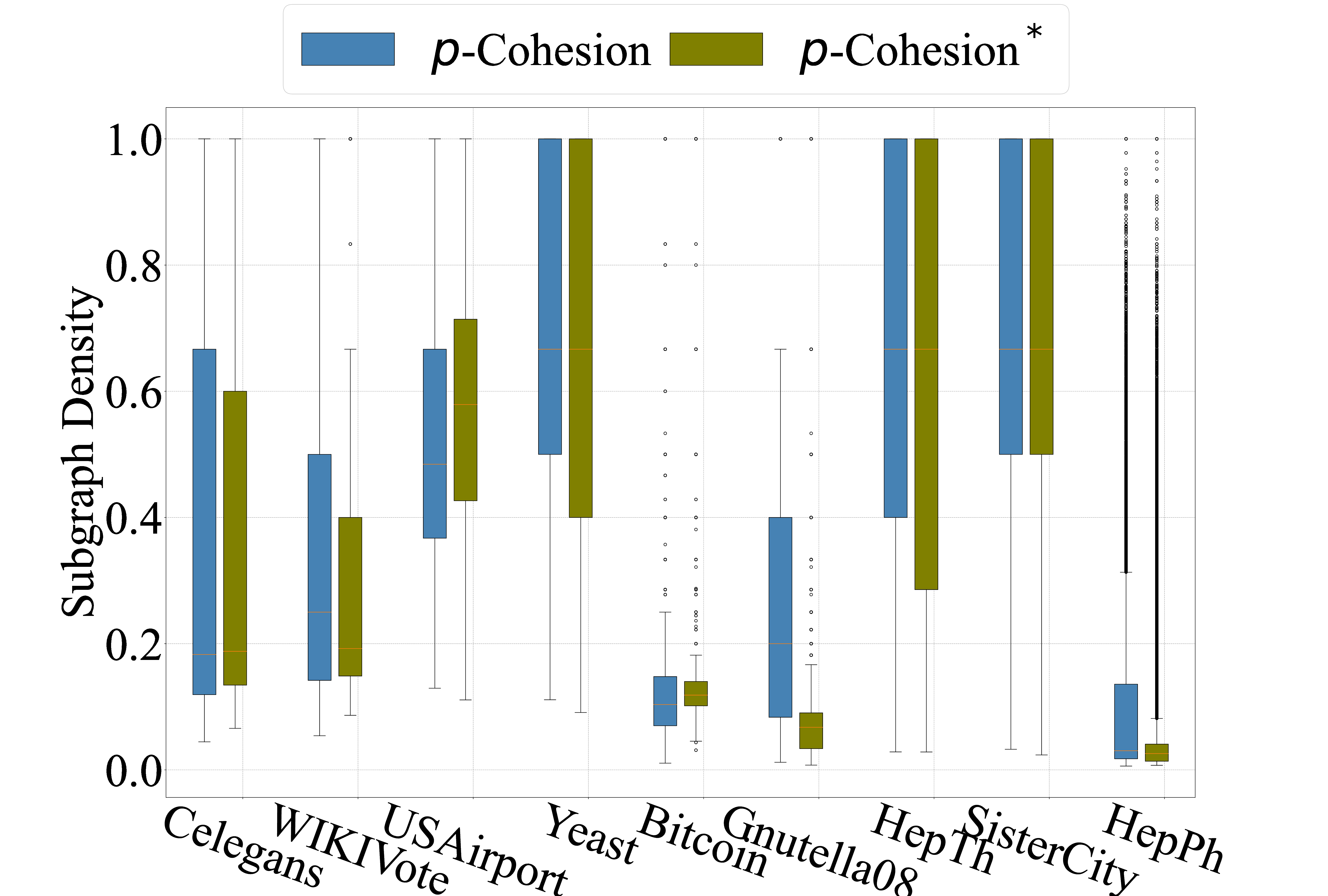}
        \caption{{\revision All Datasets, $p=0.3$}}
        \label{fig:eval:PCsDensity:all}
      \vspace{4mm}
    \end{subfigure}
    \hfill
    \begin{subfigure}{0.325\textwidth}
    \centering
        \includegraphics[width=\textwidth]{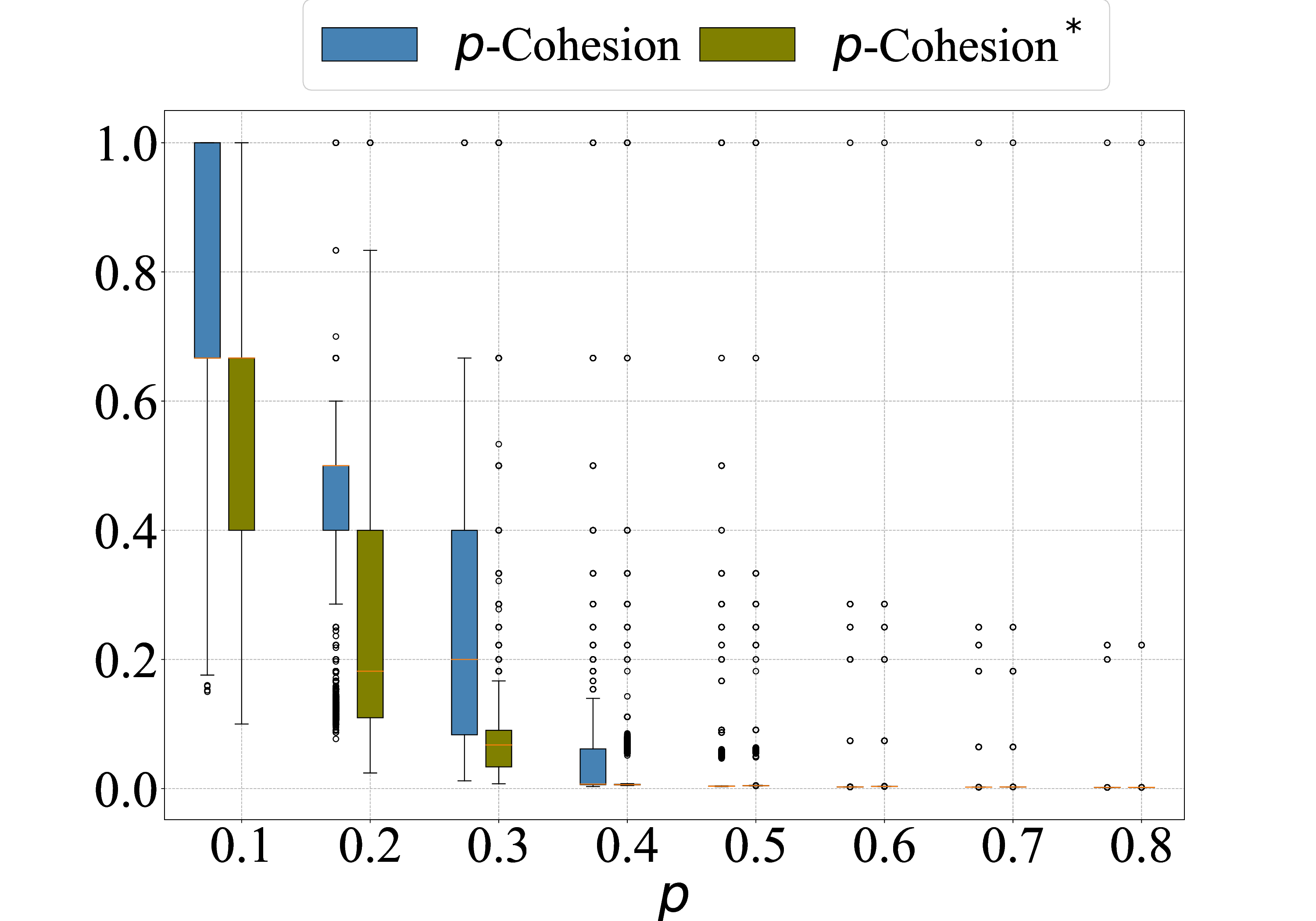}
        \caption{Gnutella08}
        \label{fig:eval:PCsDensity:gnutella08}
        \vspace{4mm}
    \end{subfigure}
    \hfill
    \begin{subfigure}{0.325\textwidth}
    \centering
        \includegraphics[width=\textwidth]{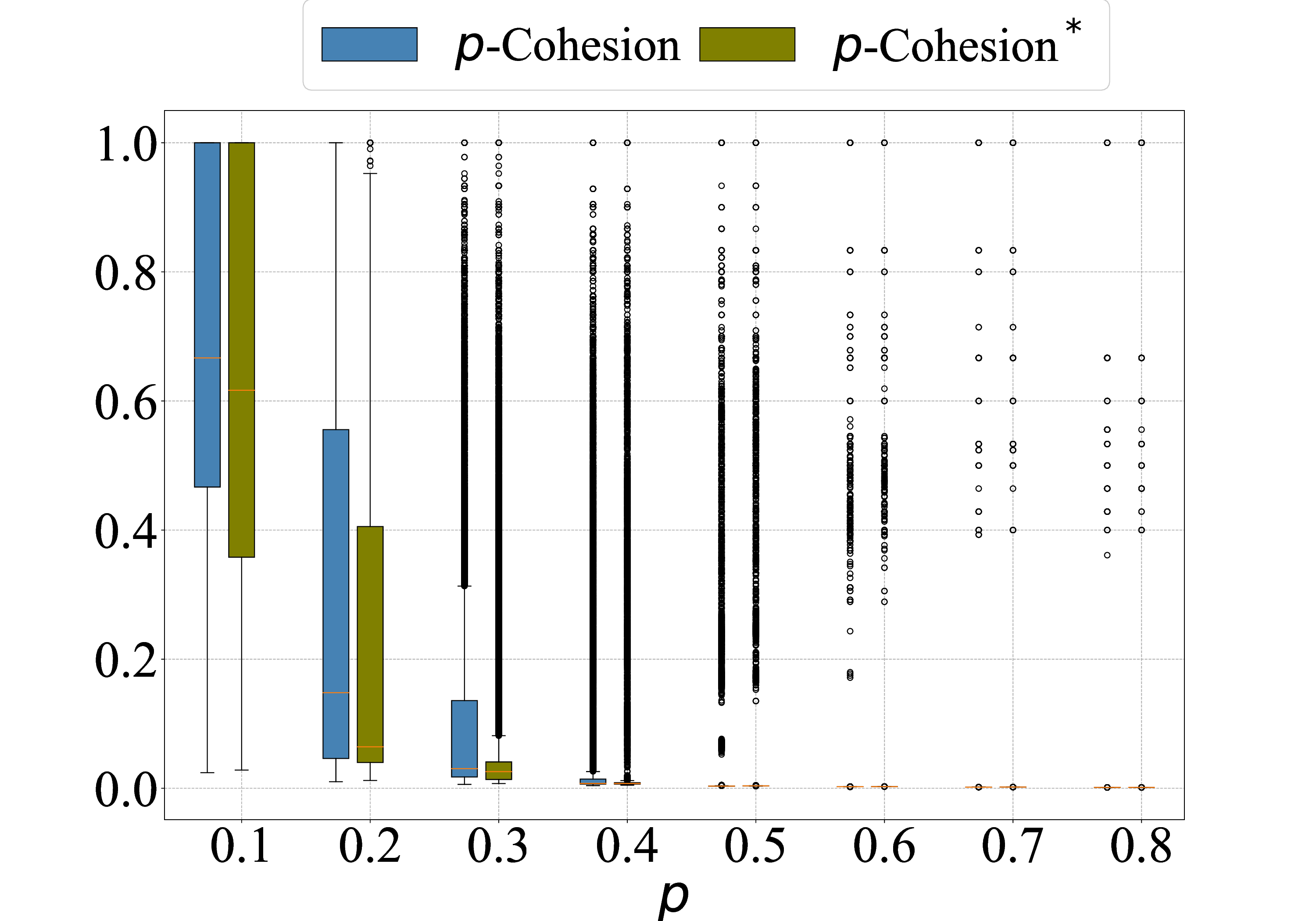}
        \caption{{\revision HepPh}}
        \label{fig:eval:PCsDensity:hepth}
        \vspace{4mm}
    \end{subfigure}
    \hfill
    \vspace{-6mm}
    \caption{\small{Density Distribution of Minimal $p$-Cohesions under Different Score Functions}}
    \label{fig:eval:PCsDensity}
\end{center}
\end{figure*}

\begin{figure*}[!ht]
\centering
    \begin{subfigure}{0.34\textwidth} 
    \centering
        \includegraphics[width=\textwidth]{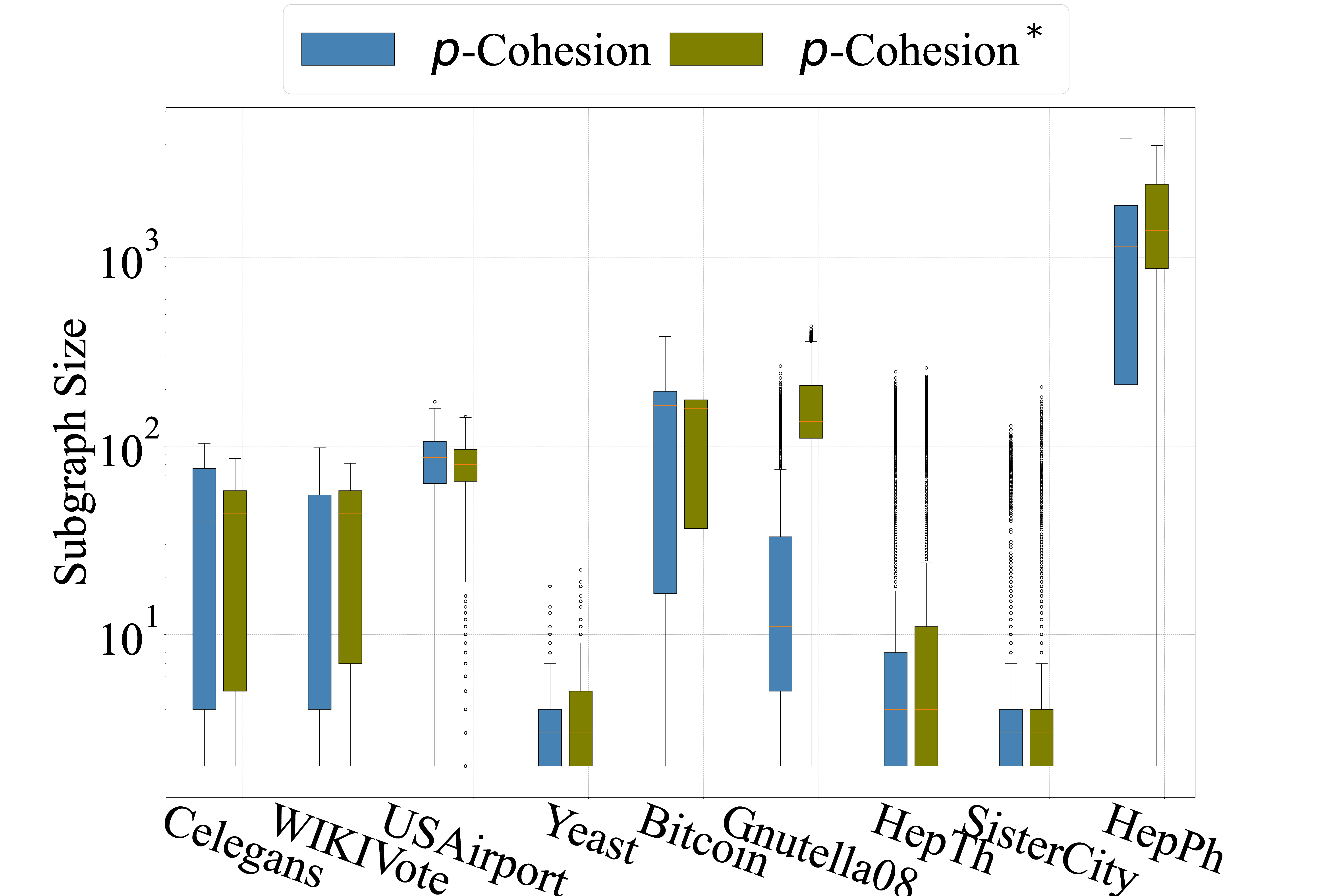}
        \caption{{\revision All Datasets, $p=0.3$}}
        \label{fig:eval:PCsSize:all}
      \vspace{4mm}
    \end{subfigure}
    \hfill
    \begin{subfigure}{0.325\textwidth}
    \centering
        \includegraphics[width=\textwidth]{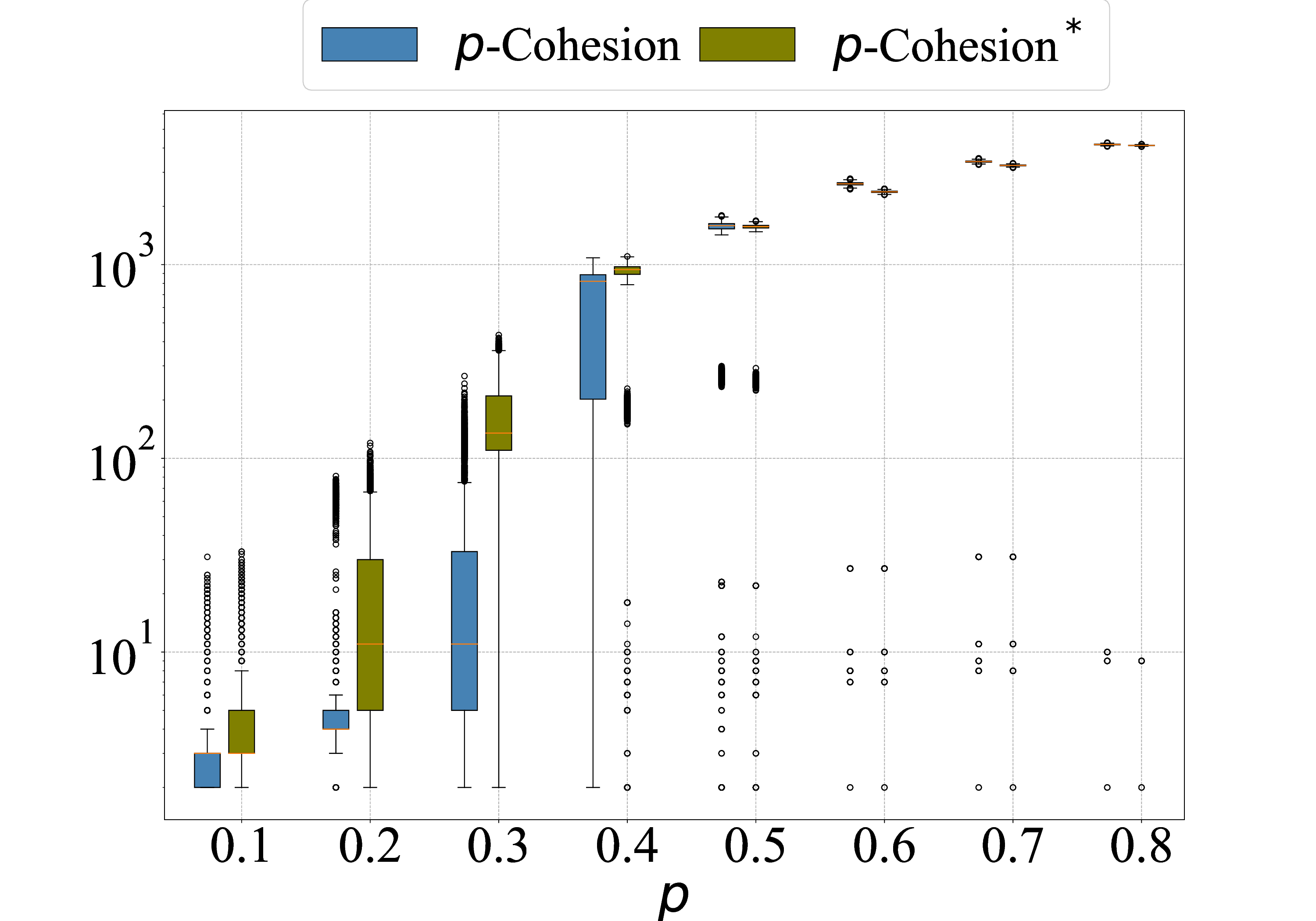}
        \caption{Gnutella08}
        \label{fig:eval:PCsSize:usairport}
        \vspace{4mm}
    \end{subfigure}
    \hfill
    \begin{subfigure}{0.325\textwidth}
    \centering
        \includegraphics[width=\textwidth]{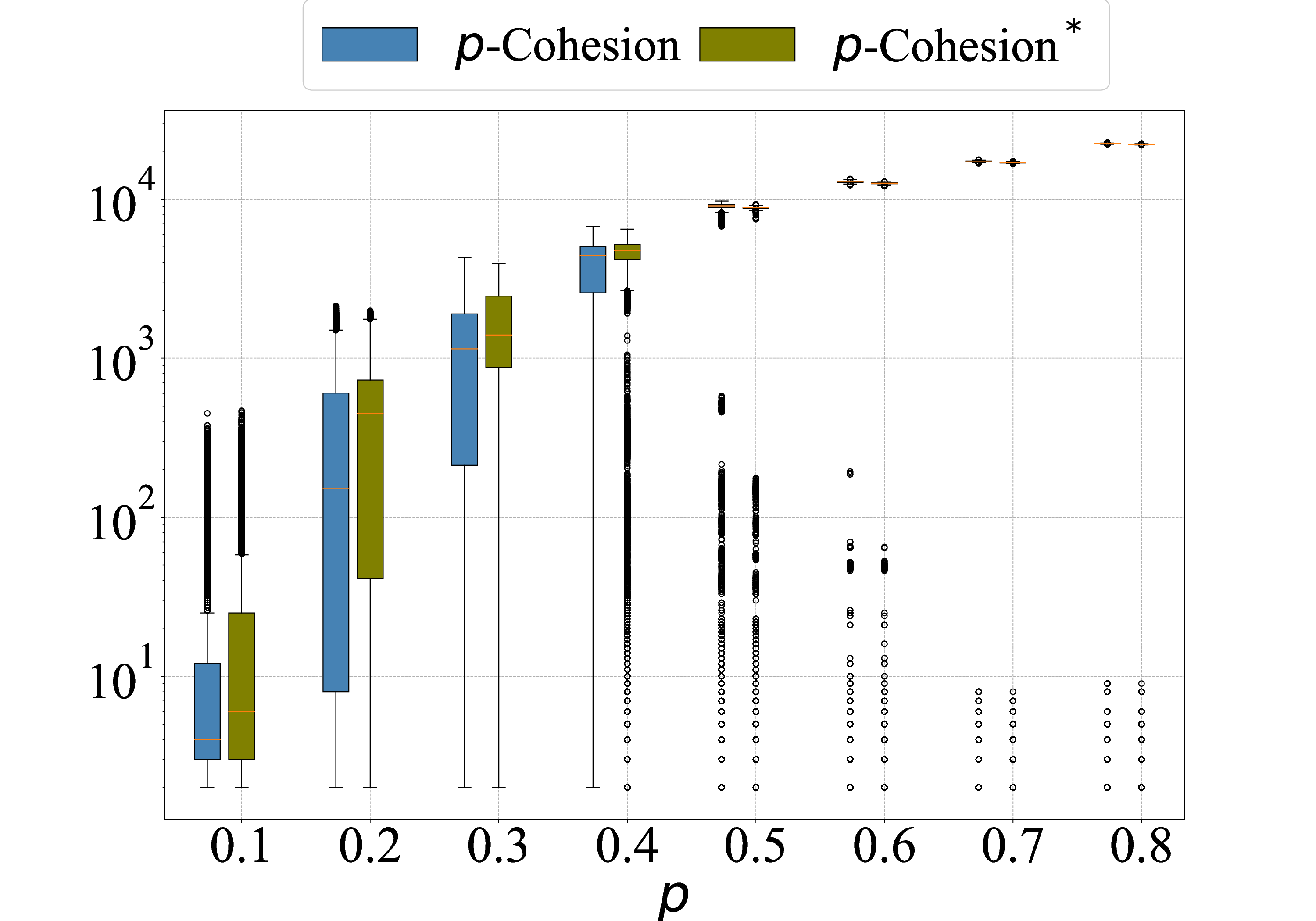}
        \caption{{\revision HepPh}}
        \label{fig:eval:PCsSize:sistercity}
        \vspace{4mm}
    \end{subfigure}
    \hfill
    \vspace{-6mm}
    \caption{\small{Size Distribution of Minimal $p$-Cohesions under Different Score Functions}}
    \label{fig:eval:PCsSize}
\end{figure*}
\begin{figure*}[!ht]
\centering
    \begin{subfigure}{0.34\textwidth} 
    \centering
        \includegraphics[width=\textwidth]{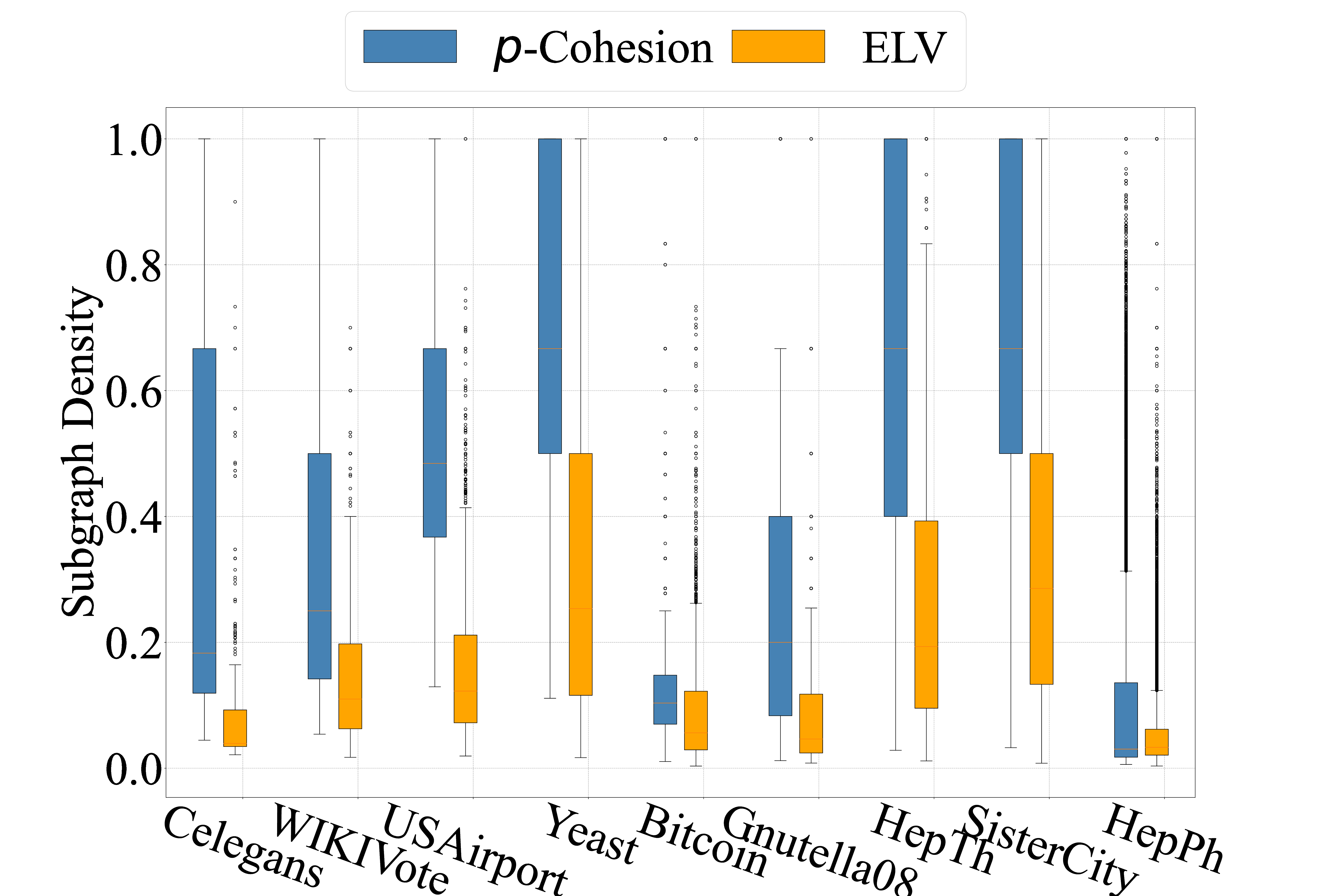}
        \caption{{\revision All Datasets, $p=0.3$}}
        \label{fig:eval:PCELVDen:all}
      \vspace{4mm}
    \end{subfigure}
    \hfill
    \begin{subfigure}{0.325\textwidth}
    \centering
        \includegraphics[width=\textwidth]{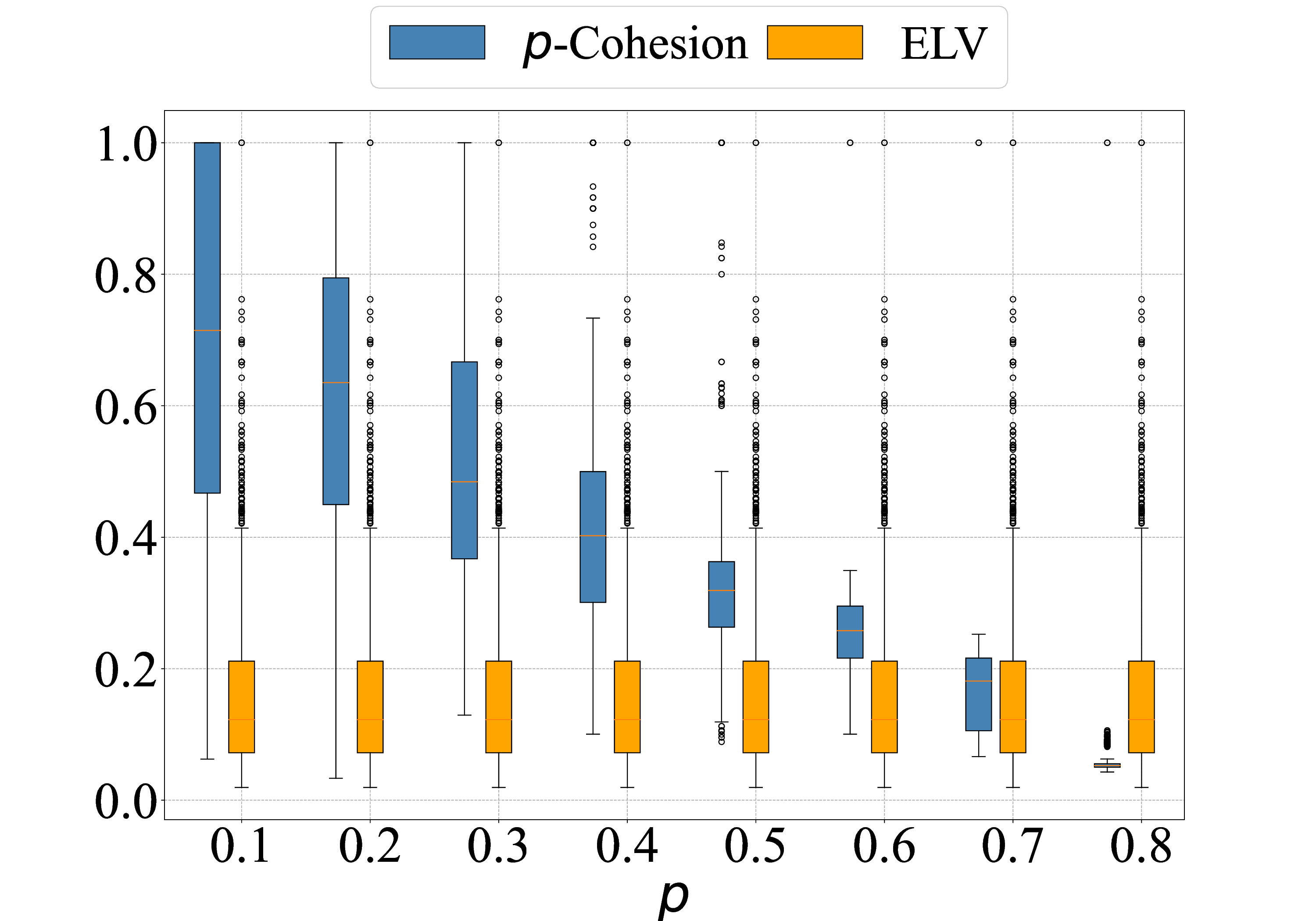}
        \caption{USAirport}
        \label{fig:eval:PCELVDen:usairport}
        \vspace{4mm}
    \end{subfigure}
    \hfill
    \begin{subfigure}{0.325\textwidth}
    \centering
        \includegraphics[width=\textwidth]{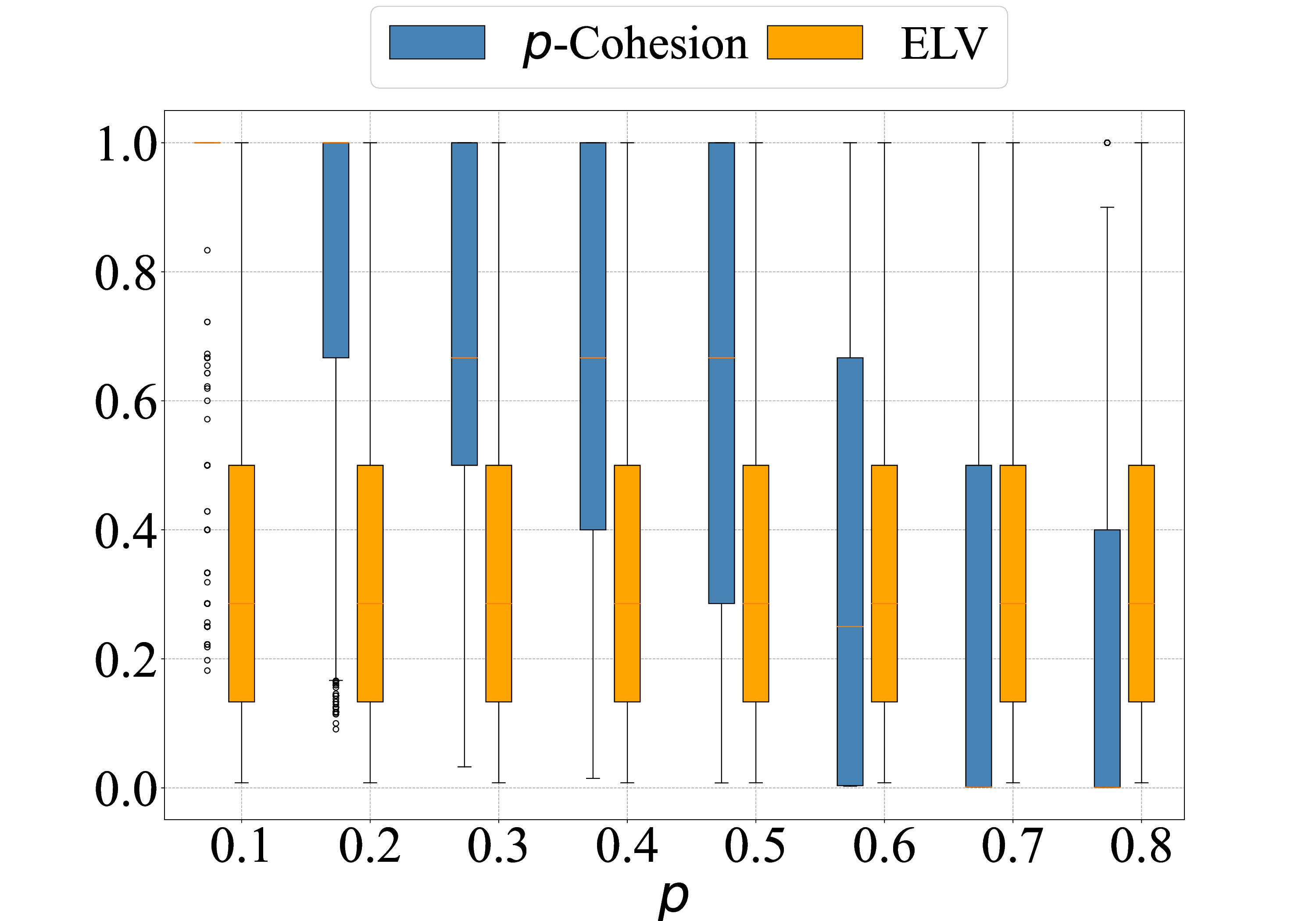}
        \caption{SisterCity}
        \label{fig:eval:PCELVDen:sistercity}
        \vspace{4mm}
    \end{subfigure}
    \hfill
    \vspace{-6mm}
    \caption{\small{Density Distribution of Minimal $p$-Cohesion and ELV}}
    \label{fig:eval:PCELVDen}
\end{figure*}
\begin{figure*}[!ht]
\centering
    \begin{subfigure}{0.34\textwidth} 
    \centering
        \includegraphics[width=\textwidth]{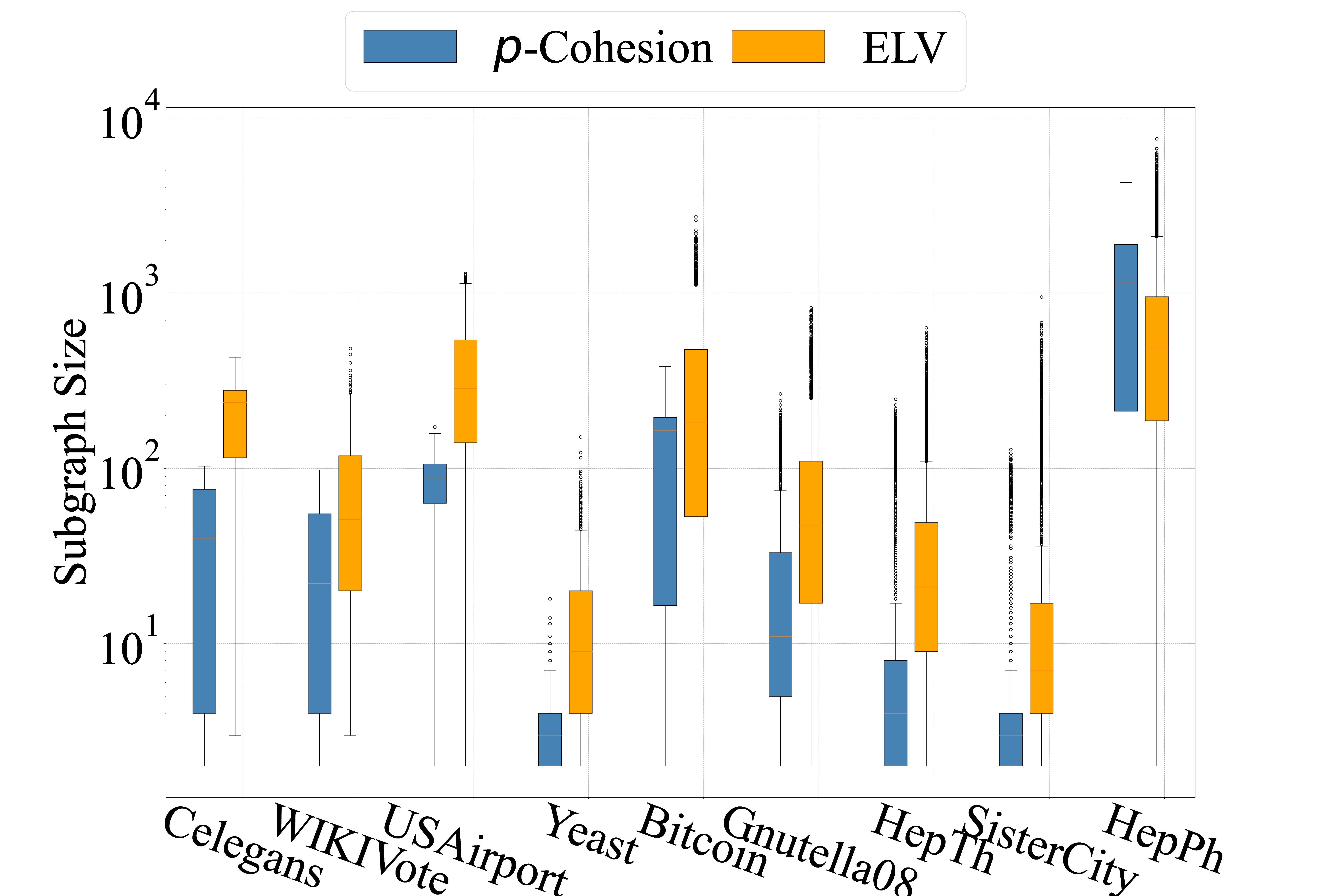}
        \caption{{\revision All Datasets, $p=0.3$}}
        \label{fig:eval:PCELVSize:all}
      \vspace{4mm}
    \end{subfigure}
    \hfill
    \begin{subfigure}{0.325\textwidth}
    \centering
        \includegraphics[width=\textwidth]{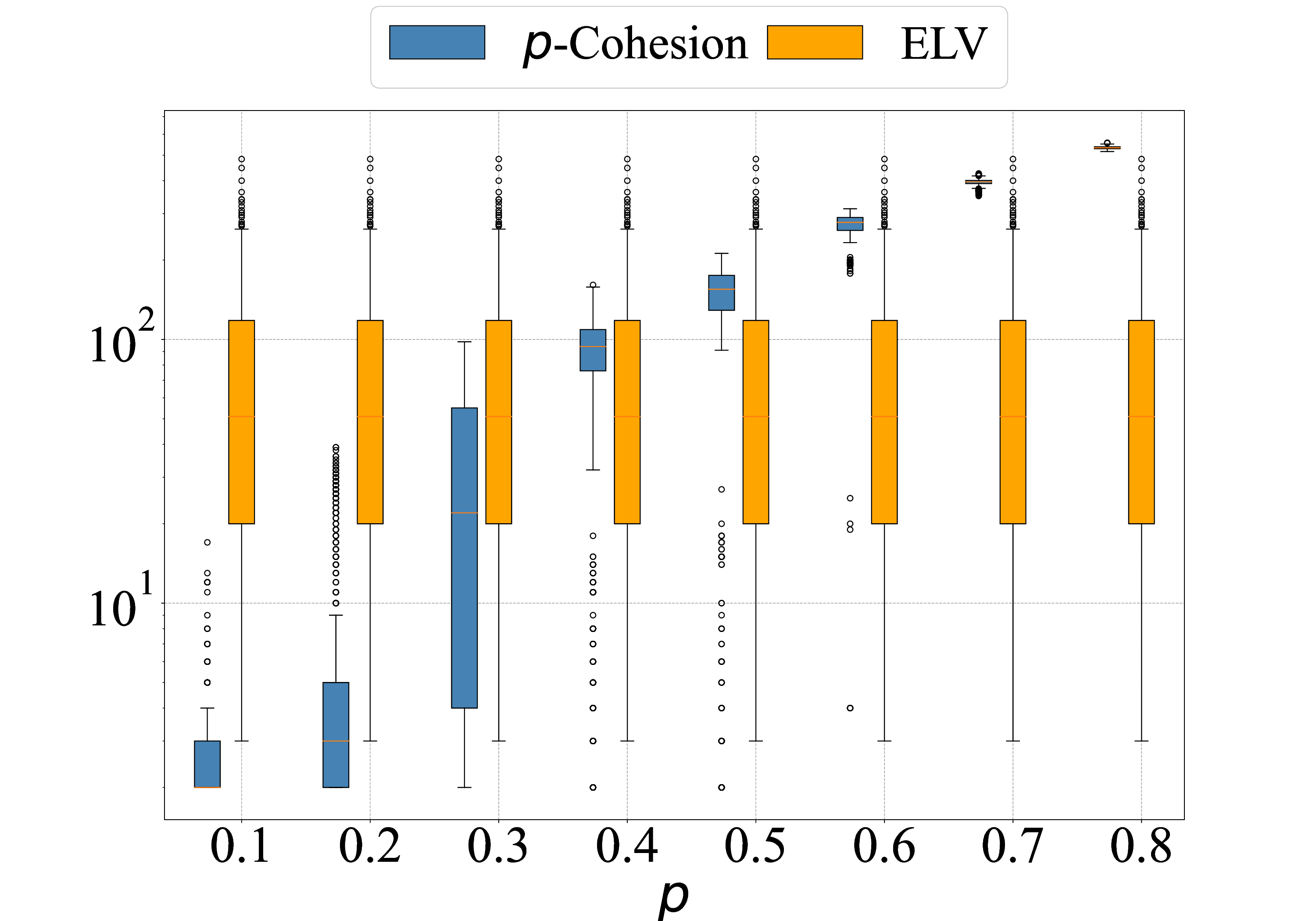}
        \caption{WIKIVote}
        \label{fig:eval:PCELVSize:wikivote}
        \vspace{4mm}
    \end{subfigure}
    \hfill
    \begin{subfigure}{0.325\textwidth}
    \centering
        \includegraphics[width=\textwidth]{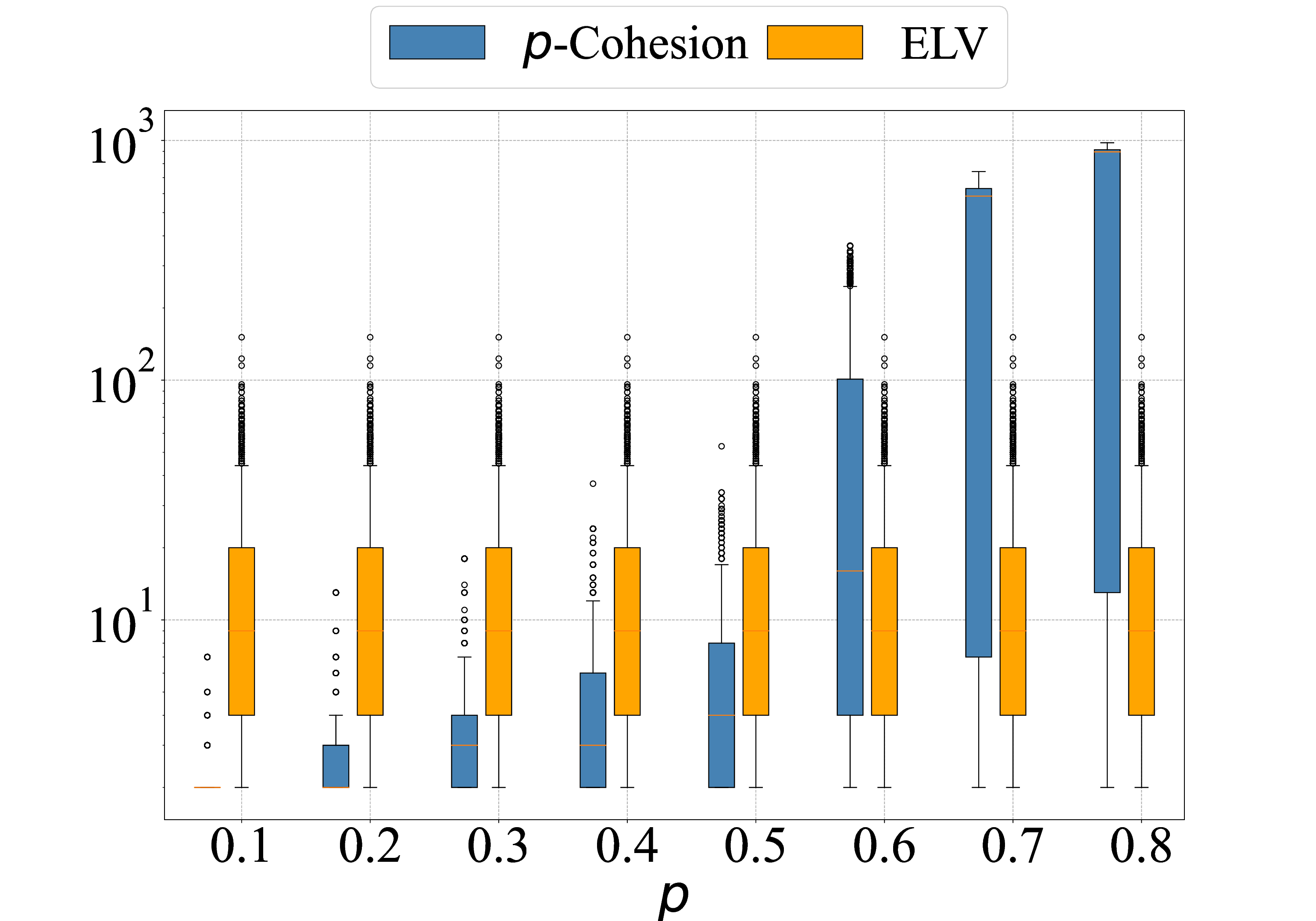}
        \caption{Yeast}
        \label{fig:eval:PCELVSize:yeast}
        \vspace{4mm}
    \end{subfigure}
    \hfill
    \vspace{-6mm}
    \caption{\small{Size Distribution of Minimal $p$-Cohesion and ELV}}
    \label{fig:eval:PCELVSize}
\end{figure*}
\subsection{Statistical Evaluation for Minimal $p$-Cohesions}
\label{sec:eval:stat}

\subsubsection{Minimal $p$-Cohesions with Different Score Functions}

For each vertex $v$, we identify its minimal \pc with Algorithms~\ref{alg:expand} and~\ref{alg:shrink}.
Different score functions (Line~\ref{alg:expand_3} of Algorithm~\ref{alg:expand}) result in different \pcs and minimal \pcs.
In the following, we show the performance of Algorithms~\ref{alg:expand} and \ref{alg:shrink} in both density and size.
\pcalg denotes Algorithm~\ref{alg:expand} running~(\ref{sec:mpc:eq:fw}) followed by Algorithm~\ref{alg:shrink}.
\pcalgOld replaces~(\ref{sec:mpc:eq:fw}) in \pcalg with~\cite[Eq.~(7)]{DBLP:journals/kais/LiZZQZL21}.

\vspace{1mm}
\noindent \underline{Density}.
Fig.~\ref{fig:eval:PCsDensity} reports the graph density distribution of minimal \pcs with different score functions.
\pcalg uses Eq.~(\ref{sec:mpc:eq:merit}) and Eq.~(\ref{sec:mpc:eq:penalty}), and \pcalgOld uses~\cite[Eqs.~(5) \& (6)]{DBLP:journals/kais/LiZZQZL21}.
Fig.~\ref{fig:eval:PCsDensity} displays the result over $9$ datasets at $p = 0.3$. As expected, the \pcalg outperforms \pcalgOld, because both the merit and penalty consider the advantage and disadvantage when including a new vertex $w$ to $V_p$.
For example, on \textit{Yeast} and \textit{HepTh}, all minimal \pcs returned by \pcalg are with larger density than the minimal \pcs returned by \pcalgOld.
{\revision
Figs.~\ref{fig:eval:PCsDensity:gnutella08} and~\ref{fig:eval:PCsDensity:hepth} report the result of two different score functions, with varying values of p ranging from $0.1$ to $0.8$ for \textit{Gnutella08} and \textit{HepPh}, respectively.
Both figures illustrate the density distributions of the minimal \pcs returned by \pcalg and \pcalgOld become similar as $p$ increases, since higher values of $p$ inherently entail a greater number of vertices within a \pc subgraph. Consequently, both \pcalg and \pcalgOld tend to yield comparable vertex sets.
As $p$ reaches large magnitudes, the results exhibit similarity due to the near-complete inclusion of the connected component.
}

\vspace{1mm}
\noindent \underline{Size}.
We show size distributions of all minimal \pcs returned by \pcalg and \pcalgOld in Fig.~\ref{fig:eval:PCsSize}.
Fig.~\ref{fig:eval:PCsSize:all} represents the result for $9$ datasets at $p = 0.3$. On most of the datasets, the minimal \pcs returned by \pcalg have similar sizes to \pcalgOld. On some graphs, \ie, \textit{Yeast} and \textit{HepTh}, our methods can find much smaller sizes minimal \pcs.
{\revision
Figs.~\ref{fig:eval:PCsSize:usairport} and~\ref{fig:eval:PCsSize:sistercity} report the size distributions with the two different score functions on \textit{Gnutella08} and \textit{HepPh}, respectively.
Both figures show that the sizes of minimal \pcs returned by the two methods increase as $p$ grows. This is because a larger value $p$ necessitates a greater number of vertices in a $p$-cohesion.
}

\subsubsection{Minimal $p$-Cohesion v.s. Extended Loccal View}

{\color{black}
Haipei~\etal~\cite{DBLP:conf/ccs/SunXKYQWY19} used ELV as the local view for subgraph counting. But they could not capture critical connections.
We denote the connections identified by their two-hop ELV~\cite[Definition $2.2$]{DBLP:conf/ccs/SunXKYQWY19} as \elvalg.}
From the distributions of density and size, we compare the difference between minimal \pc and ELV.
Here, \pcalg denotes the Algorithm~\ref{alg:expand} equipped with Eq.~(\ref{sec:mpc:eq:fw}), followed by Algorithm~\ref{alg:shrink}.

\vspace{1mm}
\noindent \underline{Density}.
Fig.~\ref{fig:eval:PCELVDen} reports the density distributions comparison between \pcalg and \elvalg.
Fig.~\ref{fig:eval:PCELVDen:all} displays the result across $9$ datasets at $p = 0.3$.
The \pcalg significantly outperforms \elvalg in the subgraph density because the \pcalg can identify a subgraph induced by the critical connections of a vertex. 
For example, on \textit{USAirport}, all vertices' subgraphs returned by \pcalg have larger densities than ELV.
Figs.~\ref{fig:eval:PCELVDen:usairport} and~\ref{fig:eval:PCELVDen:sistercity} depict the result of distinct methods, showcasing variations in $p$ values ranging from $0.1$ to $0.8$ for \textit{USAirport} and \textit{SisterCity}, respectively.
Both figures illustrate that when $0.1 \leq p \leq 0.7$, many vertices' minimal \pcs are with larger density than ELV.
The density of \pcalg decreases as $p$ increases, as a larger value of $p$ inherently necessitates a greater number of vertices in a p-cohesion.
{\color{black} However, when $p > 0.7$, the minimal \pcs of most vertices exhibit lower densities compared to their ELVs.
This is because a large $p$ value results in the retrieval of nearly the entire graph.
}

\vspace{1mm}
\noindent \underline{Size}.
In Fig.~\ref{fig:eval:PCELVSize}, we report the size distributions of \pcalg and \elvalg.
Fig.~\ref{fig:eval:PCELVSize:all} reports the result over $9$ datasets with $p = 0.3$. On all datasets and all vertices, all minimal \pcs returned by \pcalg are with smaller sizes than \elvalg.
Figs.~\ref{fig:eval:PCELVSize:wikivote} and~\ref{fig:eval:PCELVSize:yeast} report the size distributions of the two methods on \textit{WIKIVote} and \textit{Yeast}, as $p$ increases from $0.1$ to $0.8$.
Similarly, within both figures, the sizes of minimal \pcs yielded by \pcalg exhibit growth as $p$ increases.
This phenomenon arises due to the greater number of vertices required for a larger $p$ in a \pc.
When $p$ is large enough, \ie, $p > 0.5$ (\resp $p > 0.6$) on \textit{WIKIVote} (\resp \textit{Yeast}), the \pcalg will return larger subgraphs than \elvalg since the minimal \pc needs more vertices to meet the $p$ constraint.

\begin{figure}[htb]
\centering
    \begin{subfigure}{0.36\columnwidth}
    \vspace{-2mm}
    \centering
        \includegraphics[width=\hsize]{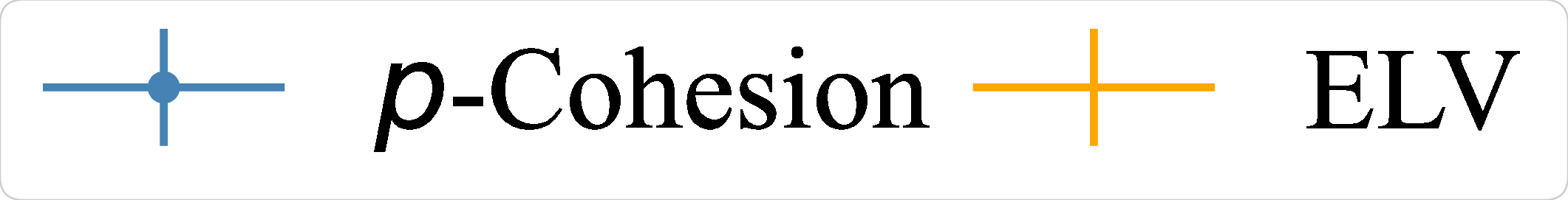}
    \end{subfigure}
    \hfill

  \begin{subfigure}{0.45\columnwidth}
  \centering
        \includegraphics[width=\hsize]{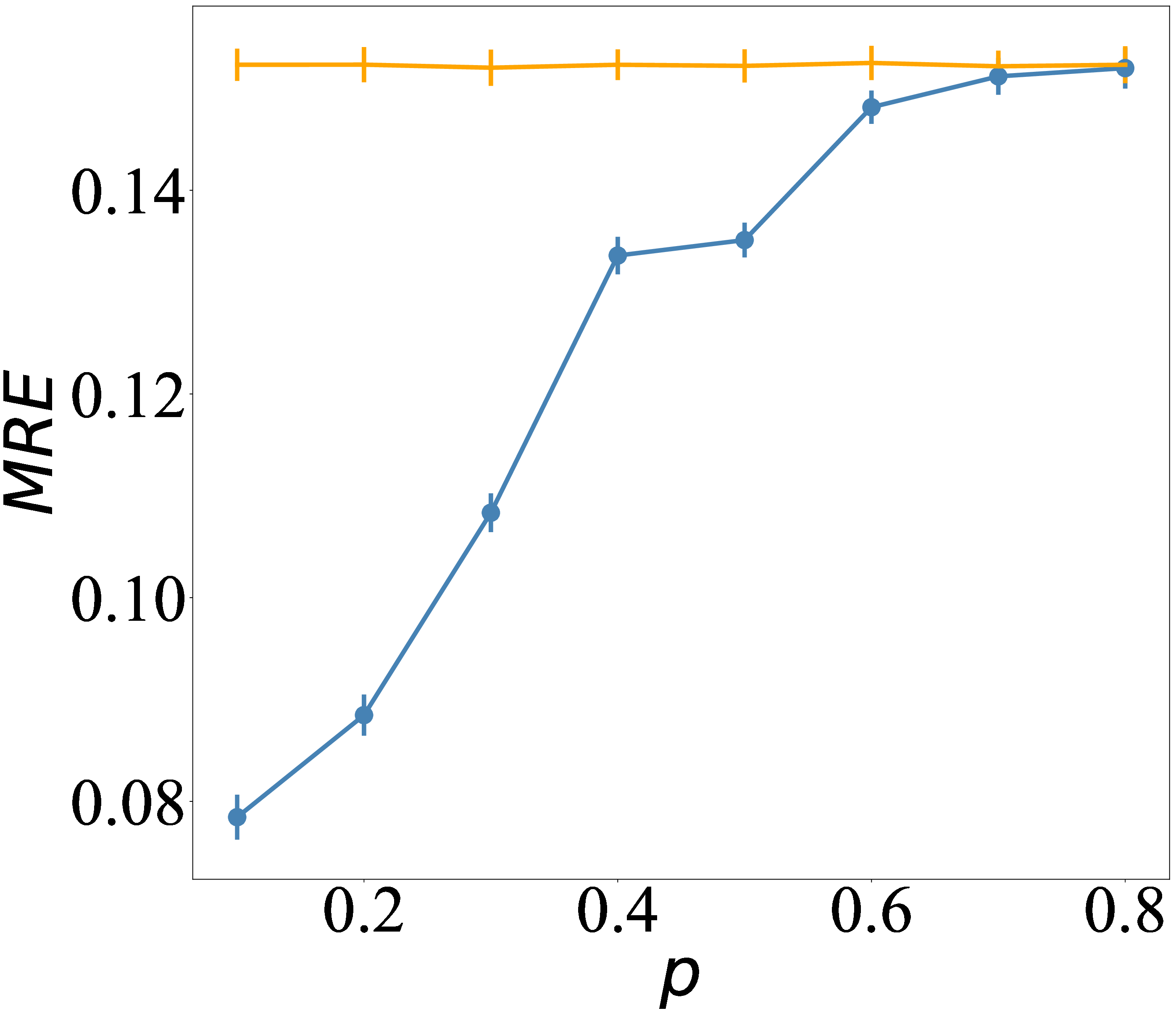}
        \caption{$h = 1$, $\varepsilon = 3$ ($\varepsilon_1 = 0.3$)}
        \label{fig:eval:hepth:h1v2}
        \vspace{4mm}
    \end{subfigure}
    \hfill
    \begin{subfigure}{0.45\columnwidth}
    \centering
        \vspace{3mm}
        \includegraphics[width=\hsize]{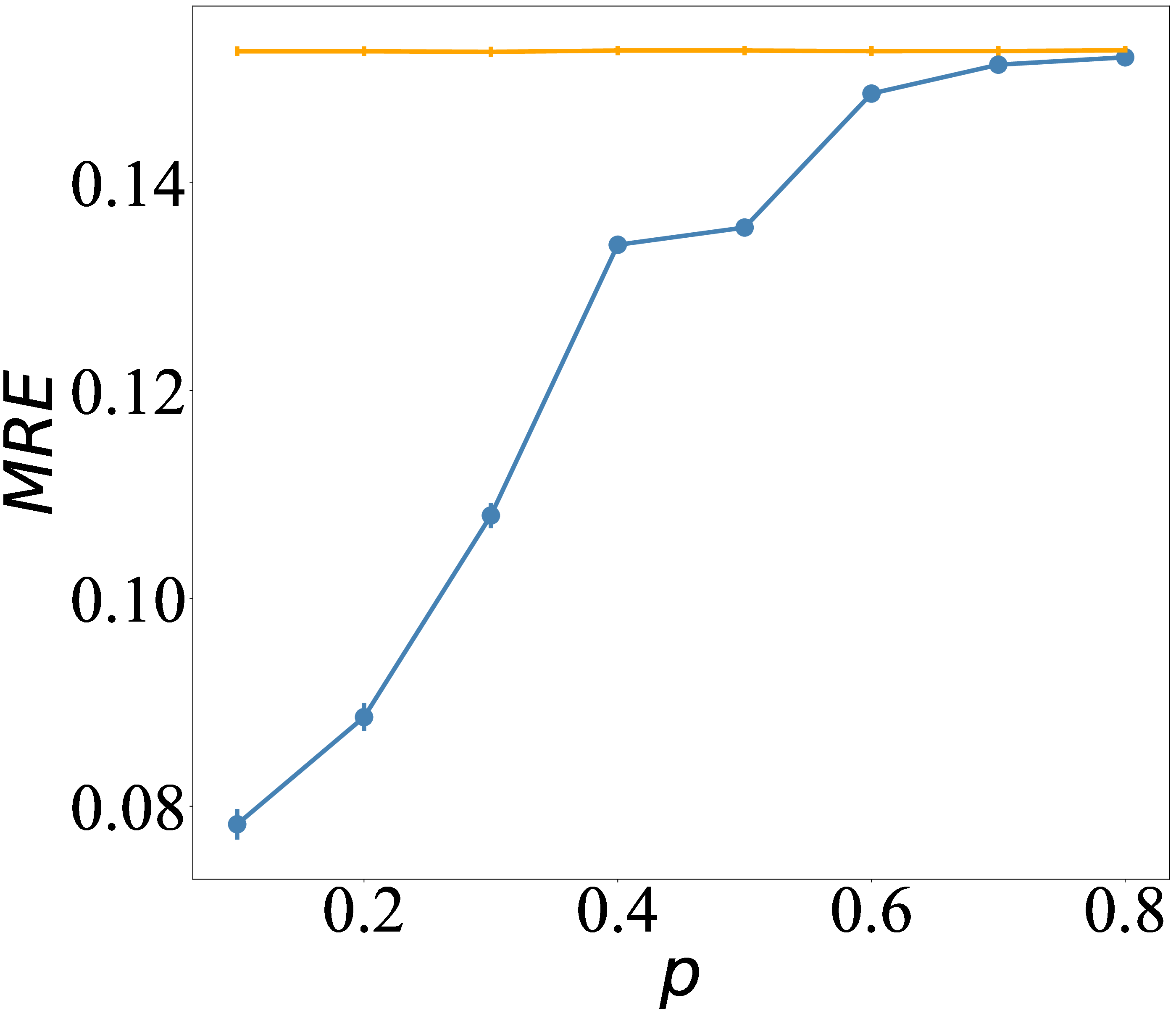}
        \caption{$h = 1$, $\varepsilon = 10$ ($\varepsilon_1 = 1.0$)}
        \label{fig:eval:hepth:h1v10}
        \vspace{4mm}
    \end{subfigure}
    \hfill
    \begin{subfigure}{0.45\columnwidth}
    \centering
        \includegraphics[width=\hsize]{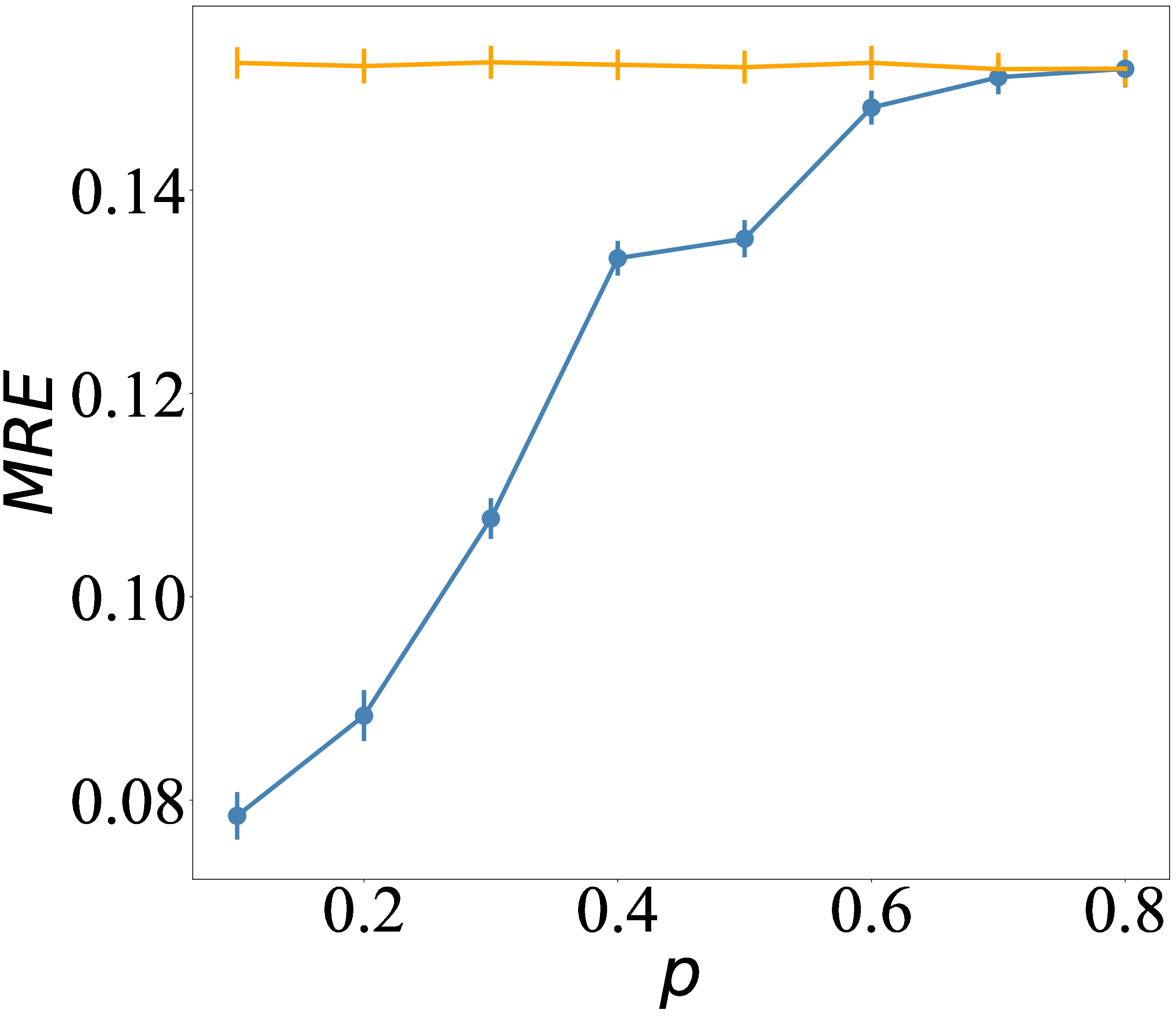}
        \caption{$h = 3$, $\varepsilon = 3$ ($\varepsilon_1 = 0.3$)}
        \label{fig:eval:hepth:h3v2}
        \vspace{4mm}
    \end{subfigure}
    \hfill
    \begin{subfigure}{0.45\columnwidth}
    \centering
        \vspace{3mm}
        \includegraphics[width=\hsize]{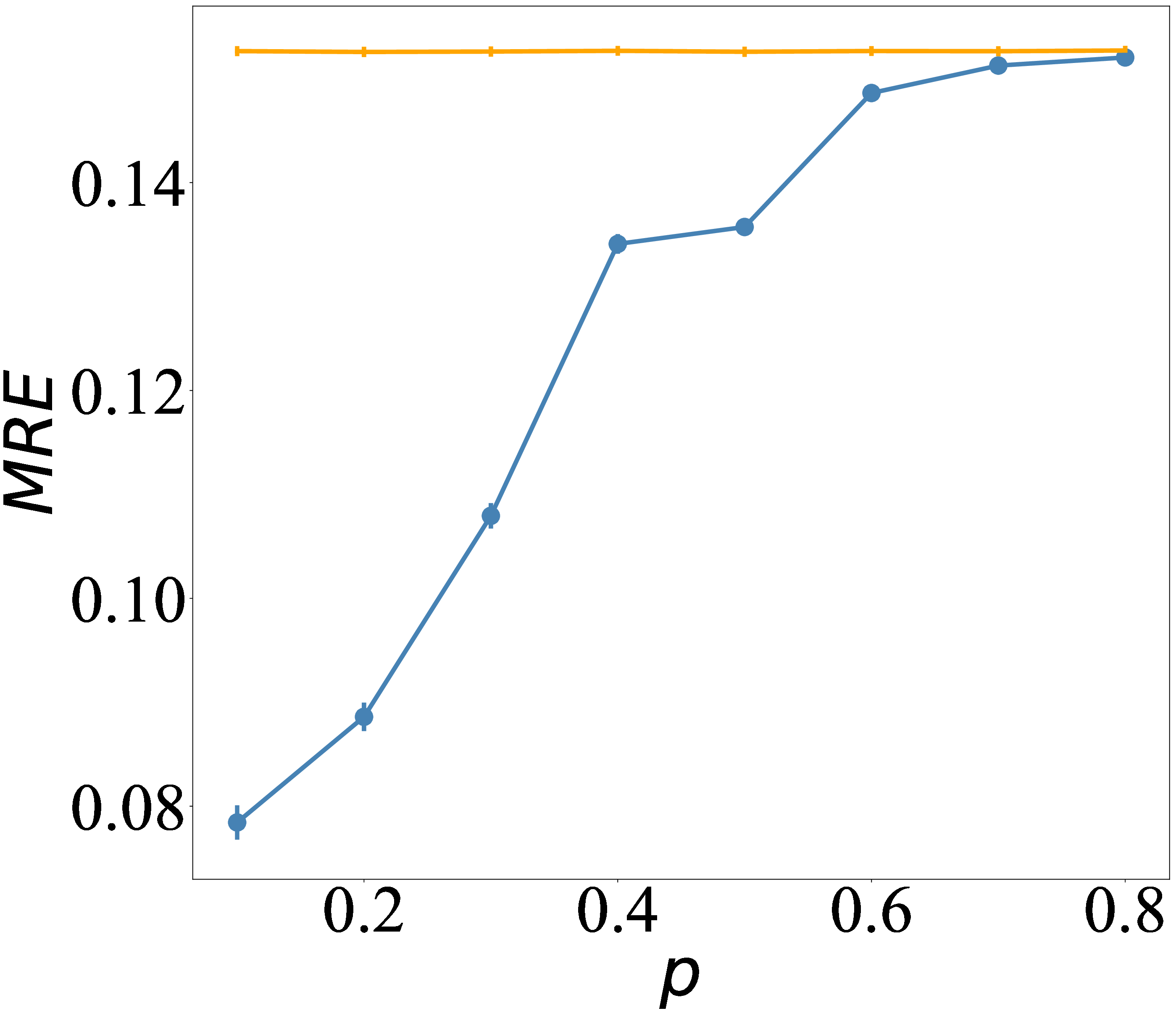}
        \caption{$h = 3$, $\varepsilon = 10$ ($\varepsilon_1 = 1.0$)}
        \label{fig:eval:hepth:h3v10}
        \vspace{4mm}
    \end{subfigure}
    \hfill
    \begin{subfigure}{0.45\columnwidth}
    \centering
        \includegraphics[width=\hsize]{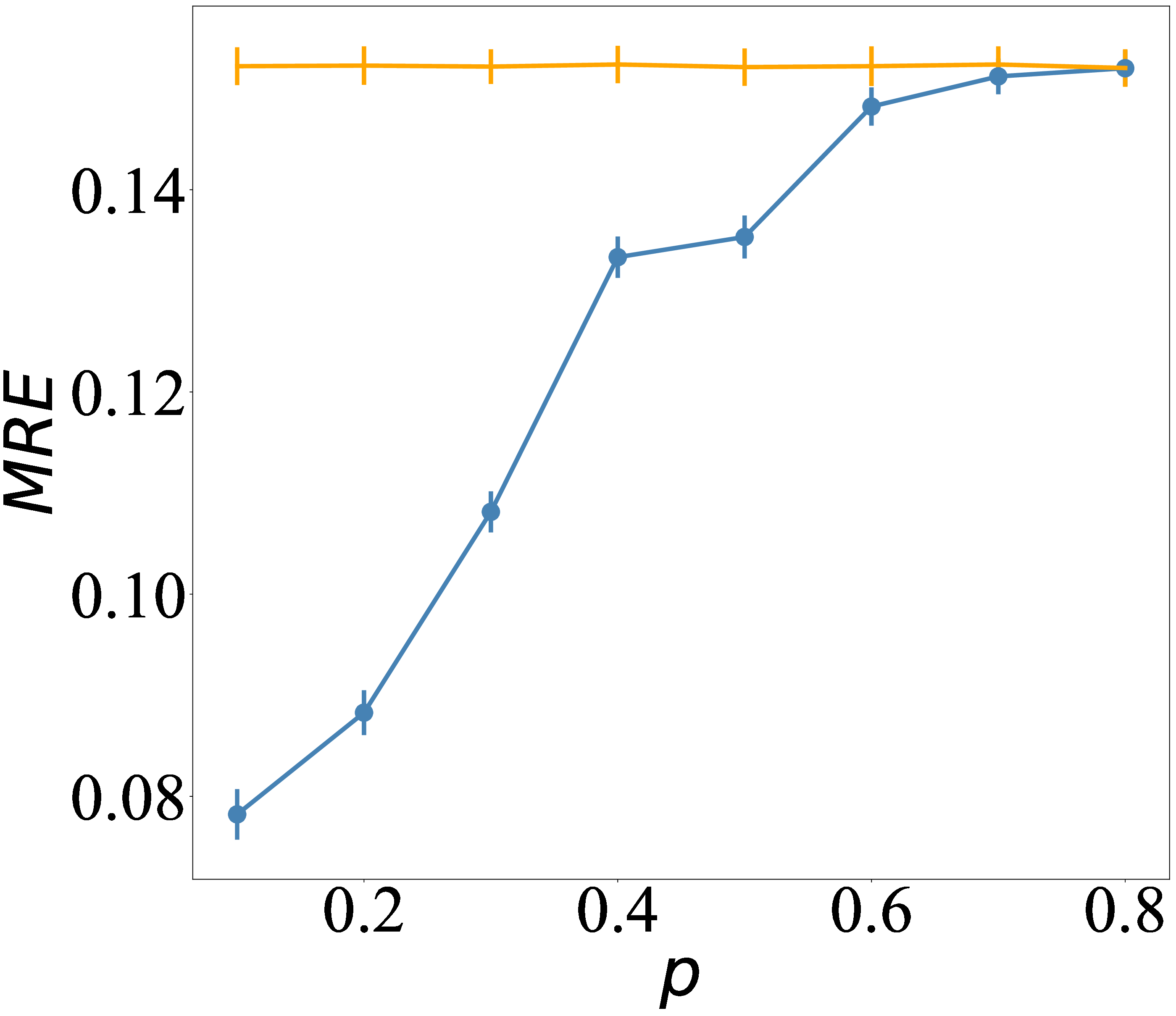}
        \caption{$h = 10$, $\varepsilon = 3$ ($\varepsilon_1 = 0.3$)}
        \label{fig:eval:hepth:h3v10}
        \vspace{4mm}
    \end{subfigure}
    \hfill
    \begin{subfigure}{0.45\columnwidth}
    \centering
        \vspace{3mm}
        \includegraphics[width=\hsize]{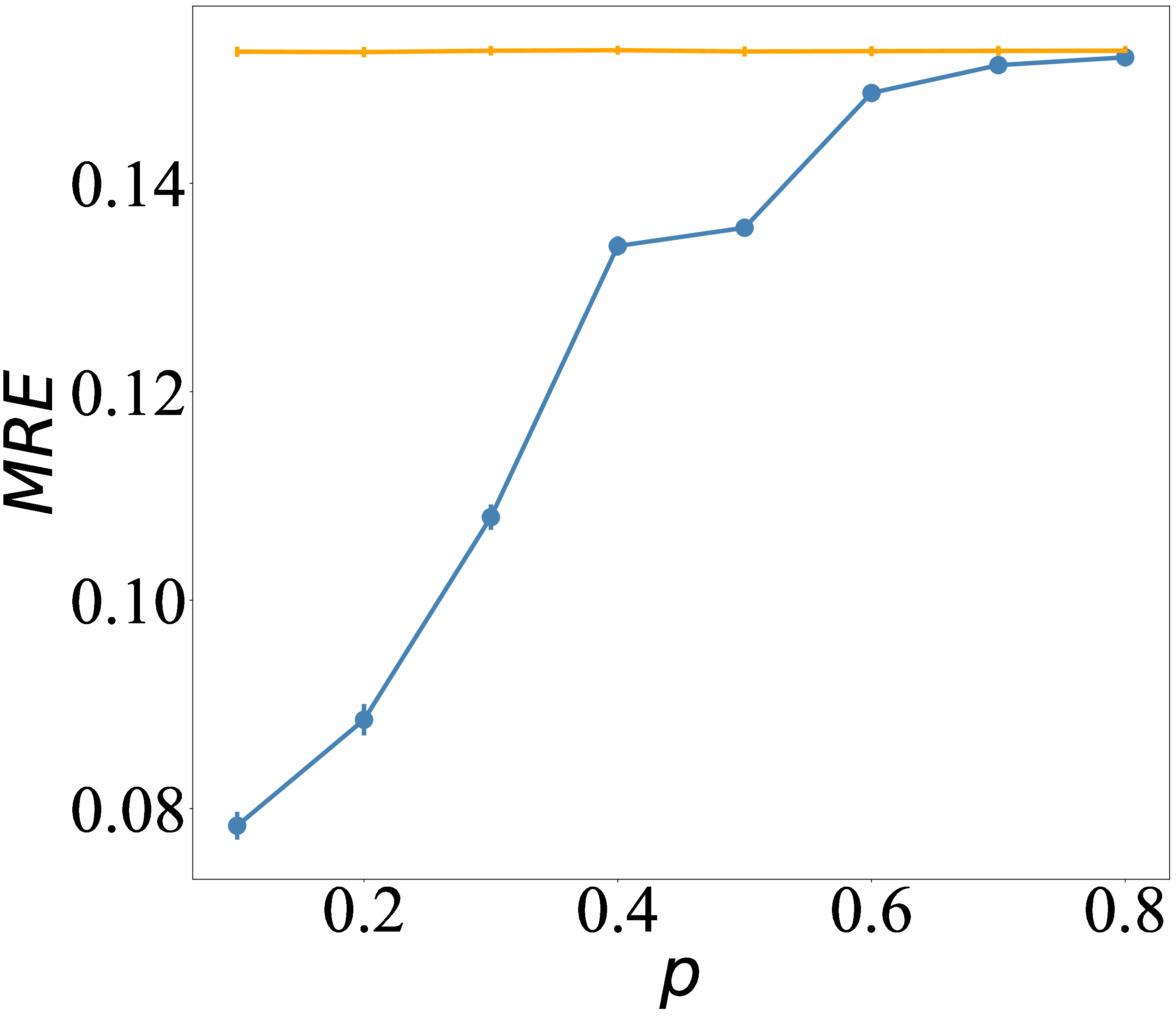}
        \caption{$h = 10$, $\varepsilon = 10$ ($\varepsilon_1 = 1.0$)}
        \label{fig:eval:hepth:var10}
        \vspace{4mm}
    \end{subfigure}
    \hfill
    \caption{\small{$p$ Selection for $3$-Clique Counting on HepTh}}
    \label{fig:eval:hepth:p}
    \vspace{3mm}
\end{figure}

\begin{figure}[!htb]
\centering
    \begin{subfigure}{0.47\columnwidth}
        \centering
        \includegraphics[width=\hsize]{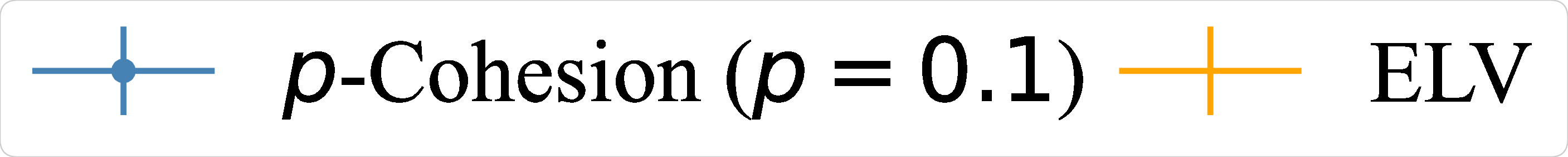}
    \end{subfigure}
    \hfill
    
    \begin{subfigure}{0.45\columnwidth}
    \centering
        \includegraphics[width=\hsize]{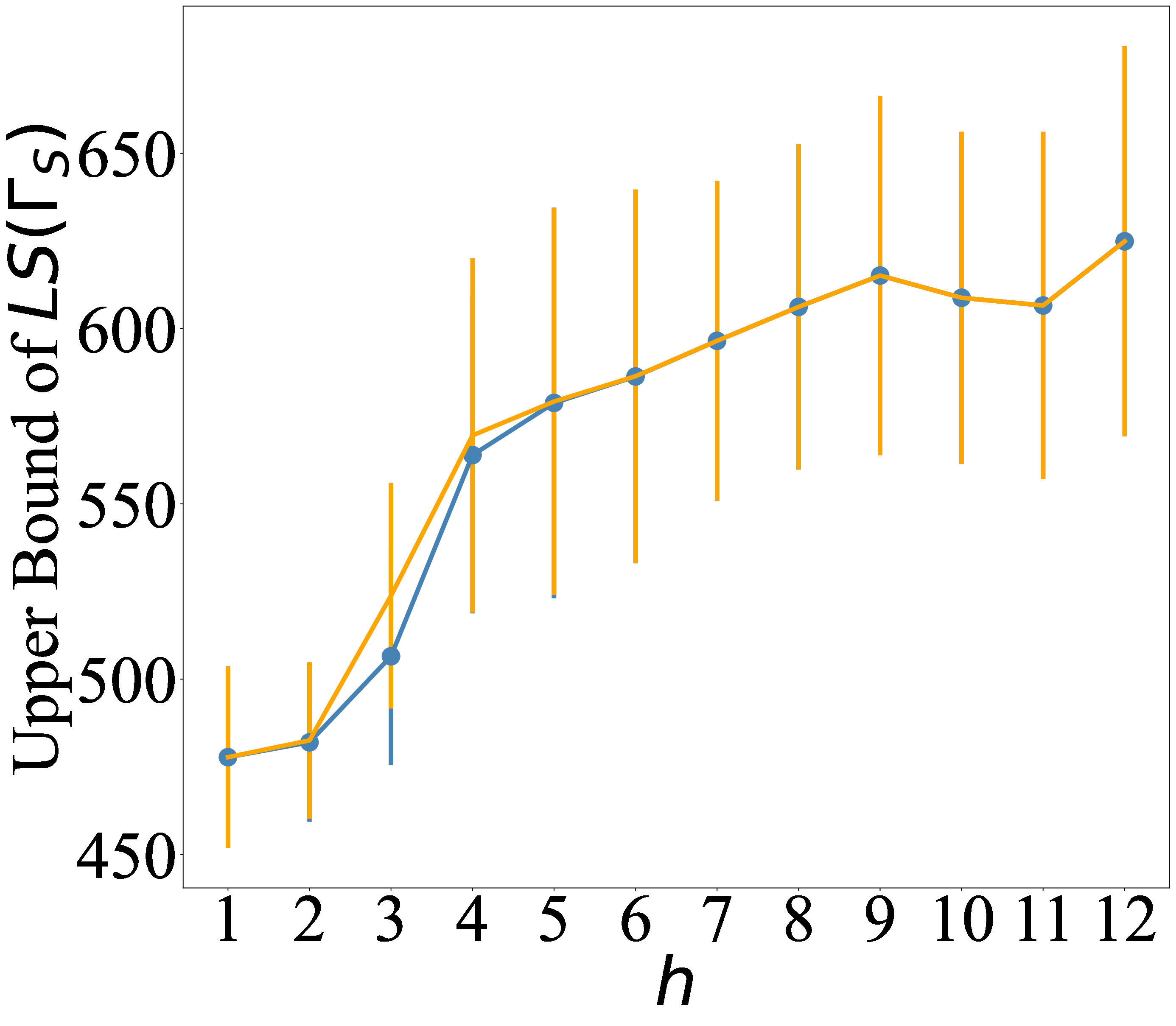}
        \caption{$\varepsilon = 1$ ($\varepsilon_1 = 0.1$)}
        \label{fig:eval:wikivote:var1}
        \vspace{4mm}
    \end{subfigure}
    \hfill
    \begin{subfigure}{0.45\columnwidth}
    \centering
        \vspace{3mm}
        \includegraphics[width=\hsize]{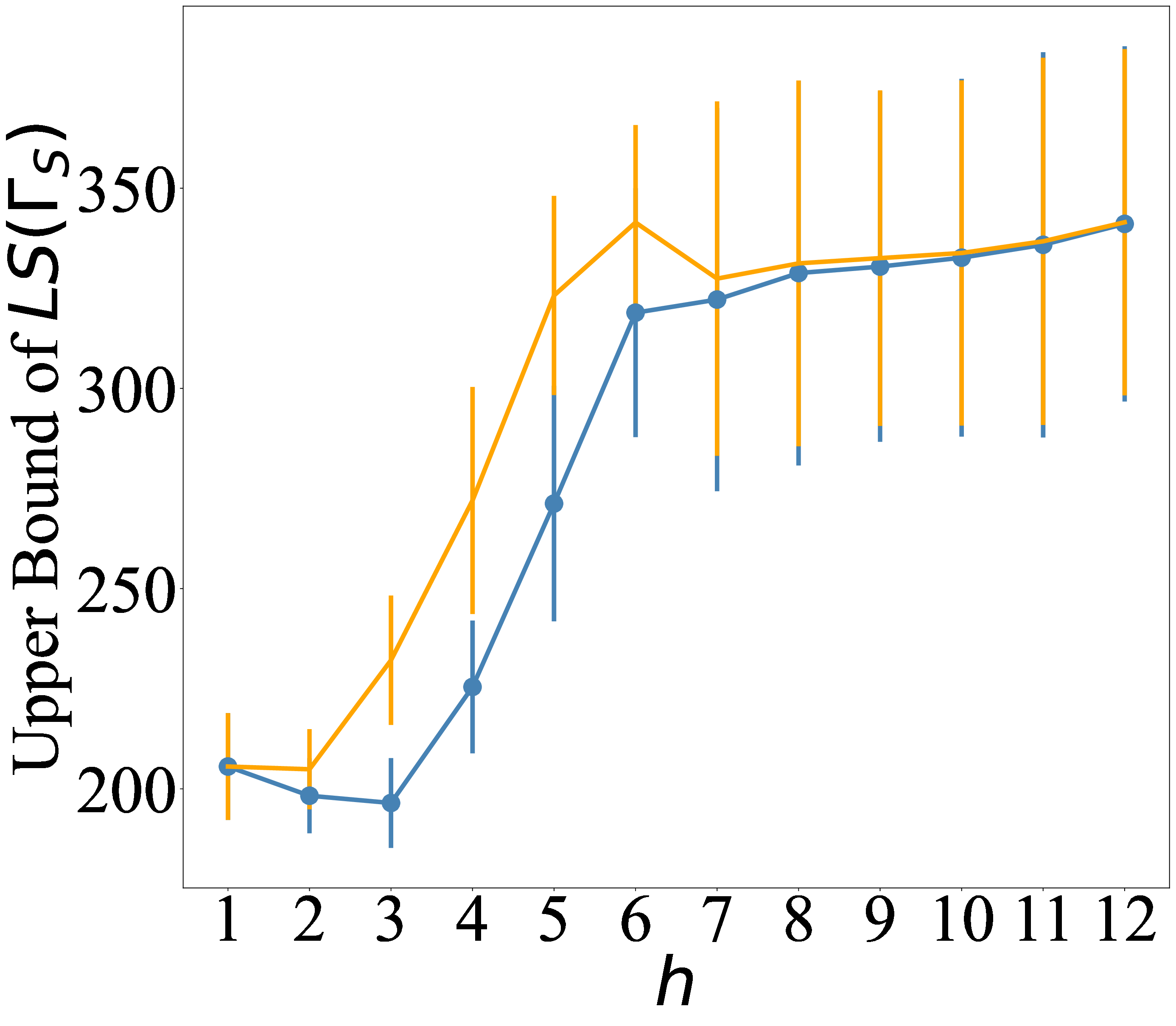}
        \caption{$\varepsilon = 3$ ($\varepsilon_1 = 0.3$)}
        \label{fig:eval:wikivote:var3}
        \vspace{4mm}
    \end{subfigure}
    \hfill
    \begin{subfigure}{0.45\columnwidth}
    \centering
        \includegraphics[width=\hsize]{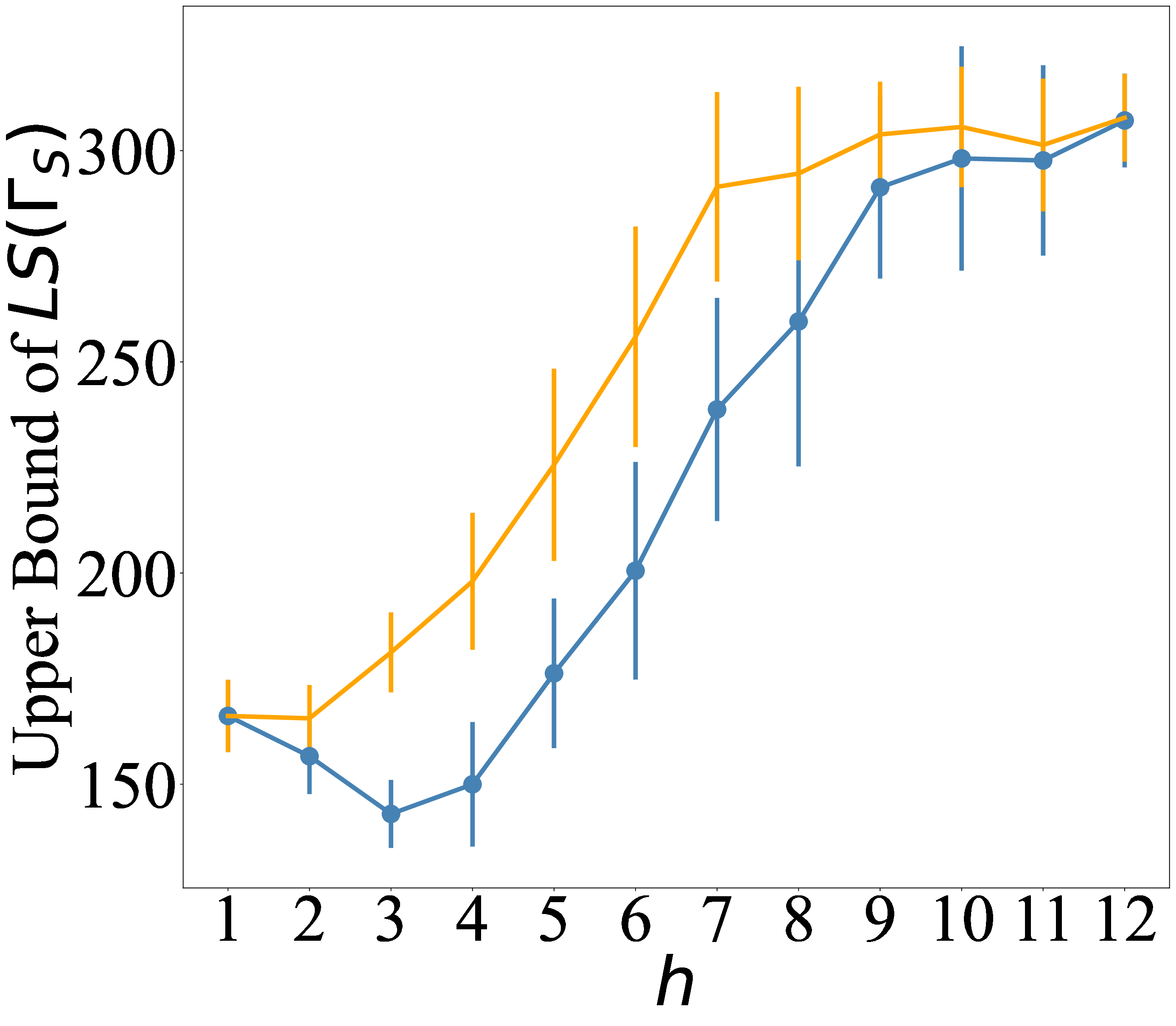}
        \caption{$\varepsilon = 5$ ($\varepsilon_1 = 0.5$)}
        \label{fig:eval:wikivote:var5}
        \vspace{4mm}
    \end{subfigure}
    \hfill
    \begin{subfigure}{0.45\columnwidth}
    \centering
        \vspace{3mm}
        \includegraphics[width=\hsize]{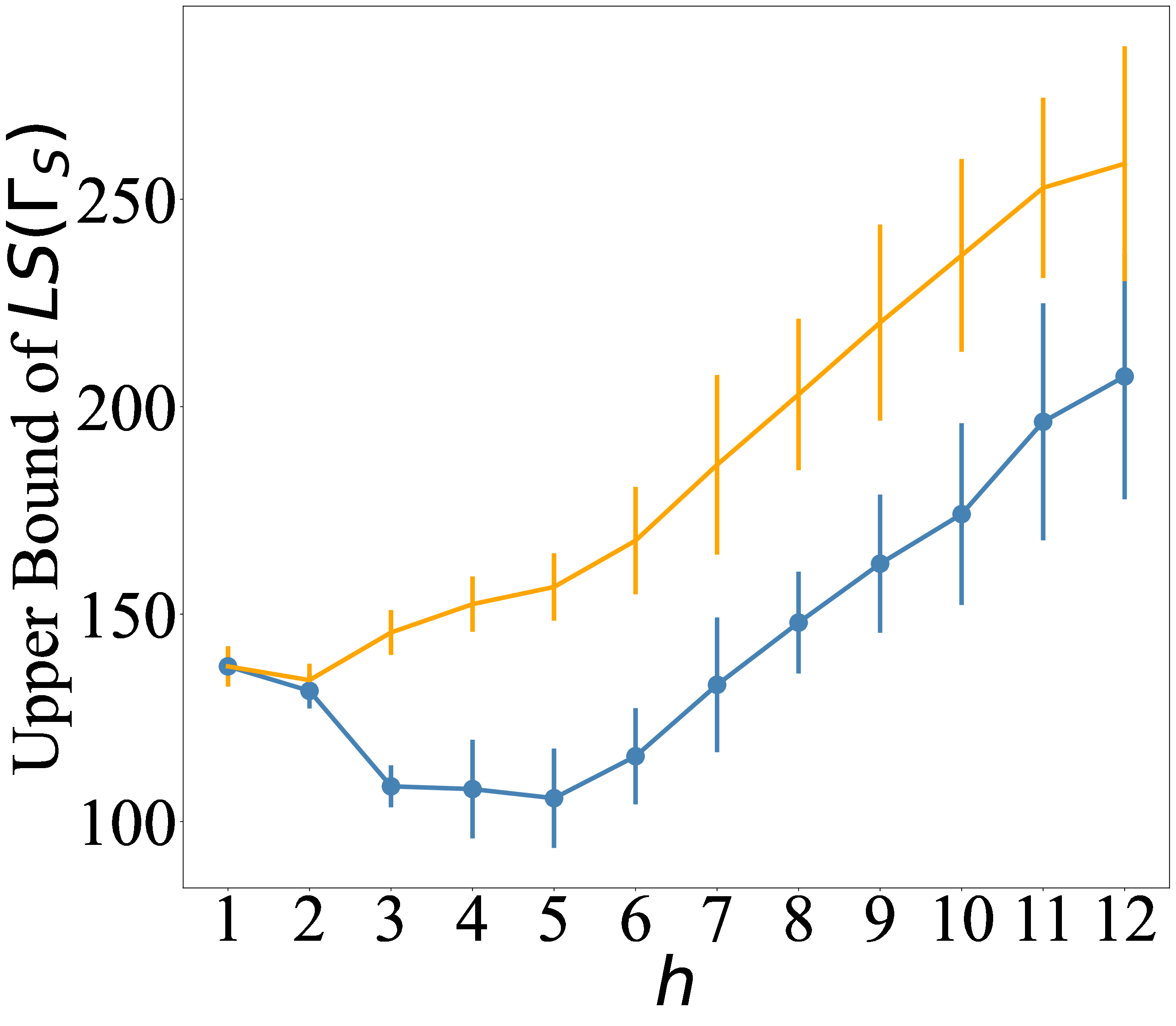}
        \caption{$\varepsilon = 10$ ($\varepsilon_1 = 1.0$)}
        \label{fig:eval:wikivote:var10}
        \vspace{4mm}
    \end{subfigure}
    \hfill
    \caption{\small{$h$ Selection for $3$-Clique Counting on Celegans, $p = 0.1$}}
    \label{fig:eval:wikivote_h}
    \vspace{3mm}
\end{figure}

\subsection{$k$-Clique Counting Under DDP}
\label{sec:eval:graphcount}

In this section, we focus on the influence of $h$, $p$, $\varepsilon$, and $\varepsilon_1$ over $k$-clique counting.
We use \underline{M}ean \underline{R}elative \underline{E}rror (\mre) to evaluate the performance when releasing responses, based on two different connection groups, for subgraphs counting queries.
We use \pcalg (\resp \elvalg) to denote the method for counting subgraphs, under the $(\varepsilon, \delta)$-DDP, based on the minimal \pcs (ELVs).
For subgraph $S$ counting, we have \mre $= \frac{|\Gamma_S^* - \Gamma_S|}{\Gamma_S}$, where $\Gamma_S$ is the real subgraph number and $\Gamma_S^*$ is the perturbed value regarding $\Gamma_S$. Each \mre reported is averaged over $100$ runs.

\subsubsection{$3$-Clique Counting}

We evaluate our algorithms for $3$-clique counting.

\vspace{1mm}
\noindent \underline{$p$ Selection}. Fig.~\ref{fig:eval:hepth:p} shows the average \mre under the $(\varepsilon, \delta)$-DDP for $3$-clique counting based on the subgraphs returned by \pcalg and \elvalg on \textit{HepTh} when $p$ varies from $0.1$ to $0.8$.
We show the results with $h = 1,3,10$ and $\varepsilon = 3.0, 10.0$ for all $p$ values.
As expected, when $p < 0.7$, the \pcalg outperforms \elvalg. The average \mre of \pcalg increases when $p$ becomes larger.
When $p > 0.4$, the \mre grows slowly because, for almost all the vertices, $p$ is large enough to include their connections within their corresponding minimal \pcs.
When $p \geq 0.7$, the ELVs will be subsets of the corresponding minimal \pcs.
It implies a reasonable $p$ can balance the data utility and privacy protection.


\begin{figure}[!htb]
\centering
    \begin{subfigure}{0.68\columnwidth}
    \centering
        \includegraphics[width=\hsize]{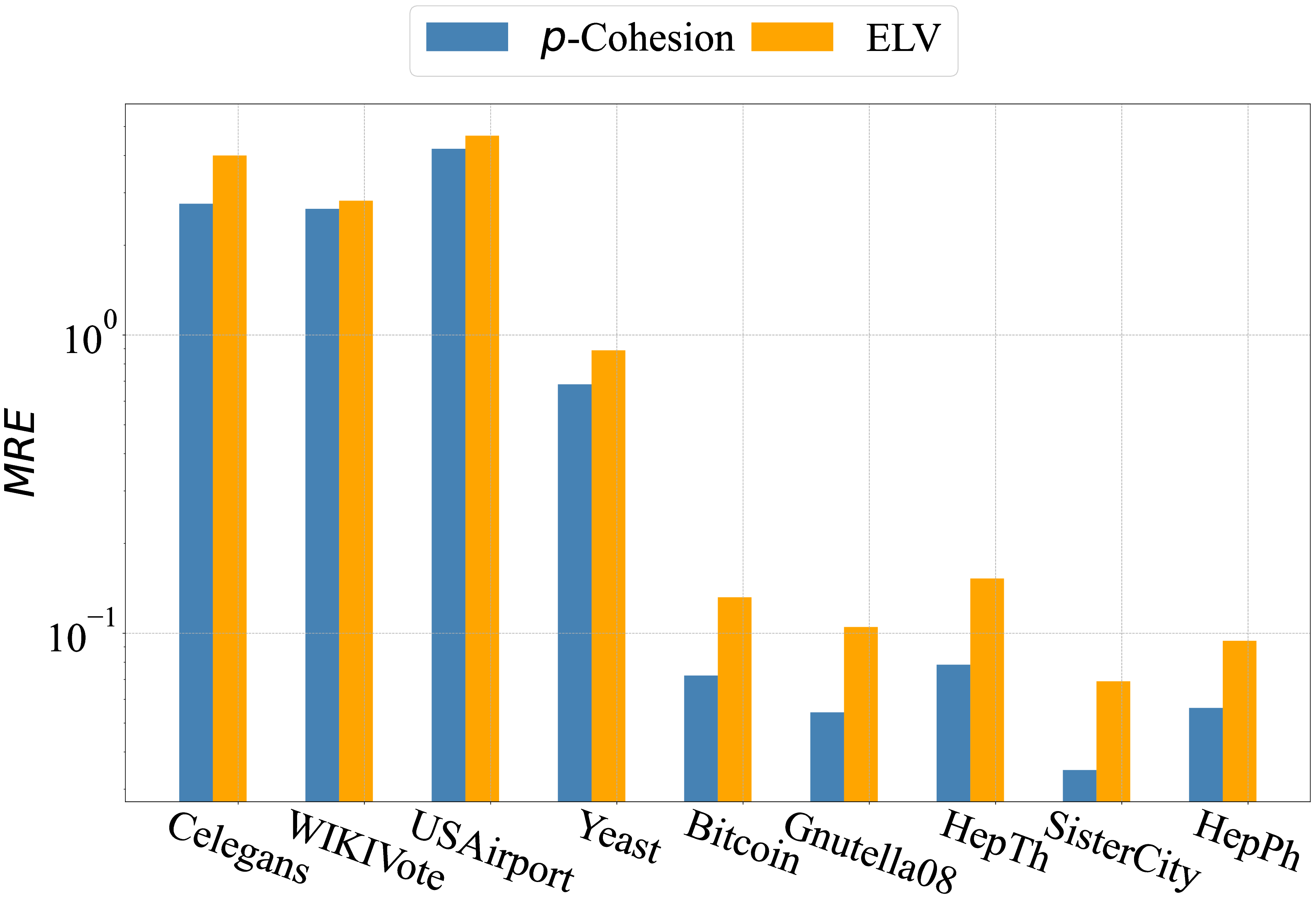}
        \caption{{\revision All Datasets with $h=3$, $\varepsilon = 10$ ($\varepsilon_1 = 1$)}}
        \label{fig:eval:USAirport:all}
        \vspace{4mm}
    \end{subfigure}
    \hfill

    \vspace{3mm}
    \begin{subfigure}{0.45\columnwidth}
    \centering
        \includegraphics[width=\hsize]{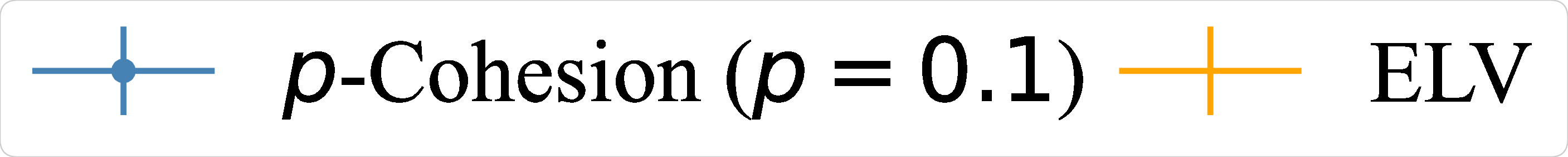}
    \end{subfigure}
    \hfill
    
    \begin{subfigure}{0.45\columnwidth}
    \centering
        \includegraphics[width=\hsize]{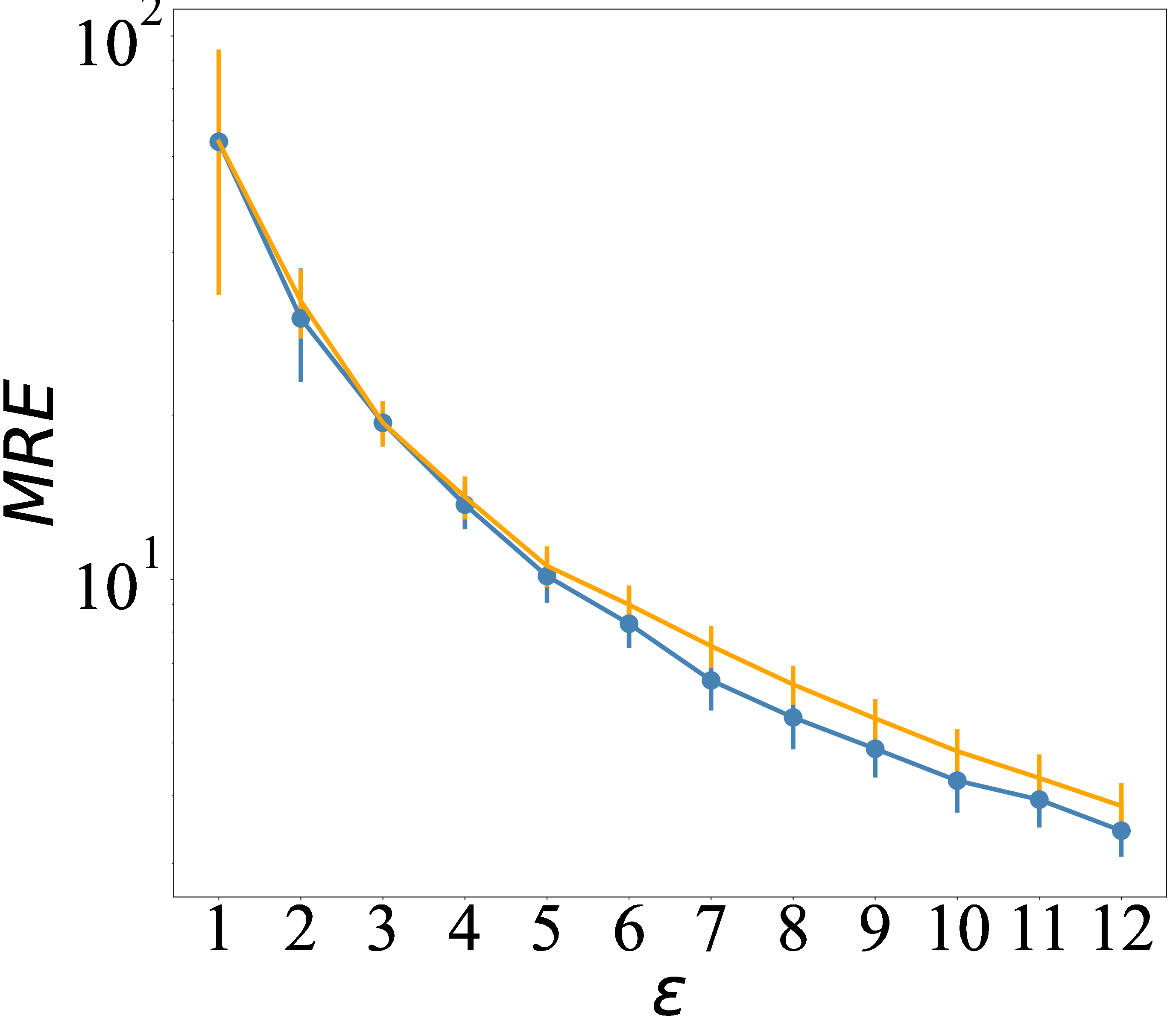}
        \caption{$h=1$ ($\varepsilon_1 = 0.1\varepsilon$)}
        \label{fig:eval:USAirport:h1}
        \vspace{4mm}
    \end{subfigure}
    \hfill
    \begin{subfigure}{0.45\columnwidth}
    \centering
        \vspace{3mm}
        \includegraphics[width=\hsize]{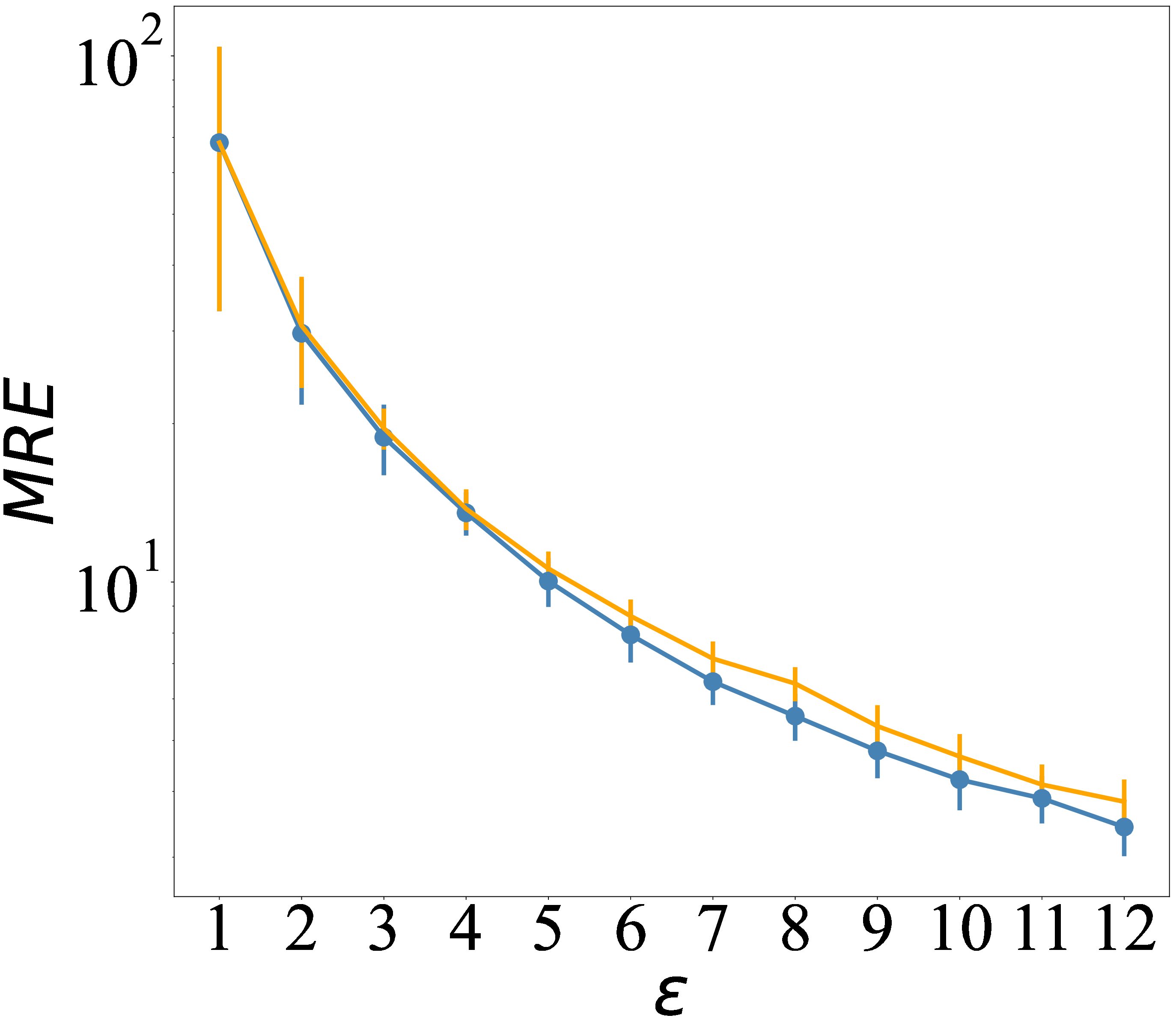}
        \caption{$h=3$ ($\varepsilon_1 = 0.1\varepsilon$)}
        \label{fig:eval:USAirport:h3}
        \vspace{4mm}
    \end{subfigure}
    \hfill
    \begin{subfigure}{0.45\columnwidth}
    \centering
        \includegraphics[width=\hsize]{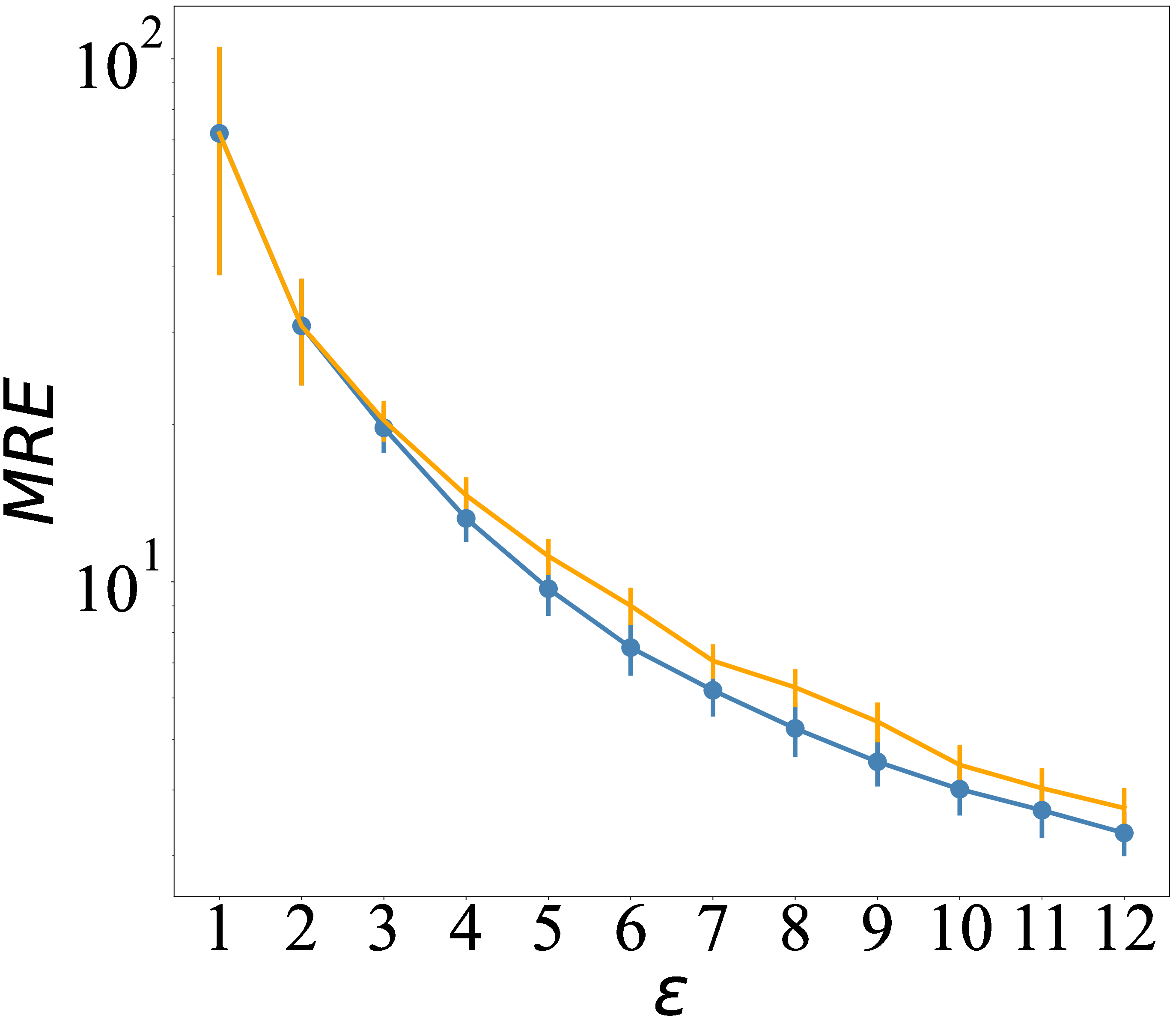}
        \caption{$h=5$ ($\varepsilon_1 = 0.1\varepsilon$)}
        \label{fig:eval:USAirport:h5}
        \vspace{4mm}
    \end{subfigure}
    \hfill
    \begin{subfigure}{0.45\columnwidth}
    \centering
        \vspace{3mm}
        \includegraphics[width=\hsize]{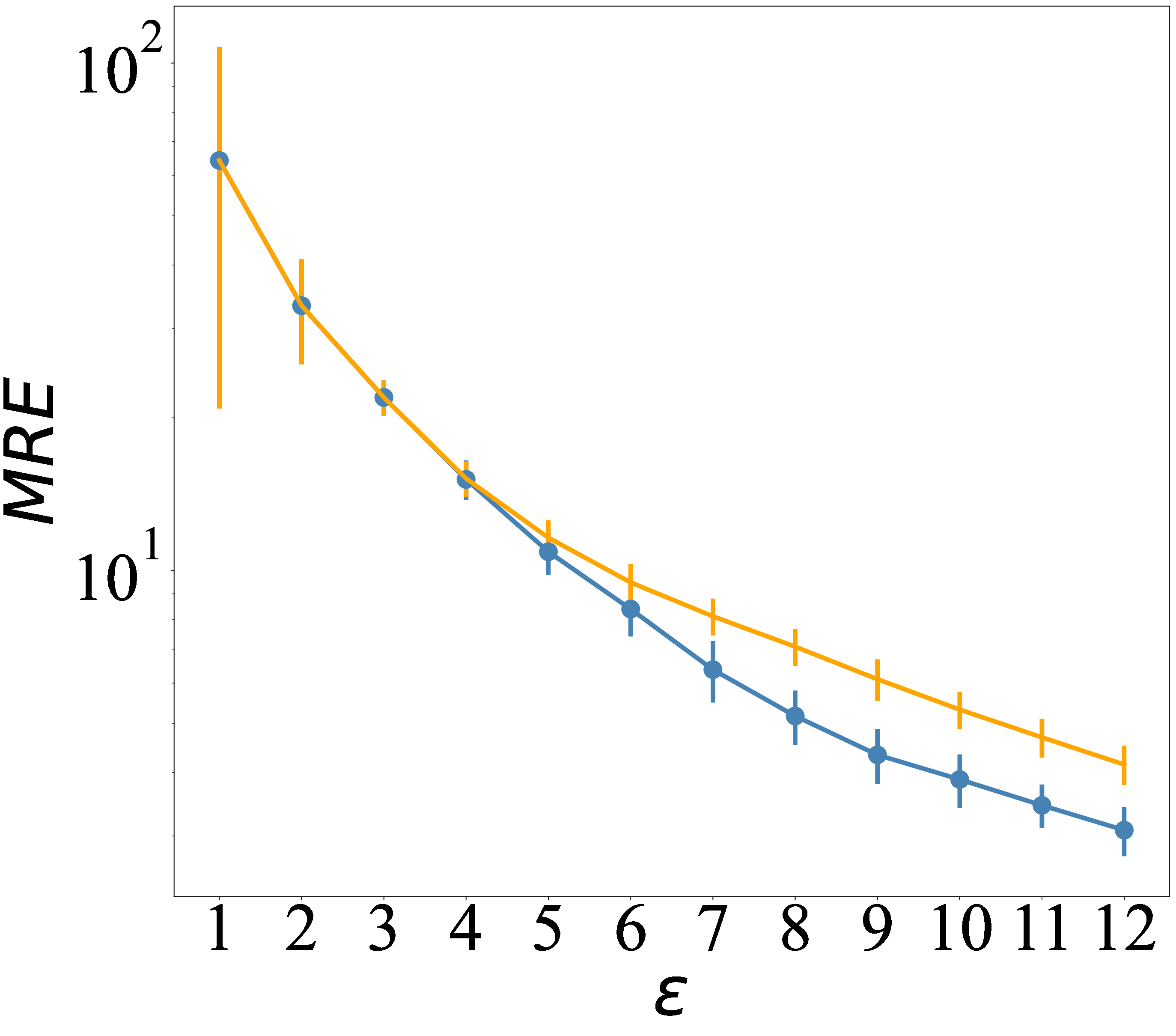}
        \caption{$h=10$ ($\varepsilon_1 = 0.1\varepsilon$)}
        \label{fig:eval:USAirport:h10}
        \vspace{4mm}
    \end{subfigure}
    \hfill
    \caption{\small{$\varepsilon$ Selection for $3$-Clique Counting on USAirport, $p = 0.1$}}
    \label{fig:eval:USAirport_var}
    \vspace{3mm}
\end{figure}

\vspace{1mm}
\noindent \underline{$h$ Selection}.
In Algorithm~\ref{alg:phase1:kclique}, to balance the efficiency and effectiveness for noise scale $\lambda$ estimation, $h$ vertices need to report their common neighbors, instead of degree (Lines~\ref{alg:phase1:kclique_9}-\ref{alg:phase1:kclique_17}).
Fig.~\ref{fig:eval:wikivote_h} reports the upper bound of the local sensitivity generated by \pcalg and \elvalg on \textit{Celegans} when $h$ varies from $1$ to $12$ with a step size of $1$.
We choose $p = 0.1$ for the minimal \pcs computation. For Phase-$1$, the privacy budget is chosen as $\varepsilon_1 = 0.1 \varepsilon$ for different $\varepsilon$ values.
For different privacy budgets $\varepsilon$, the upper bounds returned by \pcalg show similar trends as \elvalg.
When $h = 2, 3$, the upper bounds are optimal for almost all $\varepsilon$ values.
Compared to \elvalg, with the increase of $\varepsilon$, our \pcalg can find more optimal upper bounds for $LS(\Gamma_S)$. 

\vspace{1mm}
\noindent \underline{$\varepsilon$ Selection}.
Fig.~\ref{fig:eval:USAirport:all} reports the average \mre returned by \pcalg and \elvalg on all datasets with $p = 0.1$, $h = 3$, and $\varepsilon = 10$ (\ie, $\varepsilon_1 = 1.0$).
On all datasets, \pcalg outperforms \elvalg.
For example, on \textit{SisterCity}, the \pcalg is two orders of magnitude smaller than \elvalg in terms of the average \mre.
That is because our approach injects less noise into the real response and estimates tighter local sensitivity.
Fig.~\ref{fig:eval:USAirport:h1} - Fig.~\ref{fig:eval:USAirport:h10} report the average \mre over \pcalg and \elvalg on \textit{USAirport} when $\varepsilon$ varies from $1$ to $12$ with a step size of $1$. We choose $p = 0.1$ for minimal \pcs computation.
For all $\varepsilon$, the privacy budget $\varepsilon_1$ for Phase-$1$ is chosen as $\varepsilon_1 = 0.1\varepsilon$.
Our \pcalg shows similar trends as \elvalg under all $h$ but reports increasingly better performance in terms of \mre with $\varepsilon$ growing.

\begin{figure}[htb]
\centering
    \begin{subfigure}{0.68\columnwidth}
    \centering
        \includegraphics[width=\hsize]{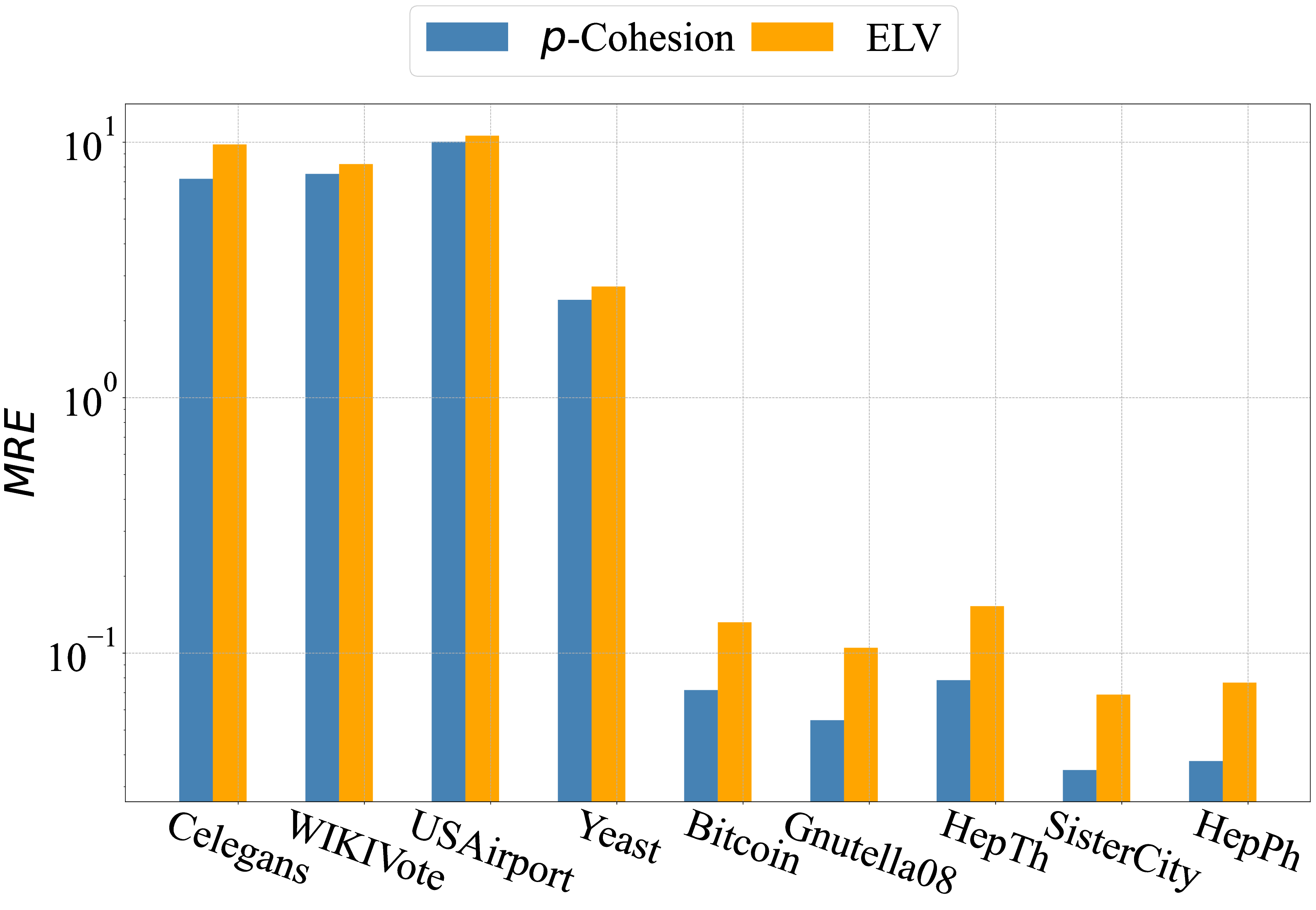}
        \caption{{\revision All Datasets with $h=3$, $\varepsilon_1 = 0.5$ ($\varepsilon = 5$)}}
        \label{fig:eval1:all}
        \vspace{4mm}
    \end{subfigure}
    \hfill
    
    \vspace{3mm}
    \begin{subfigure}{0.45\columnwidth}
    \centering
        \includegraphics[width=\hsize]{images/effectiveness/triangle/var/legend.eps}
    \end{subfigure}
    \hfill
    
    \begin{subfigure}{0.45\columnwidth}
    \centering
        \includegraphics[width=\hsize]{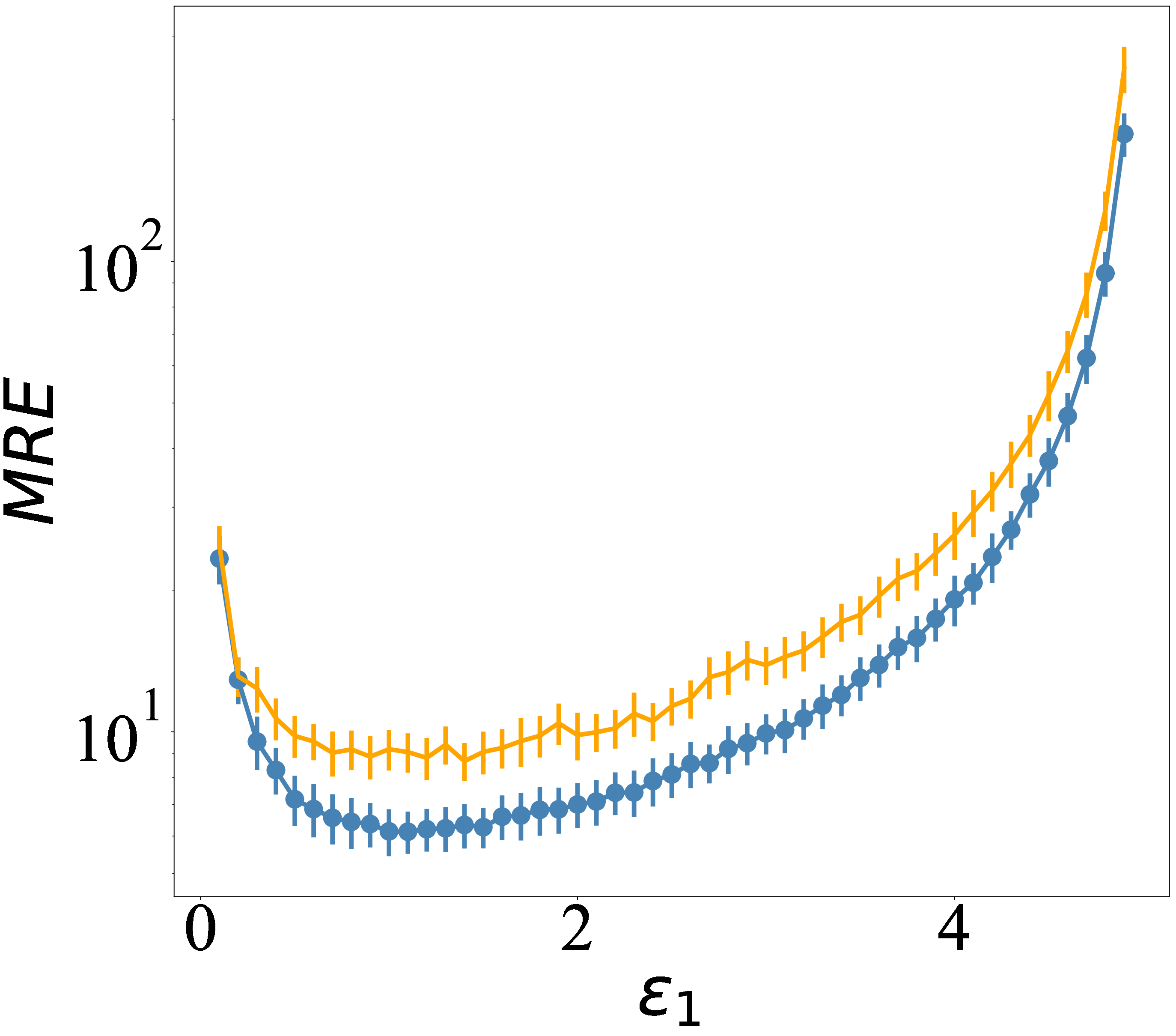}
        \caption{$h = 3$, $\varepsilon = 5$}
        \label{fig:eval1:celegan:h3}
        \vspace{4mm}
    \end{subfigure}
    \hfill
    \begin{subfigure}{0.45\columnwidth}
    \centering
        \vspace{3mm}
        \includegraphics[width=\hsize]{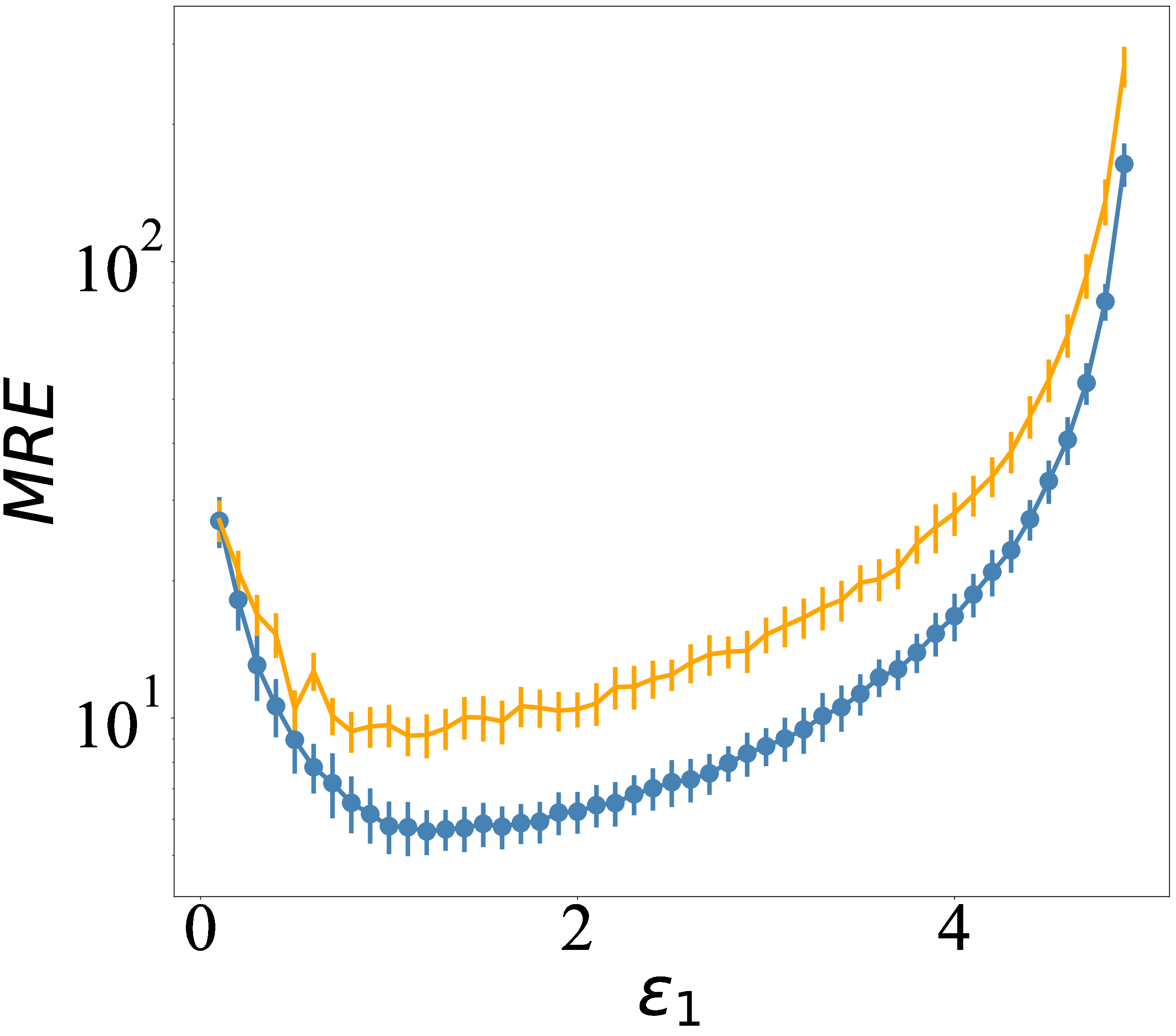}
        \caption{$h = 5$, $\varepsilon = 5$}
        \label{fig:eval1:celegan:h5}
        \vspace{4mm}
    \end{subfigure}
    \hfill
    \begin{subfigure}{0.45\columnwidth}
    \centering
        \includegraphics[width=\hsize]{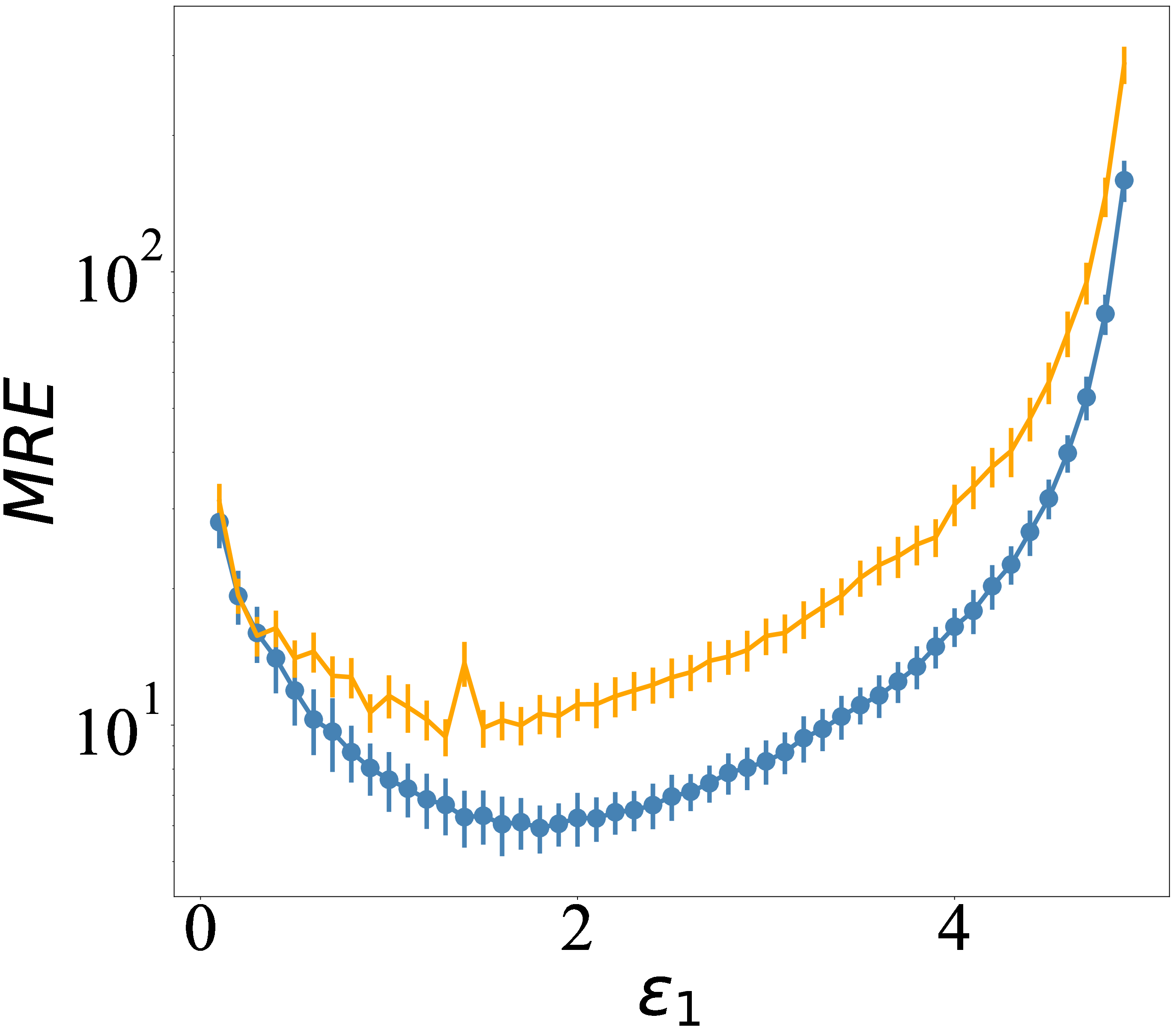}
        \caption{$h = 7$, $\varepsilon = 5$}
        \label{fig:eval1:celegan:h7}
        \vspace{4mm}
    \end{subfigure}
    \hfill
    \begin{subfigure}{0.45\columnwidth}
    \centering
        \vspace{3mm}
        \includegraphics[width=\hsize]{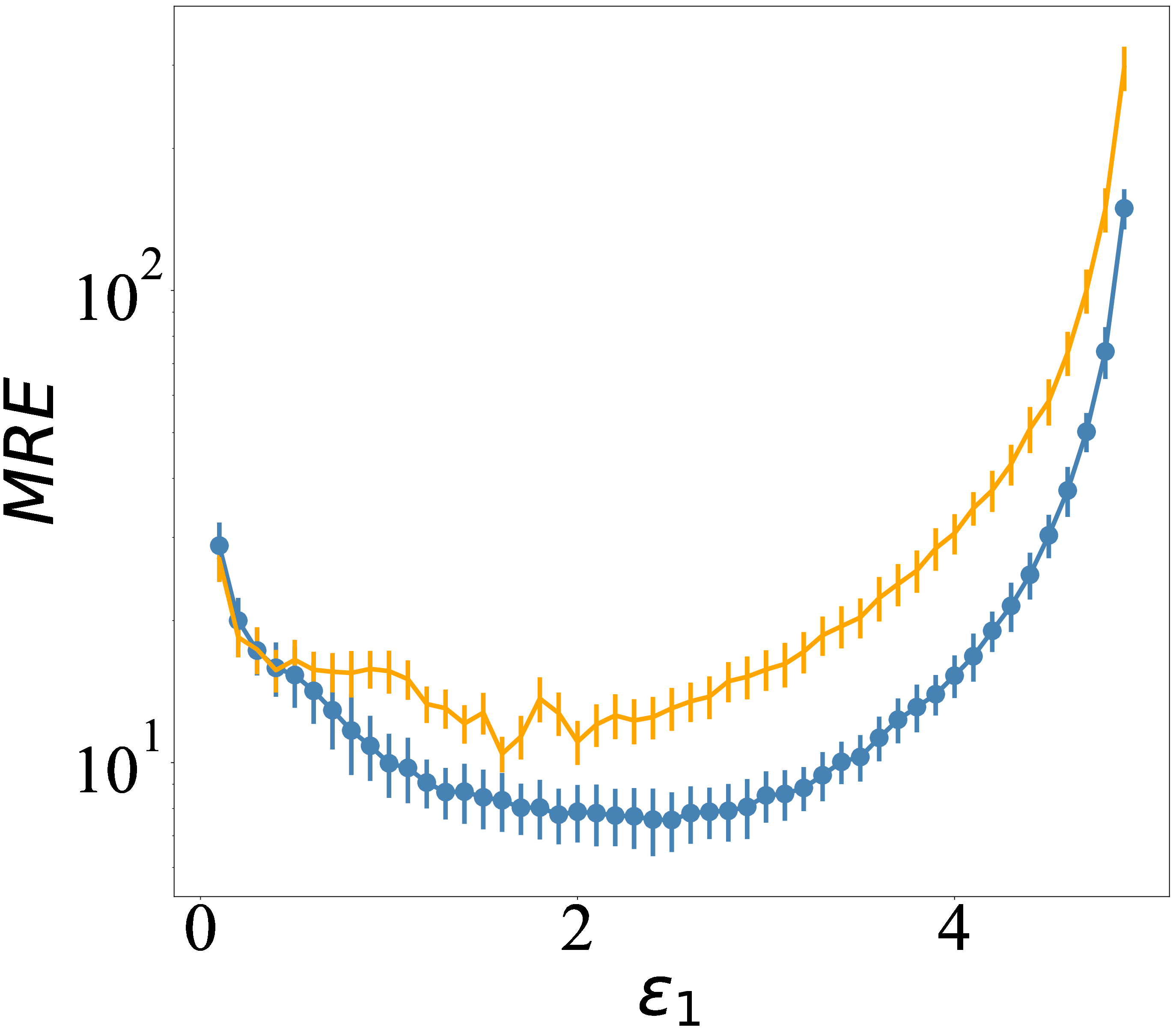}
        \caption{$ h = 10$, $\varepsilon = 5$}
        \label{fig:eval1:celegan:h10}
        \vspace{4mm}
    \end{subfigure}
    \hfill
    \caption{\small{$\varepsilon_1$ Selection for $3$-Clique Counting on Celegans, $p = 0.1$}}
    \label{fig:eval1:celegan_var1}
\end{figure}
\vspace{1mm}
\noindent \underline{$\varepsilon_1$ Selection}.
Fig.~\ref{fig:eval1:celegan_var1} depicts the influence of privacy budget $\varepsilon_1$ of Phase-$1$ on \mre.
For all $\varepsilon_1$, the total privacy budget is $\varepsilon = 5.0$, and we choose $p = 0.1$ for the minimal \pc computation.
The \texttt{MRE}s represent the average values obtained from $100$ independent tests.
Fig.~\ref{fig:eval1:all} shows the average \mre returned by \pcalg and \elvalg on all datasets with $p = 0.1$, $h = 3$, $\varepsilon_1 = 0.5$, and $\varepsilon = 5.0$.
Under the current setting, on all datasets, our \pcalg outperforms \elvalg.
Consistent with the previous result, on \textit{SisterCity}, the \pcalg is two orders of magnitude smaller than \elvalg in the average \mre.
Fig.~\ref{fig:eval1:celegan:h3} - Fig.~\ref{fig:eval1:celegan:h10} report the average \mre over \pcalg and \elvalg on \textit{Celegans} when $\varepsilon_1$ varies from $0.1$ to $4.9$ with a step size of $0.1$.
For all $\varepsilon_1$, our \pcalg outperforms \elvalg in terms of the average \mre.
When $\varepsilon_1 < 1.5$, the average \mre decreases with the growth of $\varepsilon_1$ since a larger $\varepsilon_1$ results in a smaller upper bound for the local sensitivity $LS(\Gamma_S)$.
The average \mre reaches the bottom ($\varepsilon_1 \approx 1.5$), $\varepsilon_1$ and $\varepsilon_2$ reach the balance point, where $\varepsilon_2 = \varepsilon - \varepsilon_1$.
When $\varepsilon_1 > 1.5$, the average \mre increases as $\varepsilon_1$ grows since the balance is broken and more noises are injected.

\begin{figure}[htb]
\centering
    \begin{subfigure}{0.4\columnwidth}
    \centering
        \includegraphics[width=\hsize]{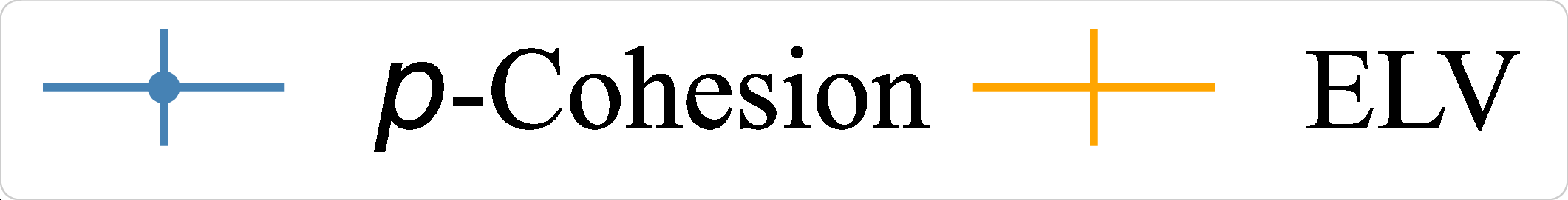}
    \end{subfigure}
    \hfill
    
  \begin{subfigure}{0.45\columnwidth}
  \centering
        \includegraphics[width=\hsize]{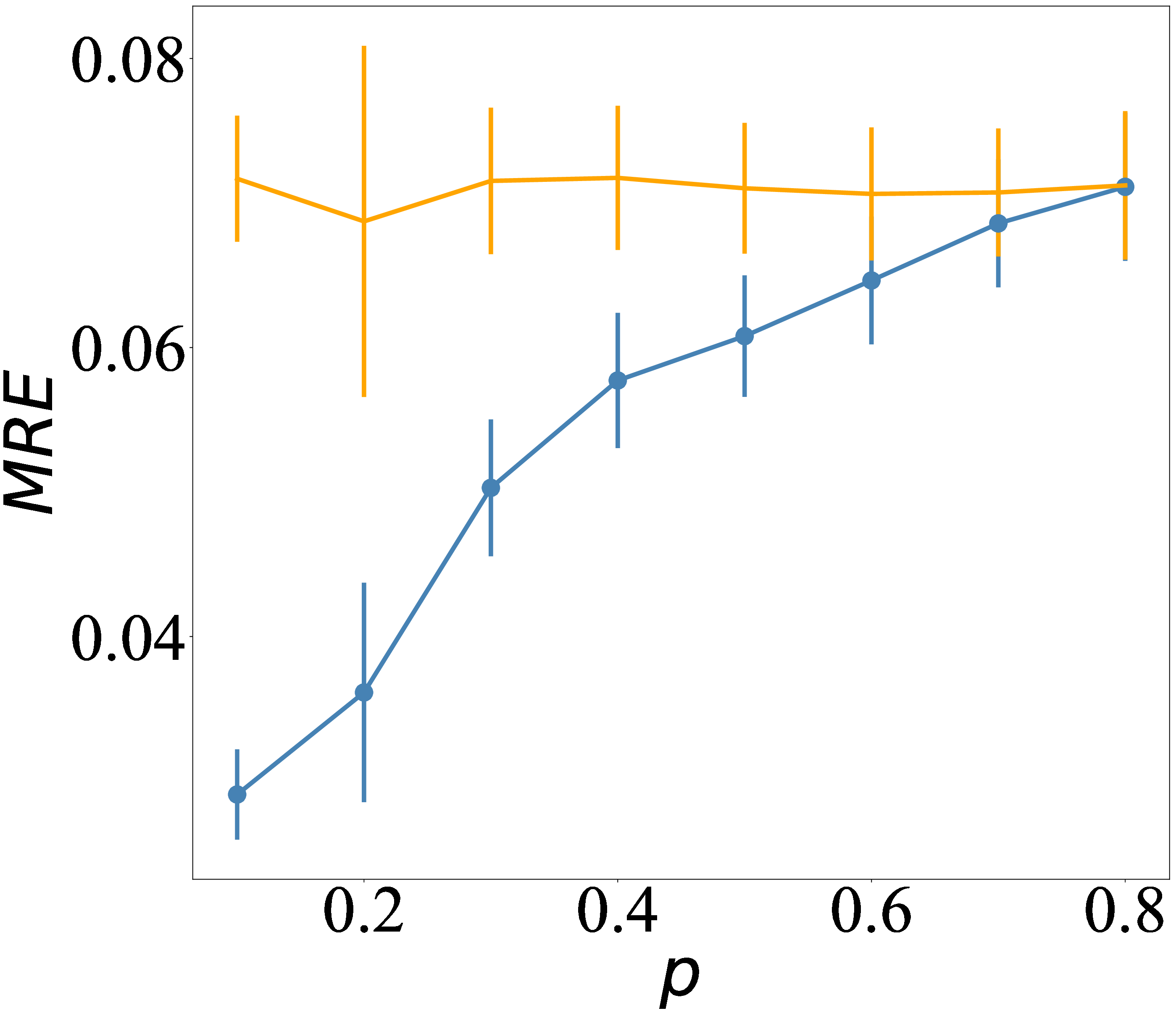}
        \caption{$h = 1$, $\varepsilon = 3$ ($\varepsilon_1 = 0.3$)}
        \label{fig:eval:Bitcoin:h1v2}
        \vspace{4mm}
    \end{subfigure}
    \hfill
    \begin{subfigure}{0.45\columnwidth}
    \centering
        \vspace{3mm}
        \includegraphics[width=\hsize]{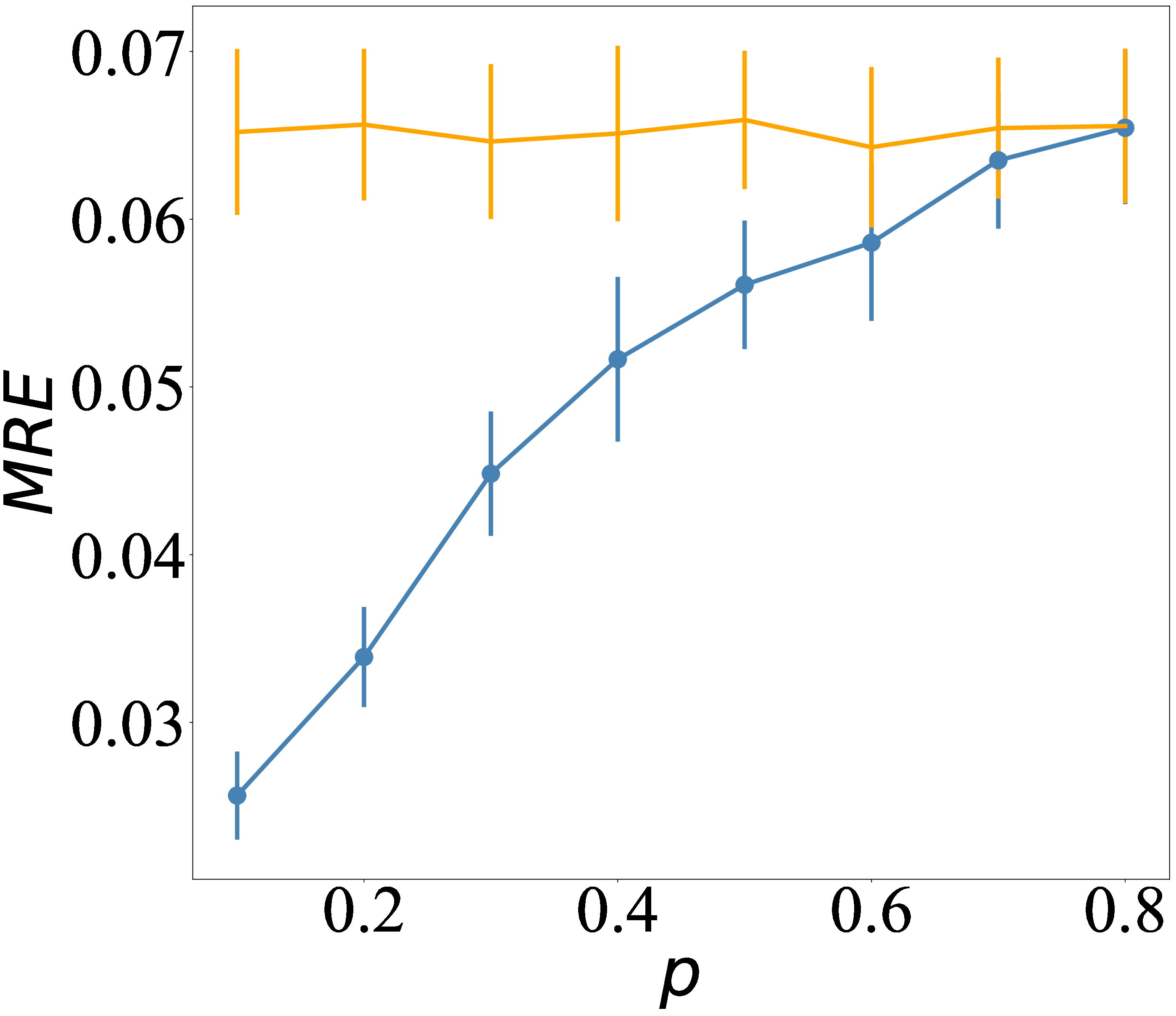}
        \caption{$h = 1$, $\varepsilon = 10$ ($\varepsilon_1 = 1.0$)}
        \label{fig:eval:Bitcoin:h1v10}
        \vspace{4mm}
    \end{subfigure}
    \hfill
    \begin{subfigure}{0.45\columnwidth}
    \centering
        \includegraphics[width=\hsize]{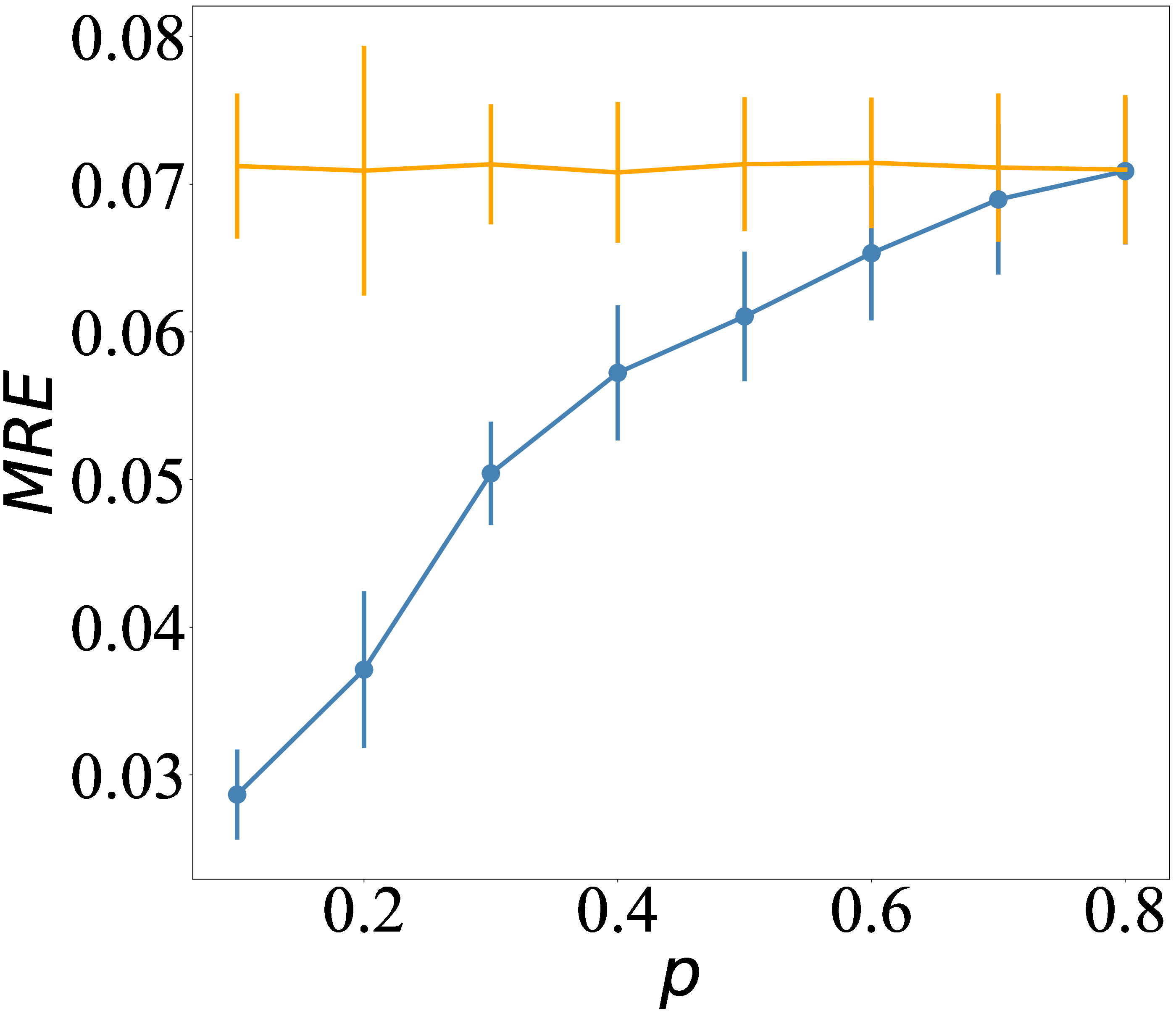}
        \caption{$h = 3$, $\varepsilon = 3$ ($\varepsilon_1 = 0.3$)}
        \label{fig:eval:Bitcoin:h3v2}
        \vspace{4mm}
    \end{subfigure}
    \hfill
    \begin{subfigure}{0.45\columnwidth}
    \centering
        \vspace{3mm}
        \includegraphics[width=\hsize]{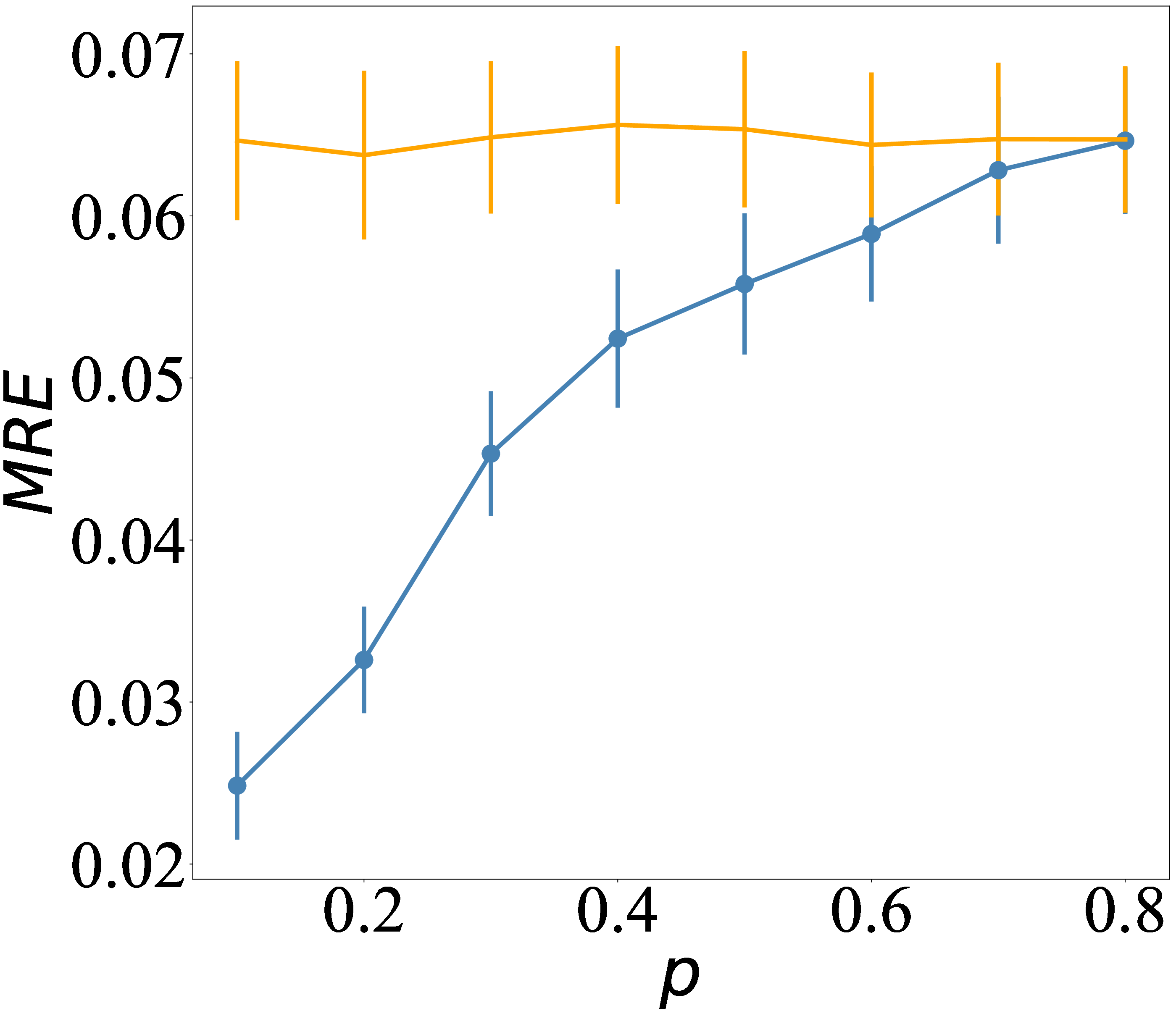}
        \caption{$h = 3$, $\varepsilon = 10$ ($\varepsilon_1 = 1.0$)}
        \label{fig:eval:Bitcoin:h3v10}
        \vspace{4mm}
    \end{subfigure}
    \hfill
    \begin{subfigure}{0.45\columnwidth}
    \centering
        \includegraphics[width=\hsize]{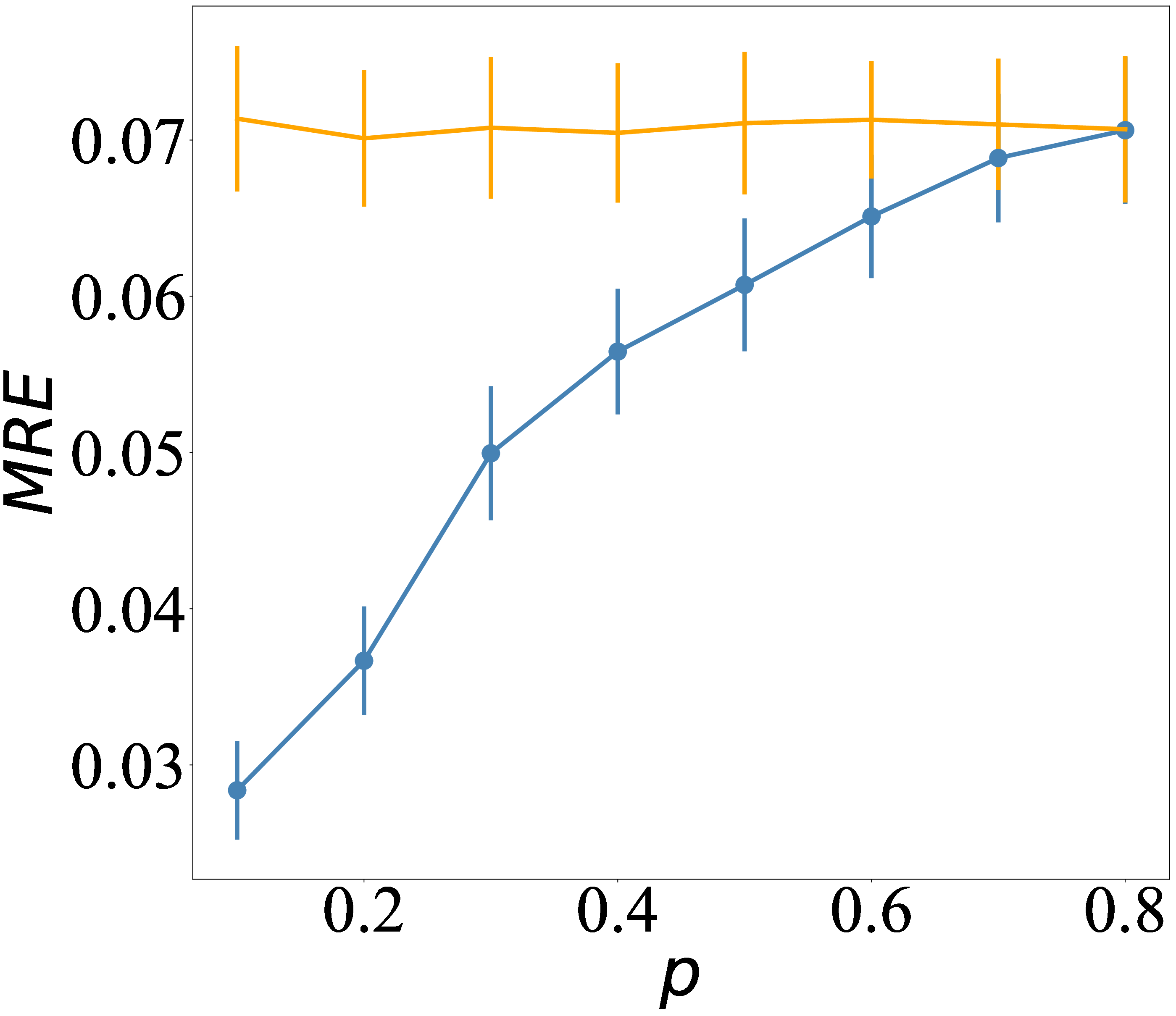}
        \caption{$h = 10$, $\varepsilon = 3$ ($\varepsilon_1 = 0.3$)}
        \label{fig:eval:Bitcoin:h3v10}
         \vspace{4mm}
    \end{subfigure}
    \hfill
    \begin{subfigure}{0.45\columnwidth}
    \centering
        \vspace{3mm}
        \includegraphics[width=\hsize]{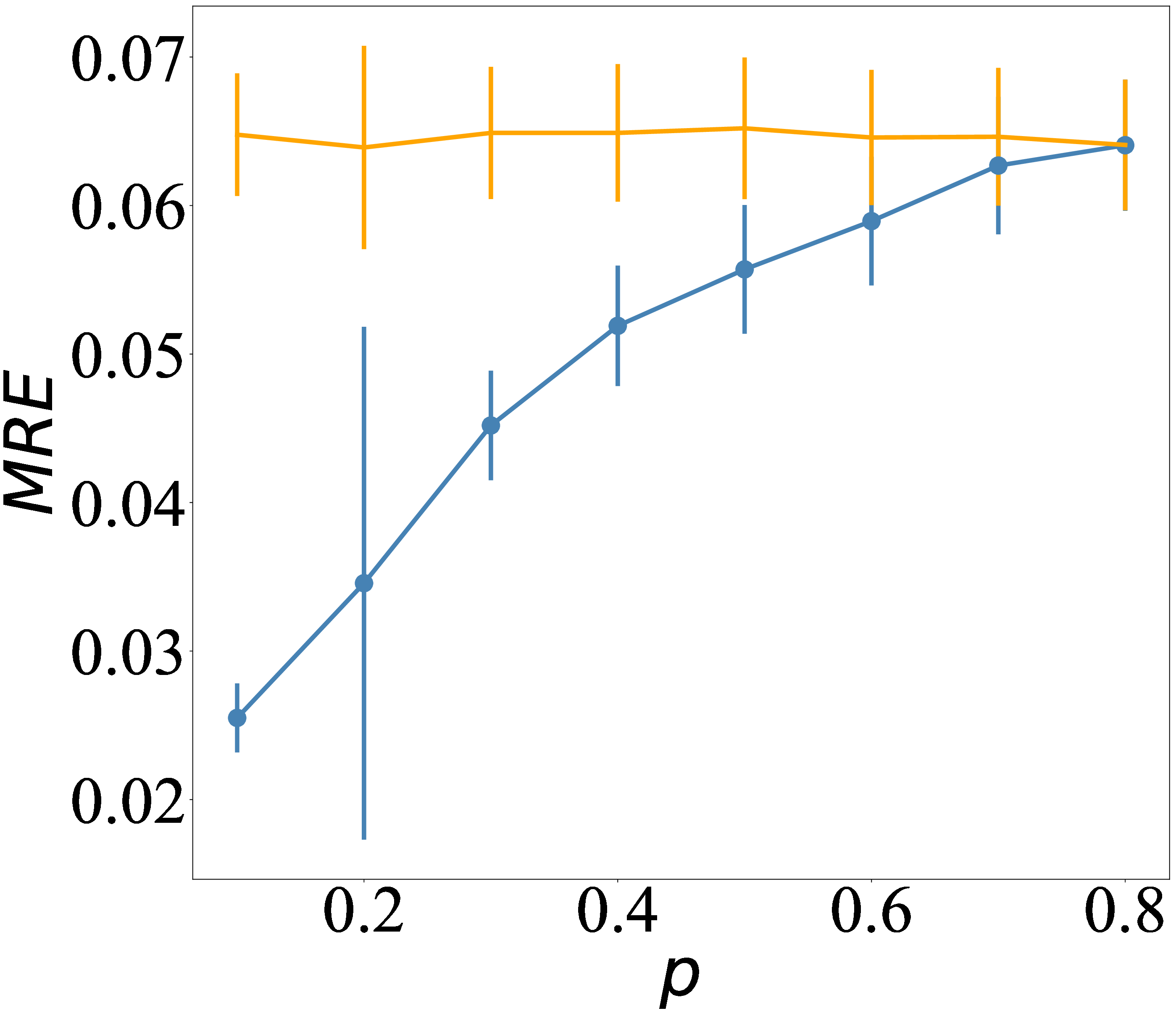}
        \caption{$h = 10$, $\varepsilon = 10$ ($\varepsilon_1 = 1.0$)}
        \label{fig:eval:Bitcoin:var10}
         \vspace{4mm}
    \end{subfigure}
    \hfill
   
    \caption{\small{$p$ Selection for $4$-Clique Counting on Bitcoin}}
    \label{fig:eval:Bitcoin:p}
\end{figure}

\begin{figure}[htb]
\centering
    \begin{subfigure}{0.68\columnwidth}
    \centering
        \includegraphics[width=\hsize]{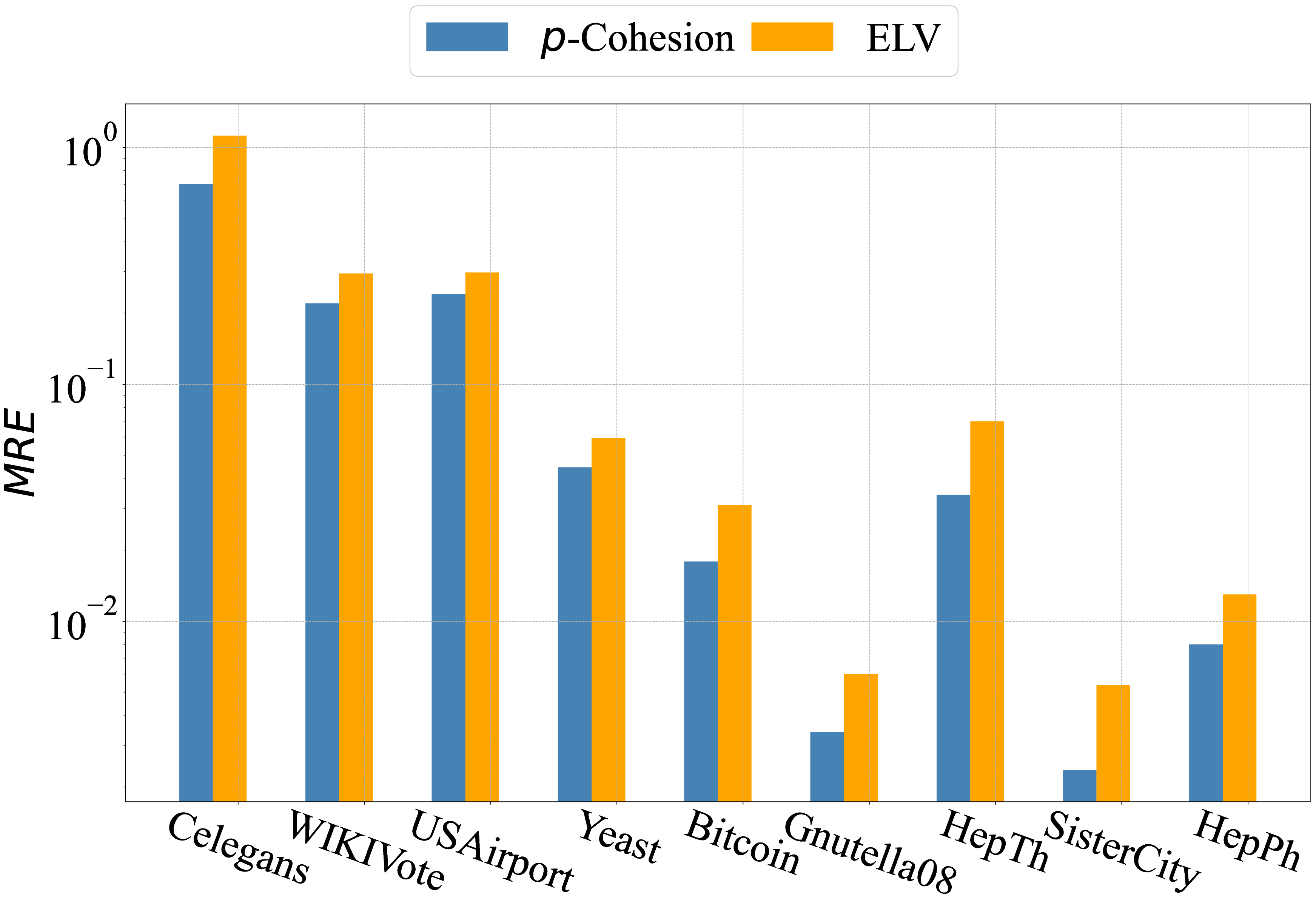}
        \caption{All Datasets with $h=3$, $\varepsilon = 10$ ($\varepsilon_1 = 1$)}
        \label{fig:eval:all}
        \vspace{4mm}
    \end{subfigure}
    \hfill

    \vspace{3mm}
    \begin{subfigure}{0.45\columnwidth}
    \centering
        \includegraphics[width=\hsize]{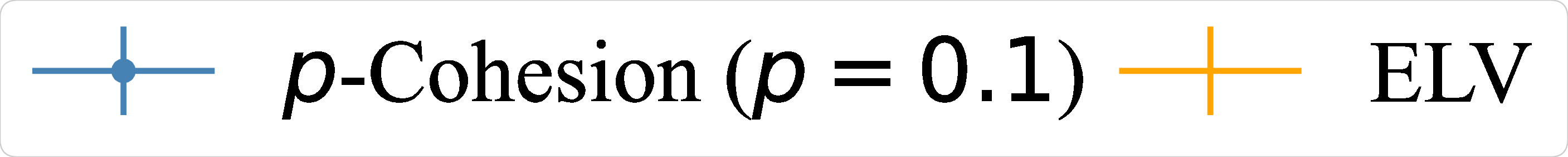}
    \end{subfigure}
    \hfill
    
    \begin{subfigure}{0.45\columnwidth}
    \centering
        \includegraphics[width=\hsize]{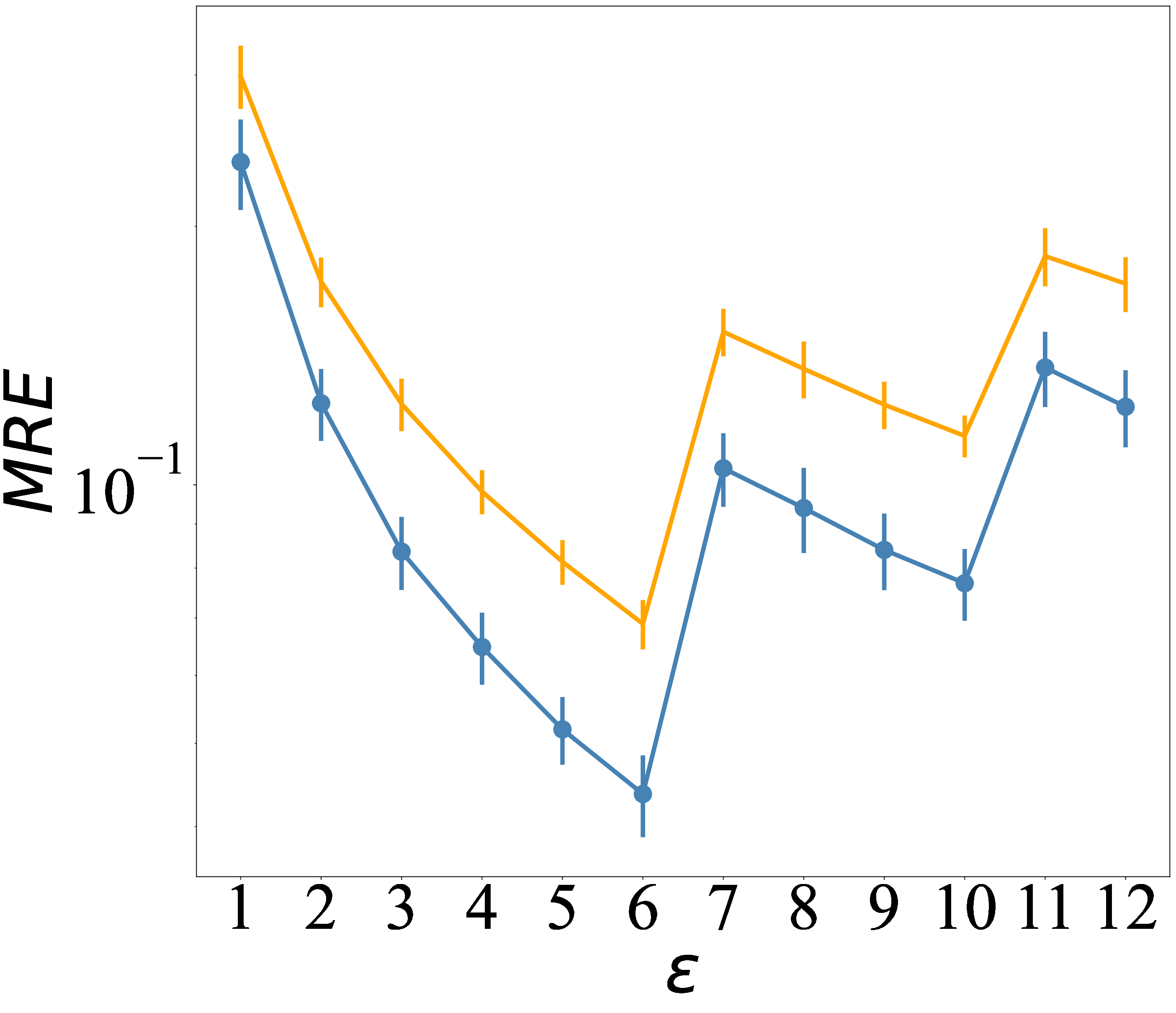}
        \caption{$h=1$ ($\varepsilon_1 = 0.1\varepsilon$)}
        \label{fig:eval:Yeast:h1}
        \vspace{4mm}
    \end{subfigure}
    \hfill
    \begin{subfigure}{0.45\columnwidth}
    \centering
        \vspace{3mm}
        \includegraphics[width=\hsize]{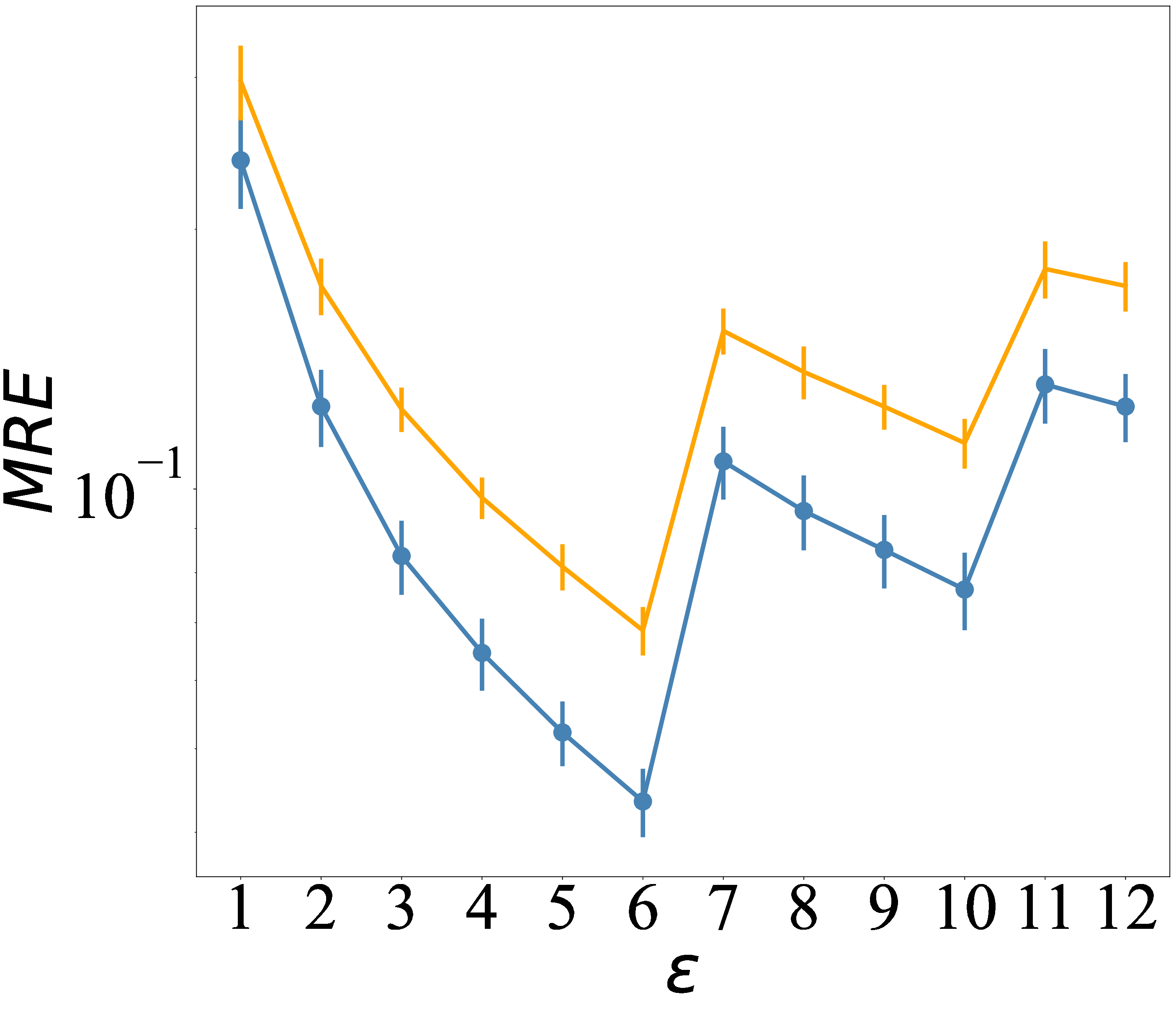}
        \caption{$h=3$ ($\varepsilon_1 = 0.1\varepsilon$)}
        \label{fig:eval:Yeast:h3}
        \vspace{4mm}
    \end{subfigure}
    \hfill
    \begin{subfigure}{0.45\columnwidth}
    \centering
        \includegraphics[width=\hsize]{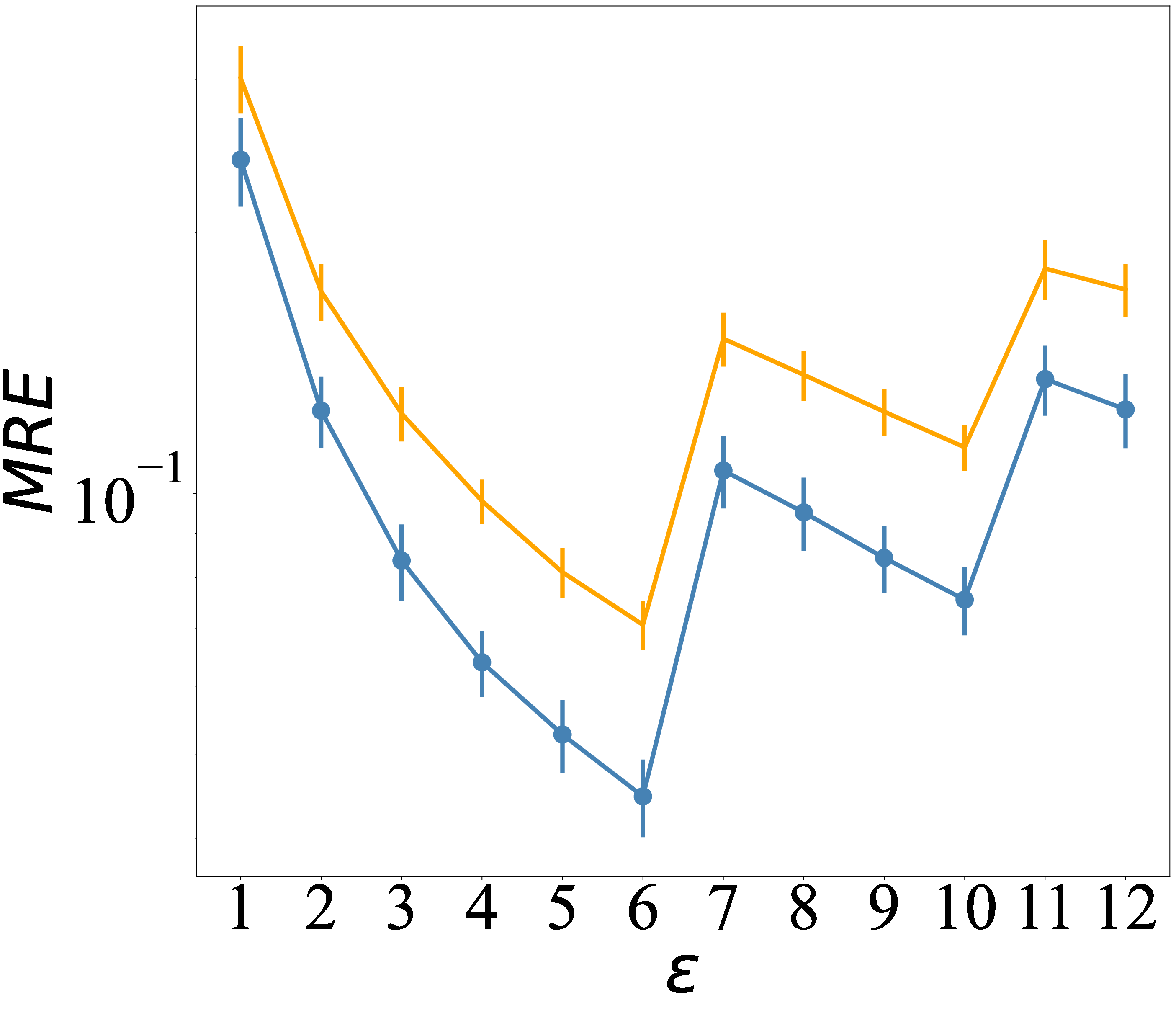}
        \caption{$h=5$ ($\varepsilon_1 = 0.1\varepsilon$)}
        \label{fig:eval:Yeast:h5}
         \vspace{4mm}
    \end{subfigure}
    \hfill
    \begin{subfigure}{0.45\columnwidth}
    \centering
        \vspace{3mm}
        \includegraphics[width=\hsize]{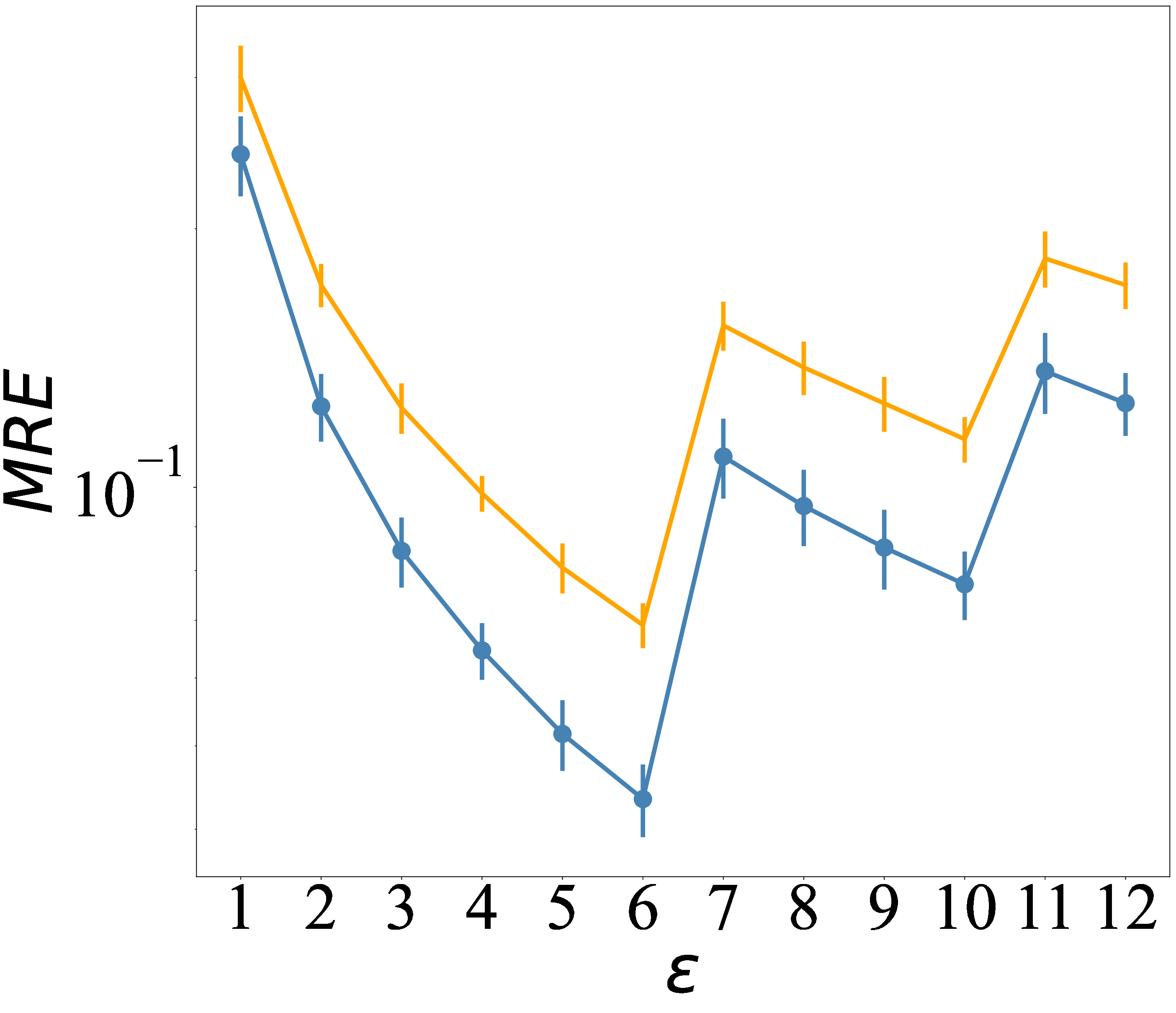}
        \caption{$h=10$ ($\varepsilon_1 = 0.1\varepsilon$)}
        \label{fig:eval:Yeast:h10}
         \vspace{4mm}
    \end{subfigure}
    \hfill
    \caption{\small{$\varepsilon$ Selection for $4$-Clique Counting on USAirport, $p = 0.1$}}
    \label{fig:eval:Yeast_var}
\end{figure}


\subsubsection{$4$-Clique Counting}

Under \ddp, we study the influence of $p$, and $\varepsilon$ over $4$-clique counting.
In Algorithm~\ref{alg:phase1:kclique}, at Line~\ref{alg:phase1:kclique_17}, the estimated local sensitivity is controlled by the $k$ value.
We use $k=3$ as the building block for $4$-clique counting, \ie, the $\lambda_4$ (\resp $\lambda_3$) returned by Algorithm~\ref{alg:phase1:kclique} for $k=4$ (\resp $k=3$) is $\lambda_4 = \frac{\lambda_3*4}{3}$.
We do not regenerate $\lambda$ for $4$-clique by Algorithm~\ref{alg:phase1:kclique}.

\vspace{1mm}
\noindent \underline{$p$ Selection}. Fig.~\ref{fig:eval:Bitcoin:p} reports the average \mre results for \pcalg and \elvalg over $100$ runs on \textit{Bitcoin} when $p$ varies from $0.1$ to $0.8$.
We show results with $h\!=\!1,3,10$ and $\varepsilon\!=\!3.0, 10.0$ for all $p$ values.
When $p < 0.8$, the \pcalg outperforms \elvalg. The average \mre of \pcalg increases when $p$ becomes larger since more noises are injected into the response.
Similar to the result for $3$-clique, when $p > 0.4$, the \mre grows slowly.
This is because, for almost all vertices, $p$ is large enough to include all their connections to the minimal \pcs.
When $p \geq 0.8$, ELVs will be subsets of the corresponding minimal \pcs.

\vspace{1mm}
\noindent \underline{$\varepsilon$ Selection}.
Fig.~\ref{fig:eval:all} reports average \mre returned by \pcalg and \elvalg on all datasets with $p = 0.1$, $h = 3$, and $\varepsilon = 10$ (\ie, $\varepsilon_1 = 1.0$).
On all datasets, our \pcalg outperforms \elvalg.
Figs.~\ref{fig:eval:Yeast:h1} to~\ref{fig:eval:Yeast:h10} report the average \mre over \pcalg and \elvalg on \textit{USAirport} when $\varepsilon$ varies from $1$ to $12$. We choose $p = 0.1$ for the minimal \pc computation.
For all $\varepsilon$, the privacy budget of Phase-$1$ is $\varepsilon_1 = 0.1\varepsilon$.
Our \pcalg shows similar trends as \elvalg under all $h$ values,
but better \mre for all $\varepsilon$ values.

%% file: background.tex
\section{Background and Related Work}
\label{sec:relwork}

We present related work from perspectives of cohesive subgraph search and graph analysis with DP.

\vspace{1mm}
\noindent \textbf{Cohesive subgraph search}.
A cohesive subgraph search aims to identify subgraphs that contain a specific query for a given scenario. Several cohesive subgraph models have been proposed. A clique is a subgraph with the highest density, where every pair of vertices is connected~\cite{DBLP:conf/www/DanischBS18,DBLP:conf/asunam/MiharaTO15}.
However, due to the restrictive nature of cliques, some relaxed models have been introduced.
These include the $k$-core model~\cite{DBLP:journals/corr/cs-DS-0310049,seidman1983network,DBLP:conf/focs/DhulipalaLRSSY22},
$k$-truss model~\cite{cohen2008trusses,DBLP:journals/ieicet/SaitoYK08},
$k$-ecc model~\cite{DBLP:journals/siamrev/Manacher89,DBLP:journals/tkde/HuWCLF17}, 
$p$-cohesive model~\cite{DBLP:journals/kais/LiZZQZL21, morris2000contagion},
and others~\cite{DBLP:conf/aaai/ChenCPWLZY21,DBLP:conf/icde/WangZLZL22,wang2023cohesive}.
{\revision
In~\cite{DBLP:journals/kais/LiZZQZL21}, minimal \pcs were detected through an expand-shrink structure and two relatively intuitive score functions. 
When including one vertex in a partial \pc, the merit score considers the contribution of a vertex to increasing the degrees of some vertices in the partial \pc.
The penalty score captures the need for extra neighbors outside the partial \pc to meet the $p$ constraint.
Using these score functions, minimal \pcs captured may not be dense; \ie, more vertices with fewer mutually connected edges are often captured.

We design new criteria (including a new merit function and a new penalty function) to effectively identify minimal \pcs with increased densities.
While adopting the generic expand-shrink framework as in~\cite{DBLP:journals/kais/LiZZQZL21}, our new score functions help find minimal \pcs with much higher density than the method developed in~\cite{DBLP:journals/kais/LiZZQZL21}.
}

\vspace{1mm}
\noindent \textbf{Graph analysis with DP}.
Many recent contributions proposed DP-based solutions for the release and/or the
analysis of graph data~\cite{pgd2021}. 
They were based either on centralized DP (CDP)~\cite{fan2013differentially, DBLP:journals/tkde/FanX14},
local DP (LDP)~\cite{DBLP:conf/sigmod/CormodeJKLSW18,DBLP:conf/ccs/QinYYKX017,DBLP:conf/uss/ImolaMC22}, DDP~\cite{DBLP:conf/ccs/SunXKYQWY19}, or relationship local DP (RLDP)~\cite{liu2022icde}.
Traditional DP techniques were primarily designed for centralized settings, where a trusted central entity collects and analyzes data. 
However, graph analysis under CDP requires the data holder to have full information about the graph.
Answering questions related to graphs under LDP requires to collect information from individuals~\cite{DBLP:journals/siamcomp/KasiviswanathanLNRS11}, which may ruin the property (lower the data utility) of the raw graph~\cite{DBLP:conf/ccs/QinYYKX017}.

To solve this problem, Haipei \etal~\cite{DBLP:conf/ccs/SunXKYQWY19} proposed privacy preservation in the context of DDP using the \elv.
Their approach ensures each participant's privacy concerning their connections in its \elv.
{\color{black}Yuhan~\etal~~\cite{liu2022icde} tried to solve the problem under the RLDP mechanism and using the \elv.}
However, these approaches solely focused on the proximity of a queried vertex and did not capture the critical connections of the vertex within the subgraph.
{\revision
While adopting the two-phase framework as in~\cite{DBLP:conf/ccs/SunXKYQWY19}, 
our algorithm is different in the sense that we identify the critical conditions of vertices, which are denser and more cohesive and often part of the ELVs. We also prioritize the critical conditions for privacy protection by dedicating the privacy budget to those connections.
As demonstrated experimentally, our approach can offer better data utility than the existing work with all connections obfuscated, \textit{e.g.},~\cite{DBLP:conf/ccs/SunXKYQWY19}.
}

%% file: conclusion.tex
\section{Conclusion}
\label{sec:conclusion}
In this paper, we investigated the problem of identifying and protecting critical connections’ information for individual participants on graphs, using minimal \pc.
New score functions were proposed to help effectively search the critical connections.
We also analytically qualified the use of the $(\epsilon, \delta)$-DDP to protect the privacy of the identified critical connections.
Extensive experiments confirmed the effectiveness of the new score functions, as well as the better trade-off between utility and privacy compared to the existing ELV-based methods.

%% file: bio.tex
\vspace{-13mm}
\begin{IEEEbiography}
[{\includegraphics[width=1in,height=1.25in,clip,keepaspectratio]{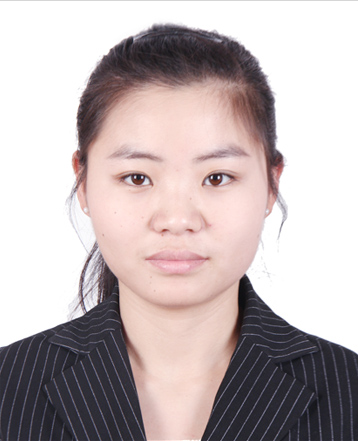}}]
{Conggai Li} is a CERC Fellow at Data61, Commonwealth Scientific and Industrial Research Organisation (CSIRO). She received her Doctoral degree from the Australian Artificial Intelligence Institute within the Faculty of Engineering and Information Technology, the University of Technology Sydney, Australia. She received her M.S. degree and B.S. degree both from Taiyuan University of Technology, China. Her research interests include graph database analytics, efficient query algorithms, and differential privacy.
\end{IEEEbiography}

\vspace{-13mm}
\begin{IEEEbiography}[{\includegraphics[width=1in,height=1.25in,clip,keepaspectratio]{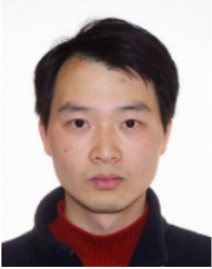}}]{Wei Ni}
[Fellow, IEEE] received the B.E. and Ph.D. degrees from Fudan University, Shanghai, China, in 2000 and 2005, respectively.
He is a Principal Research Scientist at CSIRO, Sydney, Australia, and a Conjoint Professor at the University of New South Wales. He is also an Adjunct Professor at the University of Technology Sydney and an Honorary Professor at Macquarie University. 
His research interests include machine learning, cybersecurity, network security, and privacy, as well as their applications to 6G/B6G system efficiency and integrity.
Dr. Ni has served as an Editor for IEEE Transactions on Wireless Communications since 2018, IEEE Transactions on Vehicular Technology since 2022, and IEEE Transactions on Information Forensics and Security and IEEE Communications Surveys and Tutorials since 2024. He served first as the Secretary, then the Vice-Chair and Chair of the IEEE VTS NSW Chapter from 2015 to 2022, Track Chair for VTC-Spring 2017, Track Co-chair for IEEE VTC-Spring 2016, Publication Chair for BodyNet 2015, and Student Travel Grant Chair for WPMC 2014. 
\end{IEEEbiography}

\vspace{-13mm}
\begin{IEEEbiography}
[{\includegraphics[width=1in,height=1.25in,clip,keepaspectratio]{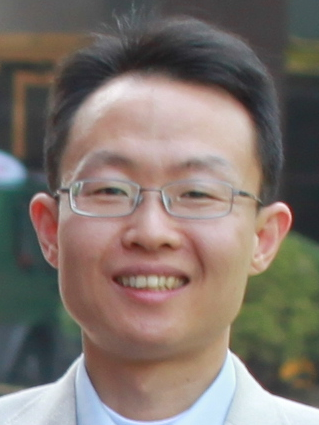}}]
{Ming Ding} [SM’17] received B.S., M.S., the Doctor of Philosophy (Ph.D.) degrees from Shanghai Jiao Tong University, China, in 2004, 2007, and 2011, respectively.
Currently, he is a Principal Research Scientist at Data61, CSIRO, in Sydney, NSW, Australia. His research interests include information technology, data privacy and security, and machine learning and AI.
He has authored more than 200 papers in IEEE journals and conferences, all in recognized venues,
around 20 3GPP standardization contributions,
as well as 
two books,
21 US patents, and has co-invented another 100+ patents on 4G/5G technologies.
He is an editor of IEEE Transactions on Wireless Communications and IEEE Communications Surveys and Tutorials.
He has served as a guest editor/co-chair/co-tutor/TPC member for multiple IEEE top-tier journals/conferences and received several awards for his research work and professional services, including the prestigious IEEE Signal Processing Society Best Paper Award in 2022.
\end{IEEEbiography}

\vspace{-13mm}
\begin{IEEEbiography}
[{\includegraphics[width=1in,height=1.25in,clip,keepaspectratio]{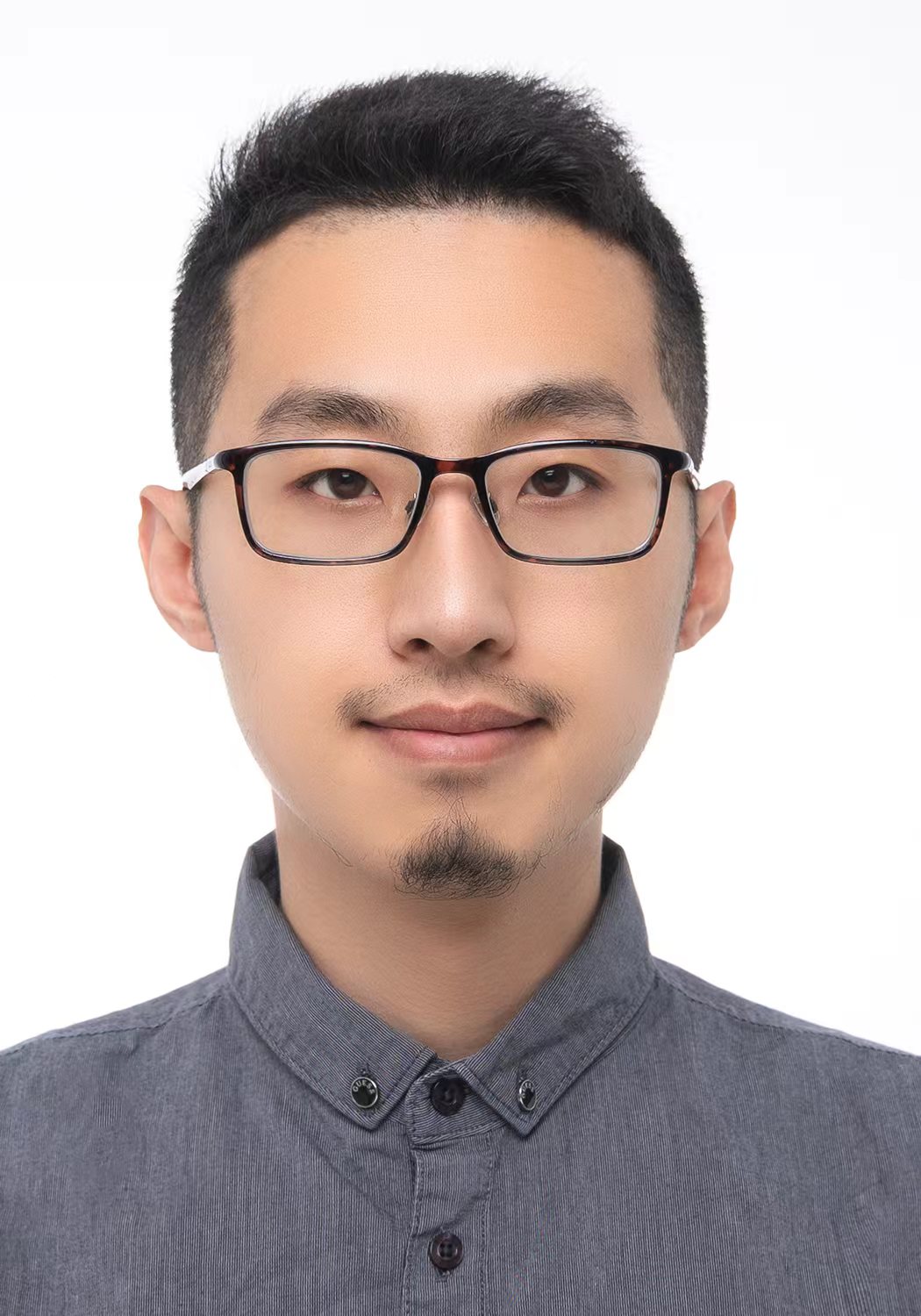}}]
{Youyang Qu} [M'19] is currently a research scientist at CSIRO, Australia.
He received his B.S. degree in 2012 and M.S. degree in 2015 from the Beijing Institute of Technology, respectively. He received his Ph.D. degree from Deakin University in 2019.
His research interests focus on Machine Learning, Big Data, IoT, blockchain, and corresponding security and customizable privacy issues.
He has over 70 publications, including high-quality journal and conference papers such as \emph{IEEE TII}, \emph{IEEE TSC}, \emph{ACM Computing Surveys}, \emph{IEEE IOTJ}, etc.
He is active in the research society and has served as an organizing committee member in SPDE 2020, BigSecuirty 2021, and Tridentcom 2021/2022.
\end{IEEEbiography}

\vspace{-13mm}
\begin{IEEEbiography}
[{\includegraphics[width=1in,height=1.25in,clip,keepaspectratio]{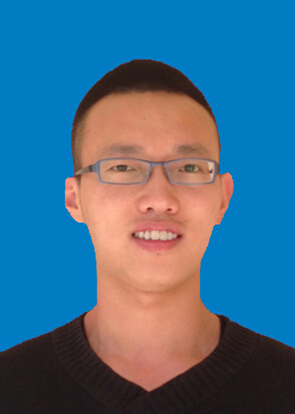}}]
{Jianjun Chen} is working toward the Ph.D. degree in the Australian Artificial Intelligence Institute within the Faculty of Engineering and Information Technology, University of Technology Sydney, Australia. He received the B.S. degree in software engineering and M.S. degree in computer science from the Taiyuan University of Technology, China. His research interests include Federated Learning, Cooperative Perception, and 3D Object Detection.
\end{IEEEbiography}

\vspace{-13mm}
\begin{IEEEbiography}
[{\includegraphics[width=1in,height=1.25in,clip,keepaspectratio]{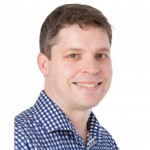}}]
{David Smith} [Member, IEEE] received the B.E. degree from the University of New South Wales, Sydney,
Australia, in 1997, and the M.E. (research) and Ph.D. degrees from the University of Technology, Sydney, Australia, in 2001 and 2004, respectively.
He is currently a Principal Research Scientist with CSIRO.
He has a variety of industry experience in electrical and telecommunications engineering. He has published over 150 technical refereed papers.
His research interests are data privacy, distributed systems privacy and edge computing data privacy, distributed machine learning, data privacy for supply chains, wireless body area networks, game theory for distributed networks, 5G/6G networks, disaster tolerant networks, and distributed optimization for smart grid.
He has made various contributions to IEEE standardization activity in personal area networks. He is an Area Editor for IET Smart Grid and has served on the technical program committees of several leading international conferences in the fields of communications and networks. He was a recipient of four conferences of best paper awards.
\end{IEEEbiography}

\vspace{-13mm}
\begin{IEEEbiography}
[{\includegraphics[width=1in,height=1.25in,clip,keepaspectratio]{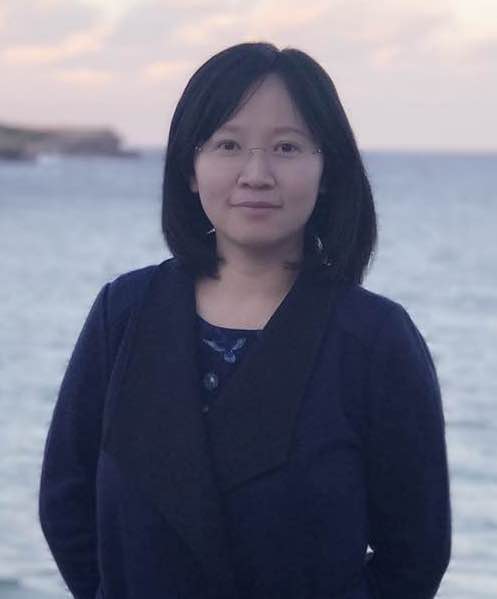}}]
{Wenjie Zhang} received the PhD degree in computer science and engineering from the University of New South Wales, in 2010. She is currently a professor and ARC Future fellow in the School of Computer Science and Engineering, the University of New South Wales, Australia.
Since 2008, she has published more than 100 papers in SIGMOD, VLDB, ICDE, TODS, the IEEE Transactions on Knowledge and Data Engineering, and the VLDB Journal.
\end{IEEEbiography}

\vspace{-13mm}
\begin{IEEEbiography}
[{\includegraphics[width=1in,height=1.25in,clip,keepaspectratio]{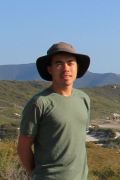}}]
{Thierry Rakotoarivelo} is the group leader of the Information Security and Privacy group at CSIRO.
His research focuses on the design and use of frameworks for privacy risk assessment, the development of privacy-enhancing technologies, and the study of trade-offs in responsible use of data in specific application domains.
He led several projects on data privacy in the Government, EdTech, and Energy sectors, where he designed and delivered novel technologies in areas such as confidentiality risk quantification or provably private synthetic data generation.
He received his PhD degree from the University of New South Wales, Australia. His thesis received the Prix Léopold Escande award from INPT in 2007. 
\end{IEEEbiography}